\documentclass[preprint]{elsarticle}

\usepackage[margin=3.8cm,bmargin=4cm,lmargin=4cm,tmargin=3.5cm]{geometry}
\usepackage{changebar}
\usepackage[normalem]{ulem}
\usepackage[utf8]{inputenc}
\usepackage{amsfonts}
\usepackage{amssymb,amsmath,amscd,latexsym}
\usepackage{gastex}
\usepackage{graphicx}
\usepackage{xcolor}
\usepackage{subfigure}
\usepackage{multirow}
\usepackage{comment}
\usepackage{array}
\newcolumntype{C}[1]{>{\centering\arraybackslash}p{#1}}
\def\abs#1{\ensuremath{\lvert #1\rvert}} 
\def\norm#1{\ensuremath{\lVert #1\rVert}}
\makeatletter
\DeclareRobustCommand\sfrac[1]{\@ifnextchar/{\@sfrac{#1}}%
                                            {\@sfrac{#1}/}}
\def\@sfrac#1/#2{\leavevmode\scalebox{.9}{\kern.1em\raise.5ex
         \hbox{$\m@th\mbox{\fontsize\sf@size\z@
                           \selectfont#1}$}\kern-.1em
         /\kern-.15em\lower.25ex
          \hbox{$\m@th\mbox{\fontsize\sf@size\z@
                            \selectfont#2}$}}}
\DeclareRobustCommand\numfrac[1]{\@ifnextchar/{\@numfrac{#1}}%
                                            {\@numfrac{#1}}}
\def\@numfrac#1{\leavevmode \hbox{$\m@th\mbox{\fontsize\sf@size\z@
                           \selectfont#1}$}}
\makeatother
\def\red{\textcolor{red}}
\def\blue{\textcolor{blue}}

\newcommand{\nat}{\mathbb N}
\newcommand{\tuple}[1]{\langle #1 \rangle}
\newcommand{\dist}{{\cal D}}

\newcommand{\q}{\hat{q}}
\newcommand{\p}{\hat{p}}

\newcommand{\M}{{\cal M}}      
\newcommand{\N}{{\cal N}}
\newcommand{\A}{{\cal A}}      
\newcommand{\B}{{\cal B}}
\newcommand{\C}{{\cal C}}      
\newcommand{\F}{{\cal F}}      
\newcommand{\LL}{{\cal L}}     
\newcommand{\Supp}{{\sf Supp}}
\newcommand{\Pref}{{\sf Pref}}

\newcommand{\Paths}{{\sf Path}}

\newcommand{\Last}{{\sf Last}}
\newcommand{\sink}{{\sf sink}}
\newcommand{\Outcomes}{\mathit{Outcome}}
\newcommand{\fsum}{\mathit{sum}}
\newcommand{\fmax}{\mathit{max}}
\newcommand{\win}[2]{\langle \! \langle 1 \rangle \! \rangle_{\mathit{#2}}^{\mathit{#1}}}
\newcommand{\winsure}[1]{\langle \! \langle 1 \rangle \! \rangle_{\mathit{sure}}^{\mathit{#1}}}
\newcommand{\winas}[1]{\langle \! \langle 1 \rangle \! \rangle_{\mathit{almost}}^{\mathit{#1}}}
\newcommand{\winlim}[1]{\langle \! \langle 1 \rangle \! \rangle_{\mathit{limit}}^{\mathit{#1}}}
\newcommand{\Bool}{{\sf B}^+}
\newcommand{\true}{{\sf true}}
\newcommand{\false}{{\sf false}}
\newcommand{\sync}{{\sf sync}}
\newcommand{\init}{{\sf init}}
\newcommand{\Act}{{\sf A}}
\newcommand{\Pre}{{\sf Pre}}
\newcommand{\post}{{\sf post}}
\newcommand{\mem}{{\sf Mem}}

\let\epsilon\varepsilon
\let\emptyset\varnothing

\newtheorem{theorem}{Theorem}
\newtheorem{lemma}{Lemma}
\newtheorem{corollary}[lemma]{Corollary}

\newtheorem{remark}[lemma]{Remark}
\newenvironment{proof}{Proof.}{}

\newenvironment{exclude}{}{}

\begin{document}

\title{{\bf The Complexity of Synchronizing \\ Markov Decision Processes}\tnoteref{t1,t2}}

\tnotetext[t1]{A preliminary version of this article appeared in the 
\emph{Proceedings of the 17th International Conference on Foundations of Software 
Science and Computation Structures} (FoSSaCS), 
Lecture Notes in Computer Science 8412, Springer, 2014, pp. 58-72, and
in the
\emph{Proceedings of the 25th International Conference on Concurrency Theory} (CONCUR),
Lecture Notes in Computer Science 8704, Springer, 2014, pp. 234-248.}

\tnotetext[t2]{
This work was partly supported by 
the Belgian Fonds National de la Recherche Scientifique (FNRS), and by 
the PICS project \emph{Quaverif} funded by the French Centre National de la Recherche Scientifique (CNRS).
}

\author[LSV]{Laurent~Doyen\corref{cor}}
\ead{doyen@lsv.fr}

\author[ULB]{Thierry Massart}
\ead{tmassart@ulb.ac.be}

\author[LSV,ULB]{Mahsa Shirmohammadi}
\ead{mahsa.shirmohammadi@gmail.com}

\cortext[cor]{Corresponding author: Laurent Doyen, LSV, CNRS UMR 8643 \& ENS Cachan, 61 avenue du Pr\'esident Wilson, 94235  Cachan Cedex, France.}

\address[LSV]{LSV, ENS Cachan \& CNRS, France}

\address[ULB]{Université Libre de Bruxelles, Belgium}

\begin{abstract}
We consider Markov decision processes (MDP) as generators of sequences of 
probability distributions over states.
A probability distribution is $p$-synchronizing if the probability
mass is at least $p$ in a single state, or in a given set of states.
We consider four temporal synchronizing modes:
a sequence of probability distributions is always $p$-synchronizing, eventually \mbox{$p$-synchronizing}, 
weakly $p$-synchronizing, or strongly $p$-synchronizing if, respectively, all, some, 
infinitely many, or all but finitely many distributions in the sequence 
are $p$-synchronizing.

For each synchronizing mode, an MDP can be 
$(i)$ \emph{sure} winning if there is a strategy that produces a $1$-synchronizing sequence; 
$(ii)$ \emph{almost-sure} winning if there is a strategy that produces 
a sequence that is, for all $\epsilon > 0$, a (1-$\epsilon$)-synchronizing sequence; 
$(iii)$ \emph{limit-sure} winning if for all $\epsilon > 0$, 
there is a strategy that produces a (1-$\epsilon$)-synchronizing sequence.

We provide fundamental results on the expressiveness, decidability, and complexity
of synchronizing properties for MDPs.
For each synchronizing mode, we consider the problem of deciding whether 
an MDP is sure, almost-sure, or limit-sure winning, 
and we establish matching upper and lower complexity bounds of the problems: 
for all winning modes, we show that the problems are PSPACE-complete
for eventually and weakly synchronizing, and PTIME-complete for always and strongly synchronizing.
We establish the memory requirement for winning strategies, and
we show that all winning modes coincide for always 
synchronizing, and that the almost-sure and limit-sure winning modes coincide for
weakly and strongly synchronizing.
\end{abstract}

\maketitle


\section{Introduction}\label{sec:intro}
Markov decision processes (MDP) are finite-state stochastic models
of dynamic systems studied in many applications such as planning~\cite{MHC03}, 
randomized algorithms~\cite{AspnesH90,PSL00}, 
communication protocols~\cite{FokkinkP06},
and in many problems related to reactive system design and verification~\cite{FV97,BBS06,EKVY08}.
MDPs exhibit both stochastic and nondeterministic behavior, as
in the control problem for reactive systems: nondeterminism
represents the possible choice of actions of the controller, and
stochasticity represents the uncertainties about the system response (see \figurename~\ref{fig:almost-limit-eventually-differ}).
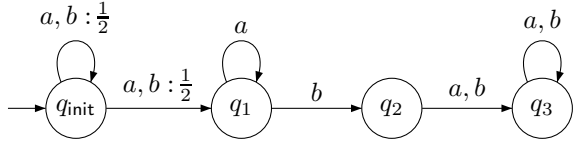
\begin{figure}[t]
\begin{center}
    \begin{picture}(80,18)(0,0)

\node[Nmarks=i,iangle=180](n0)(10,4){$q_{\init}$}
\node[Nmarks=n](n1)(32,4){$q_1$}
\node[Nmarks=n](n2)(52,4){$q_2$}
\node[Nmarks=n](n3)(72,4){$q_3$}

\drawedge(n0,n1){$a,b:\!\frac{1}{2}$}
\drawloop[ELside=l,loopCW=y, loopangle=90, loopdiam=5](n0){$a,b:\!\frac{1}{2}$}

\drawedge(n1,n2){$b$}
\drawloop[ELside=l,loopCW=y, loopangle=90, loopdiam=5](n1){$a$}

\drawedge(n2,n3){$a,b$}
\drawloop[ELside=l,loopCW=y, loopangle=90, loopdiam=5](n3){$a,b$}

\end{picture}
\end{center}
 \caption{An MDP with four states and set of actions $\{a,b\}$.
All transitions are deterministic except from $q_{\init}$ where on all actions, 
the successors are $q_{\init}$ and $q_1$ with probability $\frac{1}{2}$. 
An initial Dirac distribution (that assigns probability~$1$ to $q_{\init}$)
is depicted by the incoming arrow in $q_{\init}$.
%
\label{fig:almost-limit-eventually-differ}}
\end{figure}
%
The controller synthesis problem is to compute the largest probability with which
a control strategy  can ensure  that the system satisfies a given specification,
and to construct an optimal strategy~\cite{BA95,FV97}. The qualitative variant of the problem
is to decide if the system can satisfy the specification with probability~$1$.
Fundamental well-studied specifications are state-based and describe correct behaviors 
as infinite sequences of states of the MDP, including safety and liveness properties 
such as reachability, B\"uchi, and co-B\"uchi conditions, which require the system to 
visit a set of target states once, infinitely often, and ultimately always, 
respectively~\cite{automata,AHK07}.

In contrast to this traditional approach, we consider a distribution-based semantics
where the specification describes correct behaviors of MDPs as infinite sequences 
of probability distributions $d_i : Q \to [0,1]$ over the finite state space~$Q$ of the system, 
where $d_i(q)$ is the probability that the MDP is in state $q \in Q$ after $i$ 
execution steps. 
The distribution-based semantics is adequate in large-population models, 
such as systems biology~\cite{HMW09}, robot planning~\cite{BBMR08}, distributed systems~\cite{GGB12}, 
etc.~where the system consists of several copies of the same process (molecules, robots, sensors, etc.),
 and the relevant information along the execution of the
system is the number of processes in each state, or the relative frequency 
(i.e., the probability) of each state.
In the context of several identical processes, the same control strategy is 
used in every process, but the internal state of each process need not be the
same along the execution, since probabilistic transitions may have a
different outcome in each process. Therefore, the global execution of the
system (consisting of all the processes) is better described by the sequence 
of probability distributions over states along the execution. 
However, the control strategy is local to each process and can select
control actions depending on the full history of the process execution,
which corresponds to general perfect-information strategies that we consider in 
this work.


Previously, the special case of
blind strategies has been considered, which in each step select the
same control action at all states, and thus only depend on the number
of execution steps of the system. In automata theory, a blind strategy 
corresponds simply to an input word.  In MDPs with blind strategies, 
also known as probabilistic automata~\cite{Rabin63,PAZBook},
several basic problems are undecidable such as deciding if there exists a blind strategy that ensures
a coB\"uchi condition with probability~$1$~\cite{BGB12}, 
or deciding if a reachability condition can be ensured with probability
arbitrarily close to~$1$~\cite{GO10}.

The main contribution of this article is to establish the decidability and optimal complexity of deciding
\emph{synchronizing} properties for the distribution-based semantics of MDPs under
general strategies.  Synchronizing properties require that the sequence of
probability distributions accumulate all the probability mass
in a single state, or in a given set of states. They generalize
synchronizing properties of finite automata~\cite{Volkov08,DMS11b}.
Formally, for $0  \leq p \leq 1$ let a probability distribution $d : Q \to [0,1]$ 
be $p$-synchronized if it assigns probability at least~$p$ to some state. 
A sequence $\bar{d} = d_0 d_1 \dots$ of probability distributions is
\begin{itemize}
\item[$(a)$] \emph{always $p$-synchronizing} if $d_i$ is $p$-synchronized for all~$i$;
\item[$(b)$] \emph{eventually $p$-synchronizing} if $d_i$ is $p$-synchronized for some~$i$;
\item[$(c)$] \emph{weakly $p$-synchronizing} if $d_i$ is $p$-synchronized for infinitely many~$i$'s;
\item[$(d)$] \emph{strongly $p$-synchronizing} if $d_i$ is $p$-synchronized for all but finitely many~$i$'s.
\end{itemize}

We present a consistent and comprehensive theory of the qualitative synchronizing 
properties, corresponding to the case where either $p=1$, or $p$ tends to $1$, 
which are analogous to the traditional safety, 
reachability, B\"uchi, and coB\"uchi conditions~\cite{ConcOmRegGames}.

\begin{table}[!t]
\begin{center}
\caption{
Winning modes for always, strongly, weakly, and eventually synchronizing objectives 
(where $\M^{\alpha}_n(T)$ denotes the probability that under strategy $\alpha$,
after $n$ steps the MDP $\M$ is in a state of $T$). \label{tab:def-modes}}{
\begin{tabular}{|l@{\ } |*{2}{c@{\;}c@{\;}c@{\;}|}}
\hline        
 \large{\strut}         & \multicolumn{3}{c|}{Always} & \multicolumn{3}{c|}{Strongly} \\
\hline
 Sure \large{\strut}    & $\exists \alpha$ & $\forall n$ & $\M^{\alpha}_n(T)=1$ 
        	        & $\exists \alpha$ & $\exists N \: \forall n \geq N$   & $\M^{\alpha}_n(T)=1$ 
 \\
\hline
 Almost-sure \  \large{\strut} 	& $\exists \alpha$ & $\inf_{n}$  & $\M^{\alpha}_n(T)=1$
				& $\exists \alpha$& $\liminf_{n\to\infty}$     & $\M^{\alpha}_n(T)=1$ 
 \\
\hline
 Limit-sure \large{\strut}	& $\sup_{\alpha}$ & $\inf_{n}$   & $\M^{\alpha}_n(T)=1$
   				& $\sup_{\alpha}$& $\liminf_{n\to\infty}$    & $\M^{\alpha}_n(T)=1$  
 \\
\hline
\hline                        
 \large{\strut}   &\multicolumn{3}{c|}{Weakly} & \multicolumn{3}{c|}{Eventually} \\
\hline
 Sure \large{\strut}	  
			& $\exists \alpha$ & $\forall N \: \exists n \geq N$   & $\M^{\alpha}_n(T)=1$ 
			& $\exists \alpha$ & $\exists n$ & $\M^{\alpha}_n(T)=1$ 
 \\
\hline
 Almost-sure \  \large{\strut}	  
				& $\exists \alpha$& $\limsup_{n\to\infty}$     & $\M^{\alpha}_n(T)=1$ 
				& $\exists \alpha$ & $\sup_{n}$  & $\M^{\alpha}_n(T)=1$ 
 \\
\hline
 Limit-sure \large{\strut}      
		& $\sup_{\alpha}$& $\limsup_{n\to\infty}$    & $\M^{\alpha}_n(T)=1$  
		& $\sup_{\alpha}$ & $\sup_{n}$   & $\M^{\alpha}_n(T)=1$ 
 \\
\hline
\end{tabular}  
}
\end{center}
\end{table}

\paragraph{Applications}
A typical application scenario of synchronizing properties is the design of a 
control program for 
a group of mobile robots running in a stochastic environment~\cite{MSG10}. 
The possible behaviors of the robots and the stochastic response of 
the environment (such as obstacle encounters) are represented by an MDP,
and a synchronizing strategy corresponds to a control program that
can be embedded in every robot to ensure that they meet (or synchronize)
all the time, eventually once, infinitely often, or eventually forever.

Synchronization properties are central in large-population models in biology, 
such as yeast, where experimental synchronization methods have been developed 
to get a population of yeast in the same cell cycle stage~\cite{BDGG17, Futcher99}.
A simple abstraction of large populations of identical finite-state stochastic 
agents is to consider a continuum of agents, described by the relative fraction 
of agents in each possible state, i.e. by a distribution. For example, consider
a model of cells where at each time instant half of the cells get activated,
and once activated we can block them for a while, or release them to reach a 
fluorescent state. In the MDP of \figurename~\ref{fig:almost-limit-eventually-differ},
the state $q_{1}$ corresponds to activation, action $a$ is blocking, and action $b$ is releasing.
The fluorescent state is $q_2$. Probability mass arbitrarily close to $1$ can be 
accumulated in $q_2$, thus for all $\epsilon > 0$ we can generate an eventually $(1-\epsilon)$-synchronizing
sequence in $q_2$, but not an eventually $1$-synchronizing sequence. 
If it was possible to reset the cell state after fluorescence, such as in the MDP 
of \figurename~\ref{fig:inf-mem}, then we can obtain a sequence of distribution
that is weakly $(1-\epsilon)$-synchronizing, for all $\epsilon > 0$.
\medskip

We consider the following qualitative winning modes, summarized in 
Table~\ref{tab:def-modes}: 
$(i)$ \emph{sure} winning, if there is a strategy that generates
an \{always, eventually, weakly, strongly\} $1$-synchronizing sequence;
$(ii)$ \emph{almost-sure}
winning, if there is a strategy that generates a sequence that is, for all $\epsilon>0$,
\{always, eventually, weakly, strongly\} $(1-\epsilon)$-synchronizing;
$(iii)$ \emph{limit-sure} winning, if for all $\epsilon>0$, there is a strategy
that generates an \{always, eventually, weakly, strongly\} $(1-\epsilon)$-synchronizing sequence.

\smallskip
\noindent \emph{Contribution.} The contributions of this article are summarized as follows: 
\begin{itemize}
\item \emph{Expressiveness.} We show that the three winning modes form a strict hierarchy for
eventually synchronizing: there are limit-sure winning
MDPs that are not almost-sure winning, and 
there are almost-sure winning MDPs that are not sure winning. 
This is in contrast with the traditional state-based reachability objectives 
for which the notions of almost-sure and limit-sure winning coincide in MDPs. 
In this context, a more unexpected and difficult result is that the almost-sure and 
limit-sure modes coincide for weakly and strongly synchronizing. Thus those
two synchronizing modes are more robust than eventually synchronizing, although
we show that almost-sure weakly synchronizing strategies can be constructed from 
the analysis of eventually synchronizing (in limit-sure winning mode).
Finally, for always synchronizing the three winning modes coincide, and we show that 
they coincide with a traditional safety objective.\smallskip

\item \emph{Complexity.} 
For each synchronizing and winning mode, we consider the problem of deciding 
if a given initial distribution is winning. The complexity results are shown 
in Table~\ref{tab:complexity} (p.~\pageref{tab:complexity}).
We establish the decidability and optimal complexity bounds for all
winning modes.  Under general strategies, the decision problems have
much lower complexity than with blind strategies. We show that all
decision problems are decidable, in polynomial time for always and strongly 
synchronizing, and PSPACE-complete for eventually and weakly synchronizing.  
This is also in contrast with almost-sure winning in the traditional
semantics of MDPs, which is solvable in polynomial time for both
safety and reachability~\cite{deAlfaro97}.
\smallskip

\item \emph{Memory bounds.} We complete the picture by proving optimal memory bounds for 
winning strategies, summarized in Table~\ref{tab:memory} (p.~\pageref{tab:memory}).
Memoryless strategies are sufficient for always synchronizing (like
for safety objectives). 
We show that linear memory is sufficient for strongly synchronizing,
and we identify a variant of strongly synchronizing for which memoryless strategies 
are sufficient.
For eventually and weakly synchronizing, exponential memory is sufficient 
and may be necessary for sure winning strategies, and in general infinite memory 
is necessary for almost-sure winning. 
\end{itemize}

Some results in this article rely on insights about  games and 
alternating automata that are of independent interest. 
Firstly, the sure-winning problem for eventually synchronizing is equivalent
to a two-player game with a synchronized reachability objective, where the goal 
for the first player is to ensure that a target state is reached 
after a number of steps that is independent of the strategy of the opponent 
(and thus this number can be fixed in 
advance by the first player). This condition is stronger than plain reachability,
and while the winner in two-player reachability games can be decided in polynomial time,
deciding the winner for synchronized reachability is PSPACE-complete. This result
is obtained by turning the synchronized reachability game into a one-letter alternating
automaton for which the emptiness problem (i.e., deciding if there exists a word accepted
by the automaton) is PSPACE-complete~\cite{Holzer95,AFA1}.
Secondly, our PSPACE lower bound for the limit-sure winning problem in eventually synchronizing
uses a PSPACE-completeness result that we establish for the \emph{universal finiteness problem},
which is to decide, given a one-letter alternating automata, whether from every 
state the accepted language is finite.

\paragraph{Related Works}
The traditional state-based semantics of MDPs
has been studied extensively~\cite{Puterman,CY95,FV97} and 
plays a central role in recent 
developments of system verification and controller synthesis,
including expressiveness and complexity analysis of various classes 
of properties~\cite{FKNP11}, using techniques such as symbolic 
algorithms for B\"uchi objectives~\cite{CHJS11}, game-based abstraction 
techniques~\cite{KKNP10}, and multi-objective analysis for assume-guarantee 
model-checking~\cite{EKVY08}.

On the other hand, the distribution-based 
semantics has received a greater interest only recently,
as it is shown that relevant key properties of MDPs can only be 
expressed in a distribution-based logical framework~\cite{BRS02,KVAK10} 
and that a new useful notion of probabilistic 
bisimulation can be obtained in the distribution-based semantics~\cite{HKK14}.
Several recent works have investigated this new approach showing that 
the verification of quantitative properties of the distribution-based 
semantics is undecidable~\cite{KVAK10,FKS16}, and decidability can be obtained 
for special subclasses of systems~\cite{CKVAK11}, or through approximations~\cite{AAGT12}.
In this context, a challenging goal is to identify useful decidable properties for 
the distribution-based semantics.

Synchronization problems were first considered for deterministic finite automata (DFA)
where a \emph{synchronizing word} is a finite sequence of control actions that can be 
executed from any state of an automaton and leads to the same state 
(see~\cite{Volkov08} for a survey of results and applications). 
While the existence of a synchronizing word can be decided in NLOGSPACE 
for DFA, extensive research effort is devoted to establishing a tight bound
on the length of the shortest synchronizing word~\cite{Ber16,Nic16,Szy18}, which is
conjectured to be $(n-1)^2$ for automata with $n$ states~\cite{Cer64}.
Various extensions of the notion of synchronizing word have been proposed 
for non-deterministic and probabilistic automata~\cite{Burkhard76a,IS99,Kfo70,Shi14},
leading to results of PSPACE-completeness~\cite{Martyugin14}, or even 
undecidability~\cite{Kfo70}. 

For probabilistic systems, it is natural to consider infinite input words (i.e., blind strategies)
in order to study synchronization at the limit. 
In particular, almost-sure weakly and strongly synchronizing
with blind strategies has been studied~\cite{DMS11a} and the main result is that
the problem of deciding the existence of a blind almost-sure winning strategy 
is undecidable for weakly synchronizing, and PSPACE-complete for strongly 
synchronizing~\cite{DMS11b,DMS11Err}.
In contrast in this article, for general strategies, we establish the PSPACE-completeness and 
PTIME-completeness of deciding almost-sure weakly and strongly
synchronizing respectively.

Synchronization has been studied recently in a variety of classical computation models that extend
finite automata, such as timed automata~\cite{DJLMS14}, weighted automata~\cite{Iva14,DJLMS14},
visibly pushdown automata~\cite{CMS16}, and register automata~\cite{BQS16}. Automata with partial observability
have been considered to model systems that disclose information along their execution,
which can help the synchronizing strategy~\cite{LLS14,KLLS15}. An elegant extension 
of the computation tree logic (CTL) has been proposed to express synchronizing properties~\cite{CD16}.

\section{Markov Decision Processes and Synchronizing Properties}\label{sec:def}

A \emph{probability distribution} over a finite set~$S$ is a
function $d : S \to [0, 1]$ such that $\sum_{s \in S} d(s)= 1$. 
The \emph{support} of~$d$ is the set $\Supp(d) = \{s \in S \mid d(s) > 0\}$. 
We denote by $\dist(S)$ the set of all probability distributions over~$S$. 
Given a set $T\subseteq S$, let 
$$d(T)=\sum_{s\in T}d(s) \quad\text{ and }\quad \norm{d}_T = \fmax_{s\in T} d(s).$$

For $T \neq \emptyset$, the \emph{uniform distribution} on $T$ assigns probability 
$\frac{1}{\abs{T}}$ to every state in $T$.
Given $s \in S$, we denote by $\xi^s$ the \emph{Dirac distribution} on~$s$ that 
assigns probability~$1$ to~$s$.  

A \emph{Markov decision process} (MDP) is a tuple $\M = \tuple{Q,\Act,\delta}$ 
where $Q$ is a finite set of states, $\Act$ is a finite set of actions, 
and $\delta: Q \times \Act \to \dist(Q)$ is a probabilistic transition function.   
A state $q$ is \emph{absorbing} if $\delta(q,a)$ is the Dirac distribution on $q$
for all actions $a \in \Act$. 
Given state $q\in Q$ and action $a \in \Act$, the successor state of $q$
under action $a$ is $q'$ with probability $\delta(q,a)(q')$.
Denote by $\post(q,a)$ the set $\Supp(\delta(q,a))$, 
and given $T \subseteq Q$ let 
$$\Pre(T)= \{q \in Q \mid \exists a \in \Act: \post(q, a) \subseteq T\}$$
be the set of states from which there is an action to ensure that the
successor state is in~$T$.  For $k>0$, let $\Pre^k(T) =
\Pre(\Pre^{k-1}(T))$ with $\Pre^0(T)=T$.

Note that the sequence $\Pre^k(T)$ of iterated predecessors is ultimately periodic,
precisely there exist $k < k' < 2^{\abs{Q}}$ such that $\Pre^k(T) = \Pre^{k'}(T)$.

A \emph{path} in $\M$ is an infinite sequence $\pi = q_0 a_0 q_1 a_1 \dots$
such that $q_{i+1} \in \post(q_{i},a_{i})$ 
for all $i \geq 0$. A finite prefix $\rho = q_0 a_0 q_1 a_1 \dots q_n$
of a path (or simply
a finite path) has length $\abs{\rho}=n$ and last state $\Last(\rho)=q_{n}$.
We denote by $\Paths(\M)$ and $\Pref(\M)$ the set of all
paths and finite paths in~$\M$ respectively.

For the decision problems considered in this article, only the support
of the probability distributions in the transition function is
relevant (i.e., the exact value of the positive probabilities does not
matter); therefore, we can assume that MDPs are encoded as $\Act$-labelled
transition systems $(Q,R)$ with $R \subseteq Q \times \Act \times Q$
such that $(q,a,q') \in R$ is a transition if $q' \in \post(q,a)$.

\paragraph{Strategies}
A \textit{randomized strategy} for $\M$ (or simply a strategy) is a function 
$$\alpha: \Pref(\M) \to \dist(\Act)$$
 that, given a finite path $\rho$,
returns a probability distribution $\alpha(\rho)$ over the action set,
used to select a successor state $q'$ of $\rho$ with probability 
$\sum_{a \in \Act} \alpha(\rho)(a) \cdot \delta(q,a)(q')$ where $q=\Last(\rho)$.

A strategy $\alpha$ is
\emph{pure} if for all $\rho \in \Pref(\M)$, 
there exists an action $a \in \Act$ such that $\alpha(\rho)(a)=1$; and 
\emph{memoryless} if $\alpha(\rho) = \alpha(\rho')$
for all $\rho, \rho'$ such that $\Last(\rho) = \Last(\rho')$.
We view pure strategies as functions $\alpha: \Pref(\M)\to \Act$, and 
memoryless strategies as functions $\alpha: Q \to \dist(\Act)$.

Finally, a strategy $\alpha$ uses \emph{finite-memory} if it can be represented 
by a finite-state transducer $T = \tuple{\mem,m_0, \alpha_u, \alpha_n}$ where
$\mem$ is a finite set of modes (the memory of the strategy), 
$m_0 \in \mem$ is the initial mode,
$\alpha_u: \mem \times (\Act \times Q) \to \mem$ is an update function 
that, given the current memory, last action, and state updates the memory,
and $\alpha_n: \mem \times Q \to \dist(\Act)$ is a next-move function
that selects the probability distribution $\alpha_n(m,q)$ over
actions when the current mode is $m$ and the current state of $\M$ is $q$.
For pure strategies, we assume that $\alpha_n: \mem \times Q \to \Act$.
The \emph{memory size} of the strategy is the number $\abs{\mem}$ of modes.
For a finite-memory strategy $\alpha$, let $\M(\alpha)$ be the Markov chain
obtained as the product of $\M$ with the transducer defining $\alpha$.

\subsection{State-based semantics}
In the traditional state-based semantics, 
given an initial distribution $d_0 \in \dist(Q)$ and a strategy $\alpha$ in an MDP $\M$,
a \emph{path-outcome} is a path
$\pi = q_0  a_0 q_1 a_1 \dots$ in $\M$
such that $q_0 \in \Supp(d_0)$ and $a_i \in \Supp(\alpha(q_0 a_0 \dots q_i))$ 
for all $i \geq 0$. The probability of a finite prefix 
$\rho = q_0 a_0 q_1 a_1 \dots q_n$ of $\pi$ is 
$$d_0(q_0) \cdot
\prod_{j=0}^{n-1} \alpha(q_0 a_0\dots q_j)(a_j) \cdot
\delta(q_j,a_j)(q_{j+1}).$$
We denote by $\Outcomes(d_0,\alpha)$ the set of all path-outcomes from $d_0$ under strategy $\alpha$.
An \emph{event} $\Omega \subseteq \Paths(\M)$ is a measurable set of paths, and given an initial distribution 
$d_0$ and a strategy $\alpha$, the probability $Pr^{\alpha}(\Omega)$ of
$\Omega$ is uniquely defined~\cite{Vardi-focs85}.
We consider the following classical winning modes. 
Given an initial distribution $d_0$ and an event $\Omega$, we say that $\M$ is:
\emph{sure winning} if there exists a strategy $\alpha$ such that $\Outcomes(d_0,\alpha) \subseteq \Omega$;
\emph{almost-sure winning} if there exists a strategy $\alpha$ such that $\Pr^{\alpha}(\Omega) = 1$;
\emph{limit-sure winning} if $\sup_{\alpha} \Pr^{\alpha}(\Omega) = 1$, that is 
the event $\Omega$ can be realized with probability arbitrarily close to~$1$.
Given a set $T \subseteq Q$ of target states, and $k \in \nat$, 
we define the following events: 
\begin{itemize}
\item $\Box T = \{q_0 a_0 q_1 \dots \in \Paths(\M) \mid \forall i: q_i \in T\}$ 
the safety event of always staying in~$T$; 
\item $\Diamond T = \{q_0 a_0 q_1 \dots \in \Paths(\M) \mid \exists i: q_i \in T\}$
the event of reaching~$T$; 
\item $\Diamond^{k}\, T = \{q_0 a_0 q_1 \dots \in \Paths(\M) \mid q_k \in T \}$ the event
of reaching~$T$ after exactly $k$~steps;
\item $\Diamond^{\leq k}\, T = \bigcup_{j \leq k} \Diamond^{j}\, T$ the event
of reaching~$T$ within at most $k$~steps.
\end{itemize}

\noindent For example, if $\Pr^{\alpha}(\Diamond T)=1$ then almost-surely  a state in~$T$ is reached 
under strategy $\alpha$.

It is known for reachability objectives $\Diamond T$, that an MDP is  
almost-sure winning if and only if it is limit-sure winning,
and the set of initial distributions for which an MDP 
is sure (resp., almost-sure or limit-sure) winning can 
be computed in polynomial time~\cite{deAlfaro97}.

\subsection{Distribution-based semantics}\label{sec:distribution-based-semantice}
In contrast to the state-based semantics, we consider a symbolic outcome of MDPs 
viewed as generators of sequences of probability distributions over states~\cite{KVAK10}.
Given an initial distribution $d_0 \in \dist(Q)$ and a strategy $\alpha$ in $\M$,
the \emph{symbolic outcome} of $\M$ from $d_0$ is 
the sequence $(\M^{\alpha}_n)_{n\in \nat}$ of probability distributions defined by 
$\M^{\alpha}_k(q) = Pr^{\alpha}(\Diamond^{k}\, \{q\})$ for all $k \geq 0$ and $q \in Q$.
Hence, $\M^{\alpha}_k$ is the probability distribution over states after $k$ steps
under strategy $\alpha$.
Note that $\M^{\alpha}_0 = d_0$ and the symbolic outcome is a deterministic sequence 
of distributions: each distribution $\M^{\alpha}_k$ has a unique (deterministic) 
successor.

Informally, synchronizing objectives require that the probability of some state 
(or some group of states) tends to $1$ in the sequence $(\M^{\alpha}_n)_{n\in \nat}$,
either always, once, infinitely often, or always after some point. 
Given a set $T \subseteq Q$, consider the functions 
$$
\begin{array}{l}
\fsum_T: \dist(Q) \to [0,1] \ \text{ defined by } \fsum_T(d) = \sum_{q \in T} d(q), \text{ and} \\[3pt]
\fmax_T: \dist(Q) \to [0,1] \ \text{ defined by } \fmax_T(d) = \max_{q \in T} d(q).\\
\end{array}
$$
For $f \in \{\fsum_T,\fmax_T\}$ and $p \in [0,1]$, 
we say that a probability distribution $d$ is $p$-synchronized according to $f$ if $f(d) \geq p$,
and that a sequence $\bar{d} = d_0 d_1 \dots$ of probability distributions is:
\begin{itemize}
\item[$(a)$] \emph{always $p$-synchronizing} if $d_i$ is $p$-synchronized for all $i \geq 0$;
\item[$(b)$] \emph{event} (or \emph{eventually}) \emph{$p$-synchronizing} if $d_i$ is $p$-synchronized for some $i \geq 0$;
\item[$(c)$] \emph{weakly $p$-synchronizing} if $d_i$ is $p$-synchronized for infinitely many $i$'s;
\item[$(d)$] \emph{strongly $p$-synchronizing} if $d_i$ is $p$-synchronized for all but finitely many $i$'s.
\end{itemize}

For $p=1$, these definitions are analogous to the traditional safety, reachability, 
B\"uchi, and coB\"uchi conditions~\cite{ConcOmRegGames}. In this article, we consider the following 
winning modes where either $p=1$, or $p$ tends to~$1$ (we do not consider $p < 1$, see the
discussion in Section~\ref{sec:conclusion}). 
Given an initial distribution $d_0$ and a function $f \in \{\fsum_T,\fmax_T\}$, 
we say that for the objective of \{always, eventually, weakly, strongly\} synchronizing
from~$d_0$, the MDP $\M$ is:
\begin{itemize}
\item \emph{sure winning} if there exists a strategy $\alpha$ such that
the symbolic outcome of $\alpha$ from $d_0$ 
is \{always, eventually, weakly, strongly\} $1$-synchronizing according to~$f$;
\item \emph{almost-sure winning} if there exists a strategy $\alpha$ such that 
for all $\epsilon>0$ the symbolic outcome of $\alpha$ from $d_0$ is 
\{always, eventually, weakly, strongly\} $(1-\epsilon)$-synchronizing according to $f$;
\item \emph{limit-sure winning} if for all $\epsilon>0$, there exists a strategy $\alpha$
such that the symbolic outcome of $\alpha$ from $d_0$ is 
\{always, eventually, weakly, strongly\} $(1-\epsilon)$-synchronizing according to $f$;
\end{itemize}

Note that the winning modes for synchronizing objectives differ
from the traditional winning modes in MDPs: 
synchronizing objectives specify sequences of distributions,
in a deterministic transition system 
with infinite state space (the states are the probability distributions).
Since the transitions are deterministic and the probabilities are embedded in the state space,
the behavior of the system is non-stochastic and the specification is simply a set of sequences
(of distributions).
In contrast, the traditional almost-sure and limit-sure winning modes of MDPs
specify probability measures over sequences of states (called paths) 
in a probabilistic system with finite state space. 
Since the probabilities influence the transitions, the behavior of the system is stochastic 
and the specification is a set of probability measures over paths.
For instance almost-sure reachability requires that the probability measure 
of \emph{all} paths that visit a target state is $1$, while almost-sure eventually 
synchronizing requires that the \emph{single} symbolic outcome belongs
to the set of sequences of distributions that are $(1-\epsilon)$-synchronizing
for all $\epsilon>0$.

We often write $\norm{d}_T$ instead
of $\fmax_T(d)$ (and we omit the subscript when $T=Q$) 
and $d(T)$ instead of $\fsum_T(d)$, as 
in Table~\ref{tab:def-modes} 
where the definitions of the various winning modes and synchronizing objectives
for $f=\fsum_T$ are summarized. 


\subsection{Membership problem}

For $f \in \{\fsum_T,\fmax_T\}$ and $\lambda \in \{always, event, weakly, strongly\}$, 
the \emph{winning region} $\winsure{\lambda}(f)$ is the set of initial distributions such that
$\M$ is sure winning for $\lambda$-synchronizing (we assume that $\M$ is clear from the context). 
We define analogously the sets $\winas{\lambda}(f)$ and $\winlim{\lambda}(f)$
of almost-sure and limit-sure winning distributions.

By an abuse of notation, if a Dirac distribution $\xi^q$ belongs to $\win{\lambda}{\mu}(f)$,
we often write $q \in \win{\lambda}{\mu}(f)$ instead of $\xi^q \in \win{\lambda}{\mu}(f)$.  
For a singleton $T = \{q\}$ we have $\fsum_{T} = \fmax_{T}$,
and we simply write $\win{\lambda}{\mu}(q)$ (where $\mu \in \{sure, almost, limit\}$).
We are interested in the algorithmic complexity of the \emph{membership problem}, 
which is to decide, given a probability distribution $d_0$ and a function~$f$, 
whether $d_0 \in \win{\lambda}{\mu}(f)$.

We show that the winning region is identical for always synchronizing 
in the three winning modes (Lemma~\ref{lem:always}), whereas for eventually synchronizing, 
the winning regions of the three winning modes are in general different (Lemma~\ref{lem:dif-in-def}). 

\begin{remark}\label{rem:expressiveness}
First, note that it follows from the definitions that for all $f \in \{\fsum_T,\fmax_T\}$, 
for all $\lambda \in \{always, event, weakly, strongly\}$, and all $\mu \in \{sure, almost, limit\}$:
\begin{itemize}
\item $\win{always}{\mu}(f) \subseteq \win{strongly}{\mu}(f) \subseteq \win{weakly}{\mu}(f) \subseteq \win{event}{\mu}(f)$, and 
\item $\winsure{\lambda}(f) \subseteq \winas{\lambda}(f) \subseteq \winlim{\lambda}(f)$.
\end{itemize}
\end{remark}

\begin{lemma}\label{lem:always}
Let $T$ be a set of states. For all functions $f \in \{\fmax_T, \fsum_T\}$,  
we have $\winsure{always}(f) = \winas{always}(f) = \winlim{always}(f)$.
\end{lemma}

\begin{proof}
By Remark~\ref{rem:expressiveness}, we obtain a cyclic chain of inclusions (thus 
an overall equality) if we show that $$\winlim{always}(f) \subseteq \winsure{always}(f),$$
that is for all distributions $d_0$, if $\M$ is limit-sure always synchronizing from $d_0$, then 
$\M$ is sure always synchronizing from $d_0$. 
For $f=\fmax_T$, consider $\epsilon$ smaller than the
smallest positive probability in the initial distribution $d_0$ and in the 
transitions of the MDP $\M = \tuple{Q,\Act,\delta}$.
If $\M$ is limit-sure always synchronizing, then by definition there exists an always $(1-\epsilon)$-synchronizing 
strategy~$\alpha$, and it is easy to show by induction on~$k$ 
that the distributions $\M^{\alpha}_k$ are Dirac for all $k \geq 0$. 
In particular $d_0$ is Dirac, and let $q_{\init} \in T$ be such that $d_0(q_{\init})=1$.
It follows that there is an infinite path from $q_{\init}$
in the graph $\tuple{T,E}$ where $(q,q') \in E$ if there exists an action $a \in \Act$
such that $\delta(q,a)(q') = 1$. The existence of this
path entails that there is a loop reachable from 
$q_{\init}$ in the graph $\tuple{T,E}$, and this naturally defines a sure-winning 
always synchronizing strategy in $\M$.
A similar argument for $f=\fsum_T$ shows that for sufficiently  small 
$\epsilon$, an always $(1-\epsilon)$-synchronizing strategy~$\alpha$ must 
produce a sequence of distributions with support contained in $T$, until
some support repeats in the sequence. This naturally induces 
an always $1$-synchronizing strategy. 
\qed
\end{proof}\medskip

The results established in this article will entail that the almost-sure
and limit-sure modes coincide for weakly and strongly synchronizing (see
Theorem~\ref{theo:weakly-ls-is-as}, Corollary~\ref{col: almost-limit-strong}, and Corollary~\ref{col: almost-limit-strong-sum}). 
The other winning regions are distinct, as shown in the following lemma.

\begin{lemma}\label{lem:dif-in-def}
There exists an MDP $\M$ and states $q_1,q_2$ such that:\smallskip

\begin{tabular}{l@{\,}r@{\,}l}
& $(i)$  & $\winsure{\lambda}(q_1) \subsetneq \winas{\lambda}(q_1)$ for all $\lambda \in \{event, weakly, strongly\}$, and\\[4pt]
& $(ii)$ & $\winas{event}(q_2) \subsetneq \winlim{event}(q_2)$.
\end{tabular}
\end{lemma}

\begin{proof}
Consider the MDP $\M$ with states $q_{\init}, q_1, q_2, q_3$ and actions $a,b$ as shown in 
\figurename~\ref{fig:almost-limit-eventually-differ}. 
All transitions are deterministic except from $q_{\init}$ where on all actions, 
the successors are $q_{\init}$ and $q_1$ with probability $\frac{1}{2}$. 

To establish $(i)$, it is sufficient to prove that 
$$q_{\init} \in \winas{strongly}(q_1) \ \text{ and }\  q_{\init} \not \in \winsure{event}(q_1),$$ 
because, by Remark~\ref{rem:expressiveness}, it implies that 
$$q_{\init} \in \winas{\lambda}(q_1) \ \text{ and }\  q_{\init} \not \in \winsure{\lambda}(q_1) 
\text{ for all } \lambda \in \{event, weakly, strongly\},$$
establishing all strict inclusions at once.
To prove that $q_{\init} \in \winas{strongly}(q_1)$, consider the pure 
strategy that always plays $a$. The outcome is such that the probability
to be in $q_1$ after $k$ steps is $1-\frac{1}{2^k}$, showing that $\M$
is almost-sure winning for the strongly synchronizing objective in $q_1$ (from $q_{\init}$).
On the other hand, $q_{\init} \not \in \winsure{event}(q_1)$ because
for all strategies~$\alpha$, the probability in $q_{\init}$ remains always positive,
and thus in $q_1$ we have $\M^{\alpha}_n(q_1) < 1$ for all $n \geq 0$,
showing that $\M$ is not sure winning for the eventually synchronizing 
objective in $q_1$ (from $q_{\init}$).

To establish $(ii)$, we first show that $\M$ is limit-sure winning 
for the eventually synchronizing objective in $q_2$ (from $q_{\init}$):
for $k \geq 0$ consider a strategy that plays $a$ 
for $k$ steps, and then plays $b$. Then the probability
to be in $q_2$ after $k+1$ steps is $1-\frac{1}{2^k}$, showing that this strategy
is eventually $(1-\frac{1}{2^k})$-synchronizing in $q_2$. 

Second, we show that almost-sure eventually synchronizing is impossible
because, to get probability $1-\epsilon$ in $q_2$, the probability mass needs 
to accumulate for more and more steps in $q_1$ as $\epsilon$ gets smaller,
which cannot be achieved by a single strategy.
Formally, for all strategies, since the probability in $q_{\init}$ remains always 
positive, the probability in $q_2$ is always smaller than $1$. Moreover,
if the probability $p$ in $q_2$ is positive after $n$ steps ($p>0$), 
then after any number $m > n$ of steps, the probability in $q_2$ is bounded by
$1-p < 1$. It follows that the probability in $q_2$ is never equal to $1$ and 
cannot tend to $1$ for $m \to \infty$, showing that $\M$ is not almost-sure winning for 
the eventually synchronizing objective in $q_2$ (from $q_{\init}$).
\qed
\end{proof}\medskip

Finally, for eventually and weakly synchronizing we present in 
Lemma~\ref{lem:weakly-max-sum} a reduction of the membership problem 
with function $\fmax_T$ to the membership problem with function $\fsum_{T'}$
for a singleton ${T'}$. It follows that the complexity results established
in this article for eventually and weakly synchronizing with function $\fsum_{T}$ 
also hold with function $\fmax_{T}$
(this is trivial for the upper bounds, and for the lower bounds it follows
from the fact that our hardness results hold for $\fsum_{T}$ with singleton $T$,
and thus for $\fmax_{T}$ as well since in this case $\fsum_{T} = \fmax_{T}$).

\begin{lemma}\label{lem:weakly-max-sum}
For eventually and weakly synchronizing, in each winning mode
the following problems are polynomial-time equivalent:
\begin{itemize}
\item the membership problem with a function $\fmax_T$ where $T$ is an arbitrary subset of the state space, and
\item the membership problem with a function $\fsum_{T'}$ where $T'$ is a singleton.
\end{itemize}
\end{lemma}

\begin{figure}[t]
\begin{center}
    \begin{picture}(120,50)(0,0)

\node[Nmarks=n](l)(5,25){$r$}
\node[Nmarks=n](p)(25,38){$s$}
\node[Nmarks=n](q)(25,12){$q$}

\drawedge[ELside=r,ELdist=0](l,p){$a:\frac{1}{5}$}
\drawedge[ELside=l,ELdist=0](l,q){$a:\frac{4}{5}$}

\drawedge[ELside=l, curvedepth=6](q,l){$a$}
\drawloop[ELside=l,loopCW=y, loopangle=0, loopdiam=5](p){$a$}

\node[Nframe=n](arrow)(45,24){{\Large $\Rightarrow$}}

\node[Nmarks=n](l1)(70,34){$r_1$}
\node[Nmarks=n](l2)(70,16){$r_2$}
\node[Nmarks=n](p1)(105,45){$s_1$}
\node[Nmarks=n](p2)(105,27){$s_2$}
\node[Nmarks=n](q)(105,5){$q$}

\drawedge[ELside=l,ELpos=70, ELdist=0](l1,p1){$\frac{1}{10}$}
\drawedge[ELside=l,ELpos=73, ELdist=.5](l1,p2){$\frac{1}{10}$}
\drawedge[ELside=l,ELpos=73, ELdist=0](l1,q){$\frac{4}{5}$}

\drawedge[ELside=l,ELpos=16, ELdist=0](l2,p1){$\frac{1}{10}$}
\drawedge[ELside=r,ELpos=25, ELdist=0](l2,p2){$\frac{1}{10}$}
\drawedge[ELside=l,ELpos=40, ELdist=.5](l2,q){$\frac{4}{5}$}

\drawbpedge[ELpos=70,ELside=l](q,210,30,l1,200,30){$\frac{1}{2}$} 
\drawedge[ELside=l, ELpos=75, ELdist=0, curvedepth=5](q,l2){$\frac{1}{2}$}

\drawloop[ELside=l,loopCW=y, loopangle=0, loopdiam=5](p1){$\frac{1}{2}$}
\drawedge[ELside=l,ELpos=52, curvedepth=3](p1,p2){$\frac{1}{2}$}
\drawedge[ELside=l,ELpos=48, curvedepth=3](p2,p1){$\frac{1}{2}$}
\drawloop[ELside=l,loopCW=y, loopangle=0, loopdiam=5](p2){$\frac{1}{2}$}

\end{picture}
\end{center} 
 \caption{
State duplication ensures that the  probability mass can never 
be accumulated in a single state except in~$q$ (we omit action $a$ for readability).
\label{fig:twin}}
\end{figure}
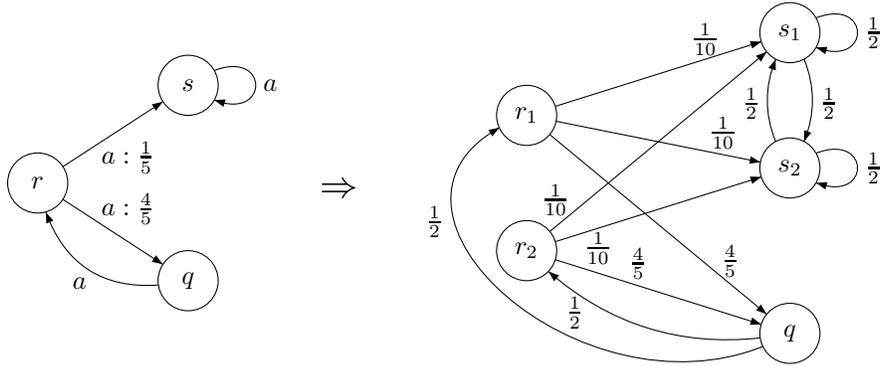

\begin{proof}
Let $\mu \in \{sure, almost, limit\}$ and $\lambda \in \{event, weakly\}$.
First we have $\win{\lambda}{\mu}(\fmax_T) = \bigcup_{q \in T} \win{\lambda}{\mu}(q)$,
showing that the membership problems for $\fmax$ and $\fmax_T$
are polynomial-time reducible to the corresponding membership problem 
for $\fsum_{T'}$ with singleton~$T'$. 

The reverse reduction is as follows. 
Given an MDP $\M$, a state $q$ and an initial distribution $d_0$,
we can construct an MDP $\M'$ and initial distribution $d'_0$
such that $d_0 \in \win{\lambda}{\mu}(q)$ iff $d'_0 \in \win{\lambda}{\mu}(\fmax_{Q'})$
where $Q'$ is the state space of $\M'$ (thus $\fmax_{Q'}$ is simply the function $\max$). 
The idea is to construct $\M'$ and $d'_0$
as a copy of~$\M$ and~$d_0$ where all states except $q$ are duplicated, and the 
initial and transition probabilities are equally distributed between 
the copies (see \figurename~\ref{fig:twin}). 
Therefore, if the probability tends to~$1$ in some state,
it has to be in~$q$ (since the probability in any other state of $\M'$ can be at most $1/2$).
\qed
\end{proof}\medskip

The rest of this article is devoted to the solution of the membership problem. 
By definition, a sequence of probability distribution is sure always synchronizing 
according to $\fsum_T$ if all supports are contained in $T$, i.e., all states in
all path-outcomes are in $T$, that is all path-outcomes are contained in $\Box T$.
Hence, it follows from the proof of Lemma~\ref{lem:always} that the winning region for  
always synchronizing according to $\fsum_T$ coincides with the set of winning initial 
distributions for the safety objective $\Box T$ in the traditional semantics,
which can be computed in polynomial time~\cite{CH12}. 
Moreover, always synchronizing according to $\fmax_T$ 
is equivalent to the existence of an infinite path staying in $T$
in the transition system $\tuple{Q,R}$ of the MDP restricted 
to transitions $(q,a,q') \in R$ such that $\delta(q,a)(q')=1$,
which can also be decided in polynomial time. 
In both cases, pure memoryless strategies are sufficient.

\begin{theorem}\label{theo:always}
The membership problem for always synchronizing can be solved in polynomial
time, and pure memoryless strategies are sufficient.
\end{theorem}

\begin{remark}\label{rmk:Dirac-initial-suffices}
For the other synchronizing modes (eventually, weakly, and strongly synchronizing),
it is sufficient to consider Dirac initial distributions 
(i.e., assuming that MDPs have a single initial state) because the answer to 
the general membership problem for an MDP $\M$ with initial distribution $d_0$ 
can be obtained by solving the membership problem for a copy of $\M$ with a
new initial state from which the successor distribution 
on all actions is $d_0$. 
\end{remark}

\begin{table}[t]
\begin{center}
\caption{Computational complexity of the membership problem.\label{tab:complexity}}{
\begin{tabular}{|l@{\ }|c|c|c|c|}
\hline                        
\large{\strut}           &  Always & Eventually & Weakly   & Strongly \\
\hline                        
Sure \large{\strut}      &         &  \;PSPACE-C\;   & \;PSPACE-C\; & \;PTIME-C\;  \\
\cline{1-1} \cline{3-5}
Almost-sure \large{\strut} & PTIME-C & PSPACE-C  & \multirow{2}{*}{PSPACE-C} & \multirow{2}{*}{PTIME-C}  \\
\cline{1-1} \cline{3-3} 
Limit-sure \large{\strut}  &       & PSPACE-C &   &   \\
\hline
\end{tabular}  
}
\end{center}
\end{table}

In the rest of the article, we present algorithms to decide the membership problem
and we establish matching upper and lower bounds for the complexity of the problem:
we show that eventually and weakly synchronizing are PSPACE-complete,
whereas strongly synchronizing is PTIME-complete (like always synchronizing).
We also establish optimal memory bounds for the memory needed by strategies to win.
Our results will show that pure strategies are sufficient in all modes.
The complexity results are summarized in Table~\ref{tab:complexity},
and we present the memory requirement for winning strategies in Table~\ref{tab:memory}.

\begin{table}[t]
\begin{center}
\caption{Memory requirement.\label{tab:memory}}{
\scalebox{0.90}{
\begin{tabular}{|l@{\ }|c|c|c|C{22mm}|C{22mm}|}
\hline                        
\large{\strut}           &  \multirow{2}{*}{Always} &  \multirow{2}{*}{Eventually} & \multirow{2}{*}{Weakly}   & \multicolumn{2}{c|}{Strongly} \\
\cline{5-6}
\large{\strut}           &                          &                              &                           & $\fsum_T$      & $\fmax_T$ \\
\hline                        
Sure \large{\strut}      &                          & \;exponential\;              & \multirow{1}{*}{\;exponential\;}          & \multirow{1}{*}{\;memoryless\;} & \multirow{1}{*}{\;linear\;}  \\
\cline{1-1}\cline{3-6}
Almost-sure \large{\strut} & \;memoryless\;         & infinite                     & \multirow{2}{*}{infinite} & \multirow{2}{*}{memoryless} & \multirow{2}{*}{linear}  \\
\cline{1-1}\cline{3-3}
Limit-sure \large{\strut}  &                        & unbounded   &  &  &  \\
\hline
\end{tabular}  
}
}
\end{center}
\end{table}

\subsection{One-Letter Alternating Automata}\label{sec:1L-AFA}

In this section, we consider one-letter alternating automata (1L-AFA) and show 
a tight connection with MDP. We present complexity results for 1L-AFA that are useful to
establish the PSPACE lower bounds for eventually and weakly synchronizing in MDPs
(in Theorem~\ref{theo:sure-eventually-pspace-c}, Lemma~\ref{lem:as-event-pspace-hard},~\ref{lem:limit-event-pspace-hard}, and~\ref{lem: sure-weakly-psapce-hradness}).

In 1L-AFA, the alphabet is a singleton, and thus only the length of a word is relevant.
The transitions of an alternating automaton are described by Boolean formulas
over the set of automaton states using only $\land$ and $\lor$ (but no negation). 
For example, if the formula $(q_2 \land q_3) \lor q_4$ describes the transitions
from a state $q_1$, then the word of length $n$ is accepted from $q_1$ if the word 
of length $n-1$ is accepted either from both $q_2$ and $q_3$, or from $q_4$.

\paragraph{One-letter alternating automata}
Let $\Bool(Q)$ be the set of positive Boolean formulas over a set~$Q$, i.e. Boolean
formulas built from elements in $Q$ using $\land$ and $\lor$ (but no negation). A set 
$S \subseteq Q$ \emph{satisfies} a formula $\varphi \in \Bool(Q)$ (denoted $S \models \varphi$)
if $\varphi$ is satisfied when replacing in $\varphi$ the elements in $S$ by \true, 
and the elements in $Q \setminus S$ by \false.

A \emph{one-letter alternating finite automaton}  is a tuple 
$\A=\tuple{Q,\delta_{\A},\F}$ where 
$Q$ is a finite set of states, 
$\delta_{\A}: Q \to \Bool(Q)$ is the transition function, 
and $\F \subseteq Q$ is the set of accepting states. 
We assume that the formulas in transition function are in disjunctive normal form.
Note that the alphabet of the automaton is omitted, as it consists of a single letter. 
As in the language of a 1L-AFA, only the length of words is relevant, define 
for all $n\geq 0$, the set $Acc_{\A}(n,\F) \subseteq Q$ of states from which 
the word of length $n$ is accepted by $\A$ as follows:
\begin{itemize}
	\item $Acc_{\A}(0,\F) = \F$;
	\item $Acc_{\A}(n,\F) = \{q \in Q \mid Acc_{\A}(n-1,\F) \models \delta(q) \}$ for all $n > 0$.
\end{itemize}


The set $\LL(\A_q) = \{n \in \nat \mid q \in Acc_{\A}(n,\F)\}$ is the 
\emph{language} accepted by $\A$ from state~$q$ (called \emph{initial} state 
in this context).

\begin{figure}[!t]
  \centering
  \begin{minipage}[b]{0.46\textwidth}
    \hrule
    \subfigure[1L-AFA.]{\label{fig:afa-AFA}\begin{picture}(63,36)(0,5)

\node[Nmarks=i](n0)(9,25){$q_{\init}$}
\node[Nmarks=r](n1)(24,35){$q_1$}
\node[Nmarks=n](n2)(24,15){$q_2$}
\node[Nmarks=r](n3)(44,35){$q_3$}
\node[Nmarks=r](n4)(44,15){$q_4$}
\node[Nmarks=n](n5)(59,25){$q_5$}

\node[Nframe=n, Nmarks=n](label)(9,37){$\A$}

\drawedge(n0,n1){}
\drawedge(n0,n2){}
\drawarc[](9,25,6.5,326.3,33.7)

\drawedge[sxo=-2, syo=-2, eyo=1, curvedepth=3](n1,n0){}
\drawedge[sxo=0, syo=-2, exo=2, curvedepth=-3](n1,n2){}
\drawarc[](24,35,6.5,235,255)

\drawedge[sxo=1.6, syo=-1.5, exo=1, curvedepth=3](n1,n2){}
\drawedge[sxo=1, syo=-1.5, eyo=-1, curvedepth=-3](n1,n3){}
\drawarc[](24,35,6.5,303,328)

\drawedge(n2,n4){}
\drawedge(n3,n1){}
\drawedge(n3,n4){}
\drawedge(n3,n5){}
\drawarc[](44,35,6.5,180,270)
\drawedge(n4,n5){}
\drawloop[ELside=l,loopCW=n, loopangle=270, loopdiam=5](n2){}
\drawloop[ELside=l,loopCW=y, loopangle=90, loopdiam=5](n5){}

\end{picture}}
    \subfigure[MDP.]{\label{fig:afa-mdp}\begin{picture}(63,44)(0,3)

\node[Nmarks=i](n0)(9,25){$q_{\init}$}
\node[Nmarks=r](n1)(24,35){$q_1$}
\node[Nmarks=n](n2)(24,15){$q_2$}
\node[Nmarks=r](n3)(44,35){$q_3$}
\node[Nmarks=r](n4)(44,15){$q_4$}
\node[Nmarks=n](n5)(59,25){$q_5$}

\node[Nframe=n, Nmarks=n](label)(9,38){$\M$}

\drawedge[ELside=l, ELdist=-2](n0,n1){\rotatebox{33.7}{{\scriptsize $a,b\!:\!\frac{1}{2}$}}}
\drawedge[ELside=r, ELdist=-2](n0,n2){\rotatebox{326.3}{{\scriptsize $a,b\!:\!\frac{1}{2}$}}}

\drawedge[ELpos=40, ELside=l, ELdist=0, sxo=-2, syo=-2, eyo=1, curvedepth=3](n1,n0){\scriptsize $b\!:\!\frac{1}{2}$}
\drawedge[ELpos=55, ELside=l, curvedepth=0](n1,n2){{\scriptsize $a,b\!:\!\frac{1}{2}$}}
\drawedge[ELpos=50, ELside=r, sxo=1, syo=-1.5, eyo=-1, curvedepth=-3](n1,n3){\scriptsize $a\!:\!\frac{1}{2}$}

\drawedge[ELside=r](n2,n4){\scriptsize $b$}
\drawedge[ELside=r](n3,n1){\scriptsize $a\!:\!\frac{1}{2}$}
\drawedge[ELside=l](n3,n4){\scriptsize $a\!:\!\frac{1}{2}$}
\drawedge[ELside=l](n3,n5){\scriptsize $b$}
\drawedge[ELside=r, ELdist=.5](n4,n5){\scriptsize $a,b$}
\drawloop[ELside=r,loopCW=n, loopangle=270, loopdiam=5](n2){\scriptsize $a$}
\drawloop[ELside=l, ELdist=.5, loopCW=y, loopangle=90, loopdiam=5](n5){\scriptsize $a,b$}

\end{picture}}
    \hrule
    \caption{1L-AFA and MDP. \label{fig:AFA}}
  \end{minipage}
  \hfill
  \begin{minipage}[b]{0.53\textwidth}
    \hrule
    \subfigure[Execution tree.]{\label{fig:afa-tree}\begin{picture}(72,60)(0,0)

\gasset{Nframe=y}

\node[Nmarks=n](n)(37,53){$q_{\init}$}

\node[Nmarks=r](n0)(27,38){$q_1$}
\node[Nmarks=n](n1)(47,38){$q_2$}

\node[Nmarks=r](n00)(22,23){$q_3$}
\node[Nmarks=n](n01)(34,23){$q_2$}
\node[Nmarks=n](n10)(47,23){$q_2$}

\node[Nmarks=r](n000)(17,8){$q_1$}
\node[Nmarks=r](n001)(27,8){$q_4$}
\node[Nmarks=r](n010)(37,8){$q_4$}
\node[Nmarks=r](n100)(47,8){$q_4$}

\drawedge(n,n0){}
\drawedge(n,n1){}

\drawedge(n0,n00){}
\drawedge(n0,n01){}
\drawedge(n1,n10){}

\drawedge(n00,n000){}
\drawedge(n00,n001){}
\drawedge(n01,n010){}
\drawedge(n10,n100){}

\end{picture}}
    \subfigure[Predecessor sequence (determinization of $\A$).]{\label{fig:afa-pred}\begin{picture}(72,20)(0,1)

\gasset{Nframe=y}

\node[Nmarks=i, iangle=0, Nadjust=h, Nw=7, Nmr=2](n0)(64,10){$\begin{array}{l}q_1\\q_3\\q_4\end{array}$}
\node[Nmarks=n, Nadjust=h, Nw=7, Nmr=2](n1)(43,10){$\begin{array}{l}q_2\\q_3\end{array}$}
\node[Nmarks=n, Nadjust=h, Nw=7, Nmr=2](n2)(24,10){$\begin{array}{l}q_1\\q_2\end{array}$}
\node[Nmarks=r, Nadjust=h, Nw=8, Nmr=2](n3)(6,10){$\begin{array}{l}q_{\init}\\q_2\end{array}$}
\drawedge[ELside=l, ELdist=1.5](n0,n1){\scriptsize $\Pre_{\M}(\cdot)$}
\drawedge[ELside=r](n0,n1){\scriptsize $Acc_{\A}(1,\cdot)$}
\drawedge(n1,n2){}
\drawedge(n2,n3){}
\drawedge[syo=2, eyo=2, curvedepth=5](n3,n2){}

\end{picture}}
    \hrule
    \caption{Execution tree and predecessor sequence.\label{fig:frequency-not-connected}} 
  \end{minipage}
\end{figure}

\paragraph{Example}
Consider the 1L-AFA $\A$ in \figurename~\ref{fig:afa-AFA} with initial state
$q_{\init}$. Transition function is defined by 
$\delta_{\A}(q_{\init}) = q_1 \land q_2$ and 
$\delta_{\A}(q_{1}) = (q_{\init} \land q_2) \lor (q_2 \land q_3)$, etc.
The word of length~$3$ is accepted by $\A$, as witnessed by
the execution tree in \figurename~\ref{fig:afa-tree}: for every node
of the tree (let $q$ be its label), 
the labels of the successor nodes form a set that satisfies the 
transition function at $q$. The root of the tree is labeled by $q_{\init}$
and all leaves are accepting. Note that all branches are of the same length,
namely $3$, the length of the input word.
\smallskip

For every 1L-AFA with $n$ states, there is an equivalent deterministic automaton
with at most $2^n$ states (that accepts the same language), which can be constructed 
as follows~\cite{CKS81}. 
It is easier to think that the deterministic automaton accepts the \emph{reverse} image
of the words in the language of $\A$, which is the same as the language of $\A$.
For all $n \geq 1$, we view $Acc_{\A}(n,\cdot)$ as an operator on $2^Q$
that, given a set $\F \subseteq Q$ computes the set $Acc_{\A}(n,\F)$. 
Note that $Acc_{\A}(n,\F) = Acc_{\A}(1,Acc_{\A}(n-1,\F))$ for all $n \geq 1$.
Call $Acc_{\A}(1,s)$ the \emph{predecessor} of $s$.
The deterministic automaton has state space $\{s_0,\dots,s_{2^n} \}$
where $s_i = Acc_{\A}(i,\F)$, and a deterministic transition from $s_i$
to its predecessor $Acc_{\A}(1,s_i)$. The sequence of predecessors 
for the 1L-AFA of \figurename~\ref{fig:afa-AFA} is shown in \figurename~\ref{fig:afa-pred}.
It is easy to see that this sequence is always ultimately periodic (for all $k > 2^n$
there exists $k' \leq  2^n$ such that $s_k = s_{k'}$), and therefore
the transitions of the deterministic automaton are well defined.
Let $s_0 = \F$ be the initial state, and let all $s_i$ such that $q_{\init} \in s_i$
be the accepting states, then its language is $\LL(\A_{q_{\init}})$.
Considering the sequence of predecessors in \figurename~\ref{fig:afa-pred} it 
is easy to see that the language of $\A$ is the set $\{n > 1 \mid n \text{ is odd}\}$
of odd numbers greater than $1$. Note that the language of $\A$ is nonempty
because in the sequence of predecessors there is a state $s$ such that $q_{\init} \in s$,
and the language of $\A$ is infinite 
because there is such a state $s$ in the periodic part of the sequence.

\paragraph{Relation with MDPs}
Consider the MDP $\M$ in \figurename~\ref{fig:afa-mdp}, obtained from $\A$ by
transforming the transition function, which is a disjunction of conjunctions
of states as follows: each conjunction is replaced by a uniform probability distribution
over its elements, and the elements of the disjunction are labeled by a 
letter from the alphabet of $\M$.

The correspondence between 1L-AFA and MDPs is based on the observation
that they have the same underlying structure of alternating graph (or AND-OR graph):

\begin{itemize}
\item the combinatorial structure of an MDP $\tuple{Q,\Act,\delta}$ can be described as an alternating
graph, with \emph{existential} vertices $q \in Q$ with successors $(q,a)$ for all
actions $a \in \Act$, and \emph{universal} vertices $(q,a) \in Q \times \Act$ 
with successors $q' \in  \Supp(\delta(q,a))$; 
\item the structure of a 1L-AFA $\tuple{Q,\delta_{\A},\F}$ is also an alternating graph,
where states $q \in Q$ are \emph{existential} vertices with successors the clauses $c_1,\dots,c_m$
such that $\delta_{\A}(q) = c_1 \lor \dots \lor c_m$ where each $c_i$ is a conjunctive clause,
and the conjunctive clauses $c_i$ are \emph{universal} vertices with successors
the states that belong to $c_i$. 
\end{itemize}

The common structure of 1L-AFA and MDP is illustrated in \figurename~\ref{fig:AFA} 
where we show only the existential vertices.
The correspondence is defined formally as follows.
Given a 1L-AFA $\A = \tuple{Q, \delta_{\A}, \F}$, assume without loss of generality
that the transition function $\delta_{\A}$ is such that 
$\delta_{\A}(q) = c_1 \lor \dots \lor c_m$ has the same number $m$ of conjunctive clauses 
for all $q \in Q$ (we can duplicate clauses if necessary). 
From $\A$, construct
the MDP $\M_{\A} = \tuple{Q, \Act, \delta_{\M}}$ where $\Act = \{a_1, \dots, a_m\}$
and $\delta_{\M}(q,a_k)$ is the uniform distribution over the states occurring
in the $k$-th clause $c_k$ in $\delta_{\A}(q)$, for all $q \in Q$ and $a_k \in \Act$.
Then, we have $Acc_{\A}(n,\F) = \Pre_{\M}^n(\F)$ for all $n \geq 0$.

Similarly, from an MDP $\M$ and a set $T$ of states, we can construct
a 1L-AFA $\A = \tuple{Q, \delta_{\A}, \F}$ with $\F = T$ such that 
$Acc_{\A}(n,\F) = \Pre_{\M}^n(T)$ for all $n \geq 0$: let $\delta_{\A}(q) = \bigvee_{a \in \Act}
\bigwedge_{q' \in \post(q,a)} q'$ for all $q \in Q$. 
For example, for $\Act = \{a,b\}$ if $\delta_{\M}(q_3,a)(q_1) = \delta_{\M}(q_3,a)(q_4) = \frac{1}{2}$
and $\delta_{\M}(q_3,b)(q_5) = 1$, then $\delta_{\A}(q_3) = (q_1 \land q_4) \lor q_5$ (see \figurename~\ref{fig:AFA}).

It follows 
that, up to the correspondence between 1L-AFA and MDPs established above,
$Acc_{\A}(n,T) = \Pre_{\M}^n(T)$. In the sequel we denote the operator $Acc_{\A}(1,\cdot)$
by $\Pre_{\A}(\cdot)$ (note the subscript). 
Then for all $n \geq 0$ the operator $Acc_{\A}(n,\cdot)$ coincides
with $\Pre^n_{\A}(\cdot)$, the $n$-th iterate of $\Pre_{\A}(\cdot)$.

\paragraph{Decision problems}
Questions related to MDPs 
have a corresponding formulation in terms of alternating automata. We show that
such connections exist between synchronizing problems for MDPs and language-theoretic
questions for alternating automata. 

Several decision problems for 1L-AFA can be solved by computing 
the sequence $Acc_{\A}(n,\F)$ (i.e., $\Pre_{\A}^n(\F)$), and analogously we show that synchronizing 
problems for MDPs can also be solved by computing the sequence $\Pre_{\M}^n(\F)$. 
Therefore, the above relationship between 1L-AFA and MDPs provides a tight 
connection that we use in Section~\ref{sec:eventually} to transfer complexity 
results between 1L-AFA and MDPs. 

We review classical decision problems
for 1L-AFA, namely the emptiness and finiteness problems, 
and establish the complexity of a new problem, the 
\emph{universal finiteness problem} which is to decide if from every initial
state the language of a given 1L-AFA is finite. 
These results of independent interest are useful to establish the PSPACE lower 
bounds for eventually and weakly synchronizing in MDPs.


\begin{itemize}
\item 
The \emph{emptiness problem} for 1L-AFA is to decide, given a 1L-AFA $\A$
and an initial state~$q$, whether $\LL(\A_q)=\emptyset$. 
The emptiness problem can be solved by checking whether $q \in \Pre^n_{\A}(\F)$
for some $n \geq 0$. 
It is known that the emptiness problem is PSPACE-complete,
even for transition functions in disjunctive normal form~\cite{Holzer95,AFA1}.  

\item 
The \emph{finiteness problem} is to decide, given a 1L-AFA $\A$
and an initial state $q$, whether $\LL(\A_q)$ is finite.
The finiteness problem can be solved in (N)PSPACE by 
guessing $n,k \leq 2^{\abs{Q}}$ such that $\Pre^{n+k}_{\A}(\F) = \Pre^n_{\A}(\F)$
and $q \in \Pre^n_{\A}(\F)$. The finiteness problem is PSPACE-complete
by a simple reduction from the emptiness problem: from an instance $(\A,q)$
of the emptiness problem, construct $(\A',q')$ where $q'=q$ and 
$\A' = \tuple{Q,\delta',\F}$ is a copy of $\A = \tuple{Q,\delta,\F}$ 
with a self-loop on~$q$ (formally, $\delta'(q) = q \lor \delta(q)$
and $\delta'(r) = \delta(r)$ for all $r \in Q \setminus \{q\}$).
It is easy to see that $\LL(\A_q)=\emptyset$ iff $\LL(\A'_{q'})$ is finite.

\item 
The \emph{universal finiteness problem} 
is to decide, given a 1L-AFA $\A$, whether $\LL(\A_q)$ is finite
for all states $q$. This problem can be solved by checking whether 
$\Pre^n_{\A}(\F) = \emptyset$ for some $n \leq 2^{\abs{Q}}$, 
and thus it is in PSPACE.
Note that if $\Pre^n_{\A}(\F) = \emptyset$, then $\Pre^m_{\A}(\F) = \emptyset$
for all $m \geq n$.
\end{itemize}

Given the PSPACE-hardness proofs of the emptiness and finiteness
problems, it is not easy to see that the universal finiteness  
problem is PSPACE-hard.

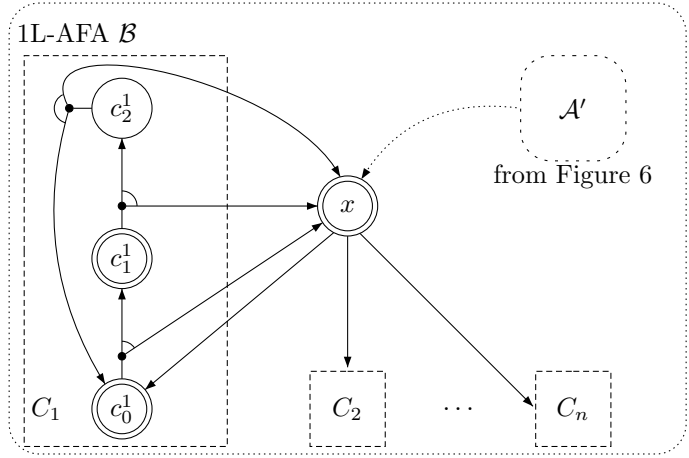
\begin{figure}[t]
\begin{center}
    \begin{picture}(90,60)(0,0)

\node[Nmarks=n, Nw=90, Nh=60, dash={0.2 0.5}0](m1)(45,30){}
\node[Nframe=n](label)(9,56){1L-AFA $\B$}

\node[Nmarks=r](n1)(45,33){$x$}

\drawpolygon[dash={0.8 0.5}0](40,11)(40,1)(50,1)(50,11)
\node[Nframe=n,Nw=1, Nh=1](n3)(45,11){}
\node[Nframe=n,Nw=1, Nh=1](nm3)(47,11){}
\node[Nframe=n](label)(45,6){$C_2$}

\node[Nframe=n](label)(60,6){$\dots $}

\drawpolygon[dash={0.8 0.5}0](70,11)(70,1)(80,1)(80,11)
\node[Nframe=n,Nw=1, Nh=1](n4)(70,6){}
\node[Nframe=n,Nw=1, Nh=1](nm4)(75,11){}
\node[Nframe=n](label)(75,6){$C_n$}

\node[Nmarks=n, Nw=14, Nh=14, dash={0.4 1}0](nm1)(75,46){}
\node[Nframe=n,Nw=1, Nh=1](aaa)(68,46){}
\node[Nframe=n](label)(75,46){$\A'$}
\node[Nframe=n](label)(75,37){from \figurename~\ref{fig:pre-empty-reduction}}

\drawedge(n1,n3){}
\drawedge[syo=-2](n1,n4){}

\drawpolygon[dash={0.8 0.5}0](2,53)(29,53)(29,1)(2,1)
\node[Nframe=n](n2)(5,6){$C_1$}

\node[Nmarks=r](c0)(15,6){$c_0^1$}
\node[Nw=1, Nh=1,Nfill=y](cc0)(15,13){}
\node[Nmarks=r](c1)(15,26){$c_1^1$}
\node[Nw=1, Nh=1,Nfill=y](cc1)(15,33){}
\node[Nmarks=n](c2)(15,46){$c_2^1$}
\node[Nw=1, Nh=1,Nfill=y](cc2)(8,46){}
\drawedge[AHnb=0](c0,cc0){}
\drawedge[AHnb=0](c1,cc1){}
\drawedge[AHnb=0](c2,cc2){}

\drawedge[ELpos=50, ELside=r, syo=-2](n1,c0){}
\drawedge[ELpos=50, ELside=r](cc0,n1){}
\drawedge[ELside=r](cc0,c1){}
\drawarc[](15,13,2,33.7,90)

\drawedge[ELpos=50, ELside=r](cc1,n1){}
\drawedge[ELside=r](cc1,c2){}
\drawarc[](15,33,2,0,90)

\drawbpedge[ELpos=70,ELside=r](cc2,120,15,n1,90,15){} 
\drawedge[ELpos=50, ELside=r,   curvedepth=-5](cc2,c0){}
\drawarc[](8,46,2,116,257)

\drawedge[ELpos=50, ELside=r, dash={0.2 0.5}0,  curvedepth=-5](aaa,n1){}

\end{picture}
\end{center}
 \caption{Sketch of reduction to show PSPACE-hardness of the universal finiteness problem for
1l-AFA. \label{fig:pre-empty-reduction2} }
\end{figure}

\begin{lemma}\label{lem:universal-finiteness-pspace-hard}
The universal finiteness problem for 1L-AFA is PSPACE-hard.
\end{lemma}

\begin{proof}
We show the result by a reduction from the emptiness problem for 1L-AFA,
which is PSPACE-complete~\cite{Holzer95,AFA1}. We first present a basic
fact about 1L-AFA, then an overview of the reduction, and a detailed
description of the reduction and the correctness argument.

\smallskip\noindent{\emph{Basic result.}} The language of a
1L-AFA~$\A = \tuple{Q, \delta, \F}$ from initial state $q_0$ is non-empty if 
$q_0 \in \Pre_{\A}^i(\F)$ for some $i\geq 0$.  Since the sequence
$\Pre_{\A}^i(\F)$ is ultimately periodic, it is sufficient to compute
$\Pre_{\A}^i(\F)$ for every $i\leq 2^{\abs{Q}}$ to decide emptiness.

\smallskip\noindent{\emph{Overview of the reduction.}}
From $\A$, we construct a 1L-AFA~$B = \tuple{Q', \delta', \F'}$ with
set $\F'$ of accepting states such that the sequence $\Pre_B^i(\F')$
in $B$ mimics the sequence $\Pre_{\A}^i(\F)$ in $\A$ for $2^{\abs{Q}}$
steps.  The automaton~$B$ contains the state space of $\A$, i.e. $Q
\subseteq Q'$.  The goal is to have 
$$\Pre_B^i(\F') \cap Q = \Pre_{\A}^i(\F) \text{ for all } i \leq 2^{\abs{Q}}, \text{ as long as } q_0 \not\in \Pre_{\A}^i(\F).$$ Moreover, if $q_0 \in \Pre_{\A}^i(\F)$ for some
$i\geq 0$, then $\Pre_B^j(\F')$ will contain $q_0$ for all $j\geq i$
(the state $q_0$ has a self-loop in $B$), and if $q_0 \not\in
\Pre{\A}^i(\F)$ for all $i\geq 0$, then $B$ is constructed such that
$\Pre_B^j(\F') = \emptyset$ for sufficiently large~$j$ (roughly for $j
> 2^{\abs{Q}}$).  Hence, the language of $\A$ is non-empty if and only
if the sequence $\Pre_B^j(\F')$ is not ultimately empty, that is if
and only if the language of $B$ is infinite from some state (namely
$q_0$).

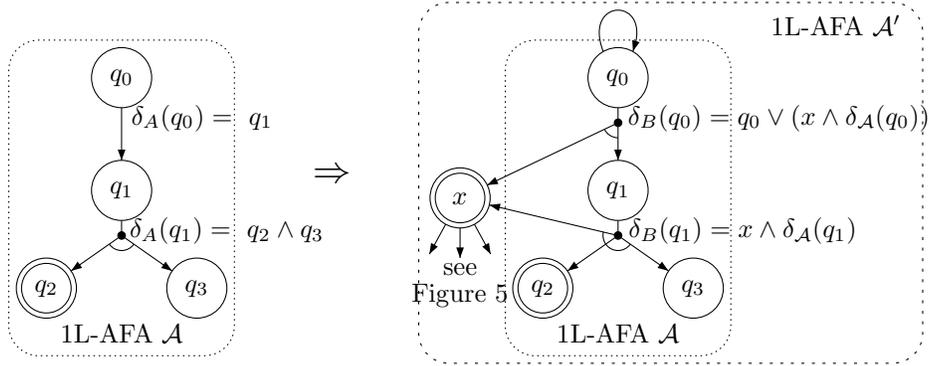
\begin{figure}[t]
\begin{center}
    \begin{picture}(122,49)(0,0)

\node[Nmarks=n, Nw=30, Nh=42, dash={0.2 0.5}0](m1)(15,22){}
\node[Nframe=n](label)(15,4){1L-AFA $\A$}
\node[Nmarks=n](n0)(15,38){$q_0$}
\node[Nframe=n](label)(25.7,32.7){$\delta_{A}(q_0) =~q_1$}
\node[Nmarks=n](n1)(15,23){$q_1$}
\node[Nframe=n](label)(29,17.7){$\delta_{A}(q_1) =~q_2 \land q_3$}
\node[Nw=1, Nh=1,Nfill=y](m1)(15,17){}
\node[Nmarks=r](n2)(5,10){$q_2$}
\node[Nmarks=n](n3)(25,10){$q_3$}
\drawedge(n0,n1){}
\drawedge[AHnb=0](n1,m1){}
\drawedge[ELpos=50](m1,n2){}
\drawedge[ELpos=50](m1,n3){}
\drawarc[](15,17,2,215,325)

\node[Nframe=n](arrow)(43,25){{\Large $\Rightarrow$}}

\node[Nmarks=n, Nw=30, Nh=42, dash={0.2 0.5}0](nm1)(81,22){}
\node[Nmarks=n, Nw=67, Nh=48, dash={0.4 1}0](m2)(88,24){}
\node[Nframe=n](label)(110,45){1L-AFA $\A'$}
\node[Nframe=n](label)(81,4){1L-AFA $\A$}
\node[Nmarks=n](nn0)(81,38){$q_0$}
\node[Nframe=n](label)(102.4,32.7){$\delta_{B}(q_0) = q_0 \lor (x \land \delta_{\A}(q_0))$}
\node[Nw=1, Nh=1,Nfill=y](mm0)(81,32){}
\node[Nmarks=n](nn1)(81,23){$q_1$}
\node[Nframe=n](label)(97.6,17.7){$\delta_{B}(q_1) = x \land \delta_{\A}(q_1)$}
\node[Nw=1, Nh=1,Nfill=y](mm1)(81,17){}
\node[Nmarks=r](nn2)(71,10){$q_2$}
\node[Nmarks=n](nn3)(91,10){$q_3$}
\drawedge[AHnb=0](nn0,mm0){}
\drawedge[AHnb=0](nn1,mm1){}

\node[Nmarks=r](nn4)(60,22){$x$}

\node[Nframe=n](dummy)(54,11){}
\drawedge[ELpos=50](nn4,dummy){}
\node[Nframe=n](dummy)(60,10){}
\drawedge[ELpos=50](nn4,dummy){}
\node[Nframe=n](dummy)(66,11){}
\drawedge[ELpos=50](nn4,dummy){}
\node[Nframe=n](dummy)(60,12.5){see}
\node[Nframe=n](dummy)(60,9){\figurename~\ref{fig:pre-empty-reduction2}}

\drawedge(mm0,nn1){}
\drawedge[ELpos=50, ELside=r](mm0,nn4){}
\drawarc[](81,32,2,205.5,270)

\drawedge[ELpos=50](mm1,nn2){}
\drawedge[ELpos=50](mm1,nn3){}
\drawedge[ELpos=50, ELside=r](mm1,nn4){}
\drawarc[](81,17,2,166.6,325)
\drawloop[ELside=l,loopCW=y, loopangle=90, loopdiam=5](nn0){}

\end{picture}
\end{center}
 \caption{Detail of the copy $\A'$ obtained from $\A$ in the reduction
of \figurename~\ref{fig:pre-empty-reduction2}.}\label{fig:pre-empty-reduction}
\end{figure}

\smallskip\noindent{\emph{Detailed reduction.}}
The key is to let $B$ simulate $\A$ for exponentially many steps, and 
to ensure that the simulation stops if and only if $q_0$ is not reached within $2^{\abs{Q}}$ steps.
We achieve this by defining $B$ as the gadget in~\figurename~\ref{fig:pre-empty-reduction2}
connected to a modified copy $\A'$ of $\A$ with the same state space. 
The transitions in $\A'$ are defined as follows, where $x$ is the entry state
of the gadget (see \figurename~\ref{fig:pre-empty-reduction}):
for all $q \in Q$ let $(i)$ $\delta_{B}(q) = x \land \delta_{\A}(q)$ if $q \neq q_0$,
and $(ii)$ $\delta_{B}(q_0) = q_0 \lor (x \land \delta_{\A}(q_0))$.
Thus, $q_0$ has a self-loop, and given a set $S \subseteq Q$ in the automaton $\A$, 
if $q_0 \not\in S$, then $\Pre_{\A}(S) = \Pre_B(S \cup \{x\})$ that is 
$\Pre_B$ mimics $\Pre_{\A}$ when $x$ is in the argument (and $q_0$ has not been reached yet).
Note that if $x \not\in S$ (and $q_0 \not\in S$), 
then $\Pre_B(S) = \emptyset$, that is unless $q_0$ has been reached, the
simulation of $\A$ by $B$ stops. 
Since we need that $B$ mimics $\A$ for $2^{\abs{Q}}$ steps, we define 
the gadget and the set $\F'$ to ensure that $x \in \F'$ and 
if $x \in \Pre_B^i(\F')$, then $x \in \Pre_B^{i+1}(\F')$ for all $i \leq 2^{\abs{Q}}$. 

In the gadget (\figurename~\ref{fig:pre-empty-reduction2}), the state $x$ has nondeterministic 
transitions $$\delta_{B}(x) = c^1_0 \lor c^2_0 \lor \dots \lor c^n_0$$
to $n$ components with state space 
$C_i = \{c^i_0, \dots, c^i_{p_i-1} \}$ where
$p_i$ is the $(i+1)$-th prime number,
and the transitions\footnote{In expression $c^i_j$,
we assume that $j$ is interpreted modulo $p_i$.} 
$\delta_{B}(c^i_j) = x \land c^i_{j+1}$ ($i=1,\dots,n$) form a loop in each component.
We choose $n$ such that $p^{\#}_n = \prod_{i=1}^{n} p_i > 2^{\abs{Q}}$ (take $n=\abs{Q}$). 
Note that the number of states in the gadget is $1 + \sum_{i=1}^{n} p_i \in O(n^2 \log n)$~\cite{BS96}
and thus the construction is polynomial in the size of $\A$. 

By construction, for all sets $S$, we have $x \in \Pre_B(S)$ whenever the first state $c^i_0$
of some component $C_i$ is in $S$, and if $x \in S$, then  
$c^i_j \in S$ implies $c^i_{j-1} \in \Pre_B(S)$.
Thus, if $x \in S$, the operator $\Pre_B(S)$ 
`shifts' backward the states in each component;
and, $x$ is in the next iteration (i.e., $x \in Pre_B(S)$) 
as long as $c^i_0 \in S$ for some component $C_i$.

Now, define the set of accepting states $\F'$ in $B$ in such a way that 
all states $c^i_0$ disappear simultaneously only after $p^{\#}_n$ iterations.
Let $\F' = \F \cup \{x\} \cup \bigcup_{1 \leq i \leq n} (C_{i} \setminus \{c^i_{p_i-1}\})$, 
thus $\F'$ contains all states of the gadget except the last state of each
component. 

\smallskip\noindent{\emph{Correctness argument.}}
It is easy to check that, irrespective of the transition relation in $\A$,
we have $x \in \Pre_B^{i}(\F')$ if and only if $0 \leq i < p_n^{\#}$. 
Therefore, if $q_0 \in \Pre_{\A}^{i}(\F)$ for some $i$, then $q_0 \in \Pre_B^{j}(\F')$ 
for all $j \geq i$ by the self-loop on $q_0$. On the other hand, 
if $q_0 \not\in \Pre_{\A}^{i}(\F)$ for all $i \geq 0$, then since $x \not\in 
\Pre_B^{i}(\F')$ for all $i > p_n^{\#}$, we have $\Pre_B^{i}(\F') = \emptyset$ 
for all $i > p_n^{\#}$. This shows that the language of $\A$ is non-empty 
if and only if the language of $B$ is infinite from some state (namely $q_0$), 
and establishes the correctness of the reduction. 
\qed
\end{proof}\medskip

\section{Eventually Synchronizing}\label{sec:eventually}

In this section, we show the PSPACE-completeness of the
membership problem for eventually synchronizing objectives
and the three winning modes. By Lemma~\ref{lem:weakly-max-sum} and 
Remark~\ref{rmk:Dirac-initial-suffices}, we consider without loss of generality 
the membership problem with function $\fsum$ and Dirac initial distributions (i.e., single
initial state).

The eventually synchronizing objective is reminiscent of a reachability 
objective in the distribution-based semantics: it requires that 
in the sequence of distributions of an MDP $\M$ under strategy $\alpha$
we have $\sup_{n} \M^{\alpha}_n(T) = 1$ (and that the $\sup$ is reached in 
the case of sure winning, that is $\M^{\alpha}_n(T) = 1$ for some $n \geq 0$).

The sure winning mode can be solved by a reachability analysis 
in the alternating graph underlying the MDP (Section~\ref{sec:sure-event-sync}). We show that the almost-sure winning 
mode can be solved by a reduction to the limit-sure winning mode (Section~\ref{sec:almost-eventually}). 
We solve the limit-sure winning mode by a reduction to a
reachability question in a modified MDP of exponential size that ensures
the probability mass reaches the target set synchronously (Section~\ref{sec:limit-sure}).
We present reductions to show PSPACE-hardness of each winning mode, 
matching our PSPACE upper bounds.


\subsection{Sure eventually synchronizing}\label{sec:sure-event-sync}
Given a target set $T$, the membership problem for sure-winning eventually 
synchronizing objective in $T$ can be solved by computing the sequence 
$\Pre^n(T)$ of iterated predecessors,  like in 1L-AFA, as shown in the 
following lemma.

\begin{lemma}\label{lem:sure-ss-pre}
Let $\M$ be an MDP and $T$ be a target set. 
For all states $q_{\init}$, we have $q_{\init} \in \winsure{event}(\fsum_T)$ 
if and only if there exists $n \geq 0$ such that $q_{\init} \in \Pre_{\M}^{n}(T)$. 
\end{lemma}

\begin{proof}
We prove the following equivalence by induction (on the length $i$):
for all initial states $q_{\init}$, there exists a strategy $\alpha$ sure-winning
in $i$ steps from $q_{\init}$ (i.e., such that $\M^{\alpha}_i(T) = 1$)
if and only if $q_{\init} \in \Pre^{i}(T)$. The case $i=0$ trivially holds 
since for all strategies $\alpha$, we have $\M^{\alpha}_0(T)=1$ if and only if $q_{\init} \in T$.

Assume that the equivalence holds for all $i < n$. 
For the induction step, show that $\M$ is sure eventually synchronizing from $q_{\init}$ (in $n$ steps) 
if and only if there exists an action $a$ such that $\M$ is sure eventually 
synchronizing (in $n-1$ steps) from all states $q' \in \post(q_{\init}, a)$ (equivalently,
$\post(q_{\init}, a) \subseteq \Pre^{n-1}(T)$ by the induction hypothesis, that is 
$q_{\init} \in \Pre^{n}(T)$ by definition of $\Pre$). First, if all successors $q'$ of $q_{\init}$ under some action $a$ 
are sure eventually synchronizing, then so is $q_{\init}$ by playing~$a$ followed by
a winning strategy from each successor $q'$. 
For the other direction, assume towards contradiction that $\M$ is sure eventually 
synchronizing from $q_{\init}$ (in $n$ steps), but for each action~$a$, 
there is a state $q' \in \post(q_{\init}, a)$ that is not sure eventually synchronizing. 
Then, from $q'$ there is a positive probability to reach a state 
not in $T$ after $n-1$ steps, no matter the strategy played. 
Hence from~$q_{\init}$, for all strategies, the probability mass in $T$ cannot be $1$ 
after $n$ steps, in contradiction with the fact that $\M$ is sure eventually 
synchronizing from $q_{\init}$ in $n$ steps.  
It follows that the induction step holds, and the proof is complete.
\qed
\end{proof}\medskip

The following theorem summarizes the results for sure eventually synchronizing. 

\begin{theorem}\label{theo:sure-eventually-pspace-c}
For sure eventually synchronizing  in MDPs:

\begin{enumerate}
\item (Complexity). The membership problem is PSPACE-complete.

\item (Memory). Exponential memory is necessary and sufficient for both pure 
and randomized strategies, and pure  strategies are sufficient. 
\end{enumerate}
\end{theorem}

\begin{figure}[!t]
\begin{center}
\def\fsize{\normalsize}

\begin{picture}(93,63)(0,0)

{\fsize

\node[Nmarks=i, iangle=180](q0)(9,33){$q_{\init}$}
\node[Nmarks=n](q1)(29,54){$q^1_1$}
\node[Nmarks=n](q2)(49,54){$q^1_2$}

\node[Nmarks=n](q3)(29,16){$q^2_1$}
\node[Nmarks=n](q4)(49,26){$q^2_2$}
\node[Nmarks=n](q5)(49,6){$q^2_3$}

\node[Nmarks=r](safe)(69,33){$q_T$}
\node[Nmarks=n](bad)(89,33){$q_{\bot}$}

\node[Nmarks=n, Nw=30, Nh=18, dash={1.5}0, ExtNL=y, NLangle=22, NLdist=1](A1)(39,54){$H_1$}
\node[Nmarks=n, Nw=30, Nh=32, dash={1.5}0, ExtNL=y, NLangle=38, NLdist=1](A2)(39,16){$H_2$}


\drawedge[ELpos=43, ELside=l, ELdist=1, curvedepth=0](q0,q1){$a,b: \frac{1}{2}$}
\drawedge[ELpos=40, ELside=r, ELdist=1, curvedepth=0](q0,q3){$a,b: \frac{1}{2}$}

\drawedge[ELpos=50, ELside=l, ELdist=1, curvedepth=4](q1,q2){$a$}
\drawedge[ELpos=50, ELside=l, ELdist=1, curvedepth=4](q2,q1){$a$}

\drawedge[ELpos=50, ELside=l, ELdist=1, curvedepth=4](q3,q4){$a$}
\drawedge[ELpos=50, ELside=r, ELdist=1.5, curvedepth=4](q4,q5){$a$}
\drawedge[ELpos=40, ELside=l, ELdist=1, curvedepth=4](q5,q3){$a$}

\drawedge[ELpos=50, ELside=l, ELdist=1, curvedepth=0](q2,safe){$b$}
\drawedge[ELpos=50, ELside=r, ELdist=1, curvedepth=0, syo=-3](q5,safe){$b$}


\drawedge[ELpos=48, ELside=l, ELdist=1, curvedepth=0](safe,bad){$a,b$}
\drawloop[ELside=r,loopCW=n, loopdiam=6, loopangle=90](bad){$a,b$}



}
\end{picture}
\caption{The MDP $\M_2$.}\label{fig:exp-mem}
\end{center}
\end{figure}

\begin{proof}
By Lemma~\ref{lem:sure-ss-pre}, 
the membership problem for sure eventually synchronizing is
equivalent to the emptiness problem of 1L-AFA, and thus PSPACE-complete.
Moreover, if $q_{\init} \in \Pre_{\M}^{n}(T)$, a finite-memory strategy with~$n$ modes
that at mode~$i$ in a state~$q$ plays an action~$a$ such that
$\post(q,a)\subseteq \Pre^{i-1}(T)$ is sure winning for eventually synchronizing.  
Note that this strategy is pure.

We present a family of MDPs $\M_n$ ($n \in \nat$) over alphabet $\{a,b\}$
that are sure winning for eventually synchronizing, and where the sure winning strategies 
require exponential memory. The MDP $\M_2$ is shown in \figurename~\ref{fig:exp-mem}.
The structure of $\M_n$ is an initial uniform probabilistic transition
to $n$ components $H_1, \dots, H_n$ where $H_i$ is a cycle of length $p_i$
the $i$th prime number. On action $a$, the next state in the cycle is reached,
and on action $b$ the target state $q_T$ is reached, only from the last
state in the cycles. From other states, 
the action $b$ leads to $q_{\bot}$ (transitions not depicted).
A sure winning strategy for eventually synchronizing in $\{q_T\}$ is to 
play $a$ in the first $p^{\#}_n = \prod_{i=1}^{n} p_i$ steps, and then play $b$.
This requires memory of size $p^{\#}_n > 2^n$ while the size of $\M_n$
is in $O(n^2 \log n)$~\cite{BS96}.
It can be proved by standard pumping arguments that no strategy of size
smaller than $p^{\#}_n$ is sure winning.
\qed
\end{proof}


\subsection{Almost-sure eventually synchronizing}\label{sec:almost-eventually}

We show an example where infinite memory is necessary to win for
almost-sure eventually synchronizing. Consider the MDP
in~\figurename~\ref{fig:inf-mem} with initial state $q_{\init}$.
  We construct a strategy that is almost-sure
  eventually synchronizing in $q_2$, showing that $q_{\init} \in
  \winas{event}(q_2)$.  First, observe that for all $\epsilon > 0$ we
  can have probability at least $1 - \epsilon$ in $q_2$ after finitely
  many steps: playing $n$ times $a$ and then
  $b$ leads to probability $1- \frac{1}{2^n}$ in $q_2$ (and $\frac{1}{2^n}$ in $q_{\init}$) .  Choosing $n$
  sufficiently large (namely, $n > \log_2(\frac{1}{\epsilon})$) shows that the MDP
  is limit-sure eventually synchronizing in $q_2$.  Moreover, the
  remaining probability mass is in~$q_{\init}$.  By playing $a$ we get again
  support $\{q_{\init}\}$, thus from any (initial) distribution with 
  support $\{q_{\init},q_2\}$, the MDP is again limit-sure eventually synchronizing in $q_2$, and with support
  in $\{q_{\init},q_2\}$. Therefore, we can take a smaller value of
  $\epsilon$ and play a strategy to have probability at least $1 - \epsilon$ in $q_2$
  in finitely many steps, then reaching back support $\{q_{\init}\}$, and we can repeat 
  this for $\epsilon \to 0$. This strategy ensures probability mass $1 - \epsilon$ in $q_2$
  for all $\epsilon > 0$, hence it is almost-sure eventually synchronizing in $q_2$.  The
  next result shows that infinite memory is necessary for almost-sure winning in this example.

\begin{figure}[t]
\begin{center}
\begin{picture}(60,28)

\node[Nmarks=i,iangle=180](n0)(10,11){$q_{\init}$}
\node[Nmarks=n](n1)(35,11){$q_1$}
\node[Nmarks=r](n2)(55,11){$q_2$}

\drawedge[ELdist=1](n0,n1){$a: \frac{1}{2}$}
\drawloop[ELside=l, loopCW=y, loopangle=90, loopdiam=5](n0){$a:\frac{1}{2}$}
\drawloop[ELside=r, loopCW=n, loopangle=-90, loopdiam=5](n0){$b$}

\drawedge(n1,n2){$b$}
\drawloop[ELside=r,loopCW=n, loopangle=-90, loopdiam=5](n1){$a$}

\drawedge[ELpos=50, ELdist=.5, ELside=r, curvedepth=-10](n2,n0){$a,b$}

\end{picture}
\caption{An MDP where infinite memory is necessary for 
almost-sure eventually and almost-sure weakly synchronizing strategies.}\label{fig:inf-mem}
\end{center}
\end{figure}
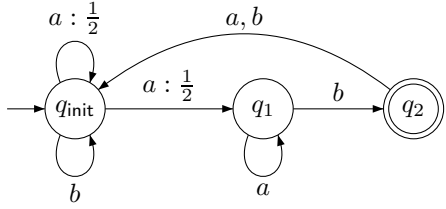

\begin{lemma}\label{lem:inf-mmeory-almost-event}
There exists an almost-sure eventually synchronizing MDP for which 
all almost-sure eventually synchronizing strategies require infinite memory.
\end{lemma}

\begin{proof}
Consider the MDP $\M$ shown in~\figurename~\ref{fig:inf-mem}. We argued 
before the lemma 
that $q_{\init} \in \winas{event}(q_2)$ 
and we now show that infinite memory is necessary from $q_{\init}$ 
for almost-sure eventually synchronizing in $q_2$. 
 Note
that $\M$ is not sure eventually synchronizing in $q_2$ since
the probability in $q_{\init}$ is positive at all times (for all strategies). 

Assume towards contradiction that there exists a (possibly randomized) finite-memory
strategy $\alpha$ that is almost-sure eventually synchronizing in
$q_2$. Consider the Markov chain $\M(\alpha)$ (the product of the MDP
$\M$ with the finite-state transducer defining $\alpha$). A state
$(q,m)$ in $\M(\alpha)$ is called a \emph{$q$-state}. Since $\alpha$
is almost-sure eventually synchronizing (but is not sure eventually
synchronizing) in $q_2$, there is a $q_2$-state in the recurrent states
of $\M(\alpha)$.  Since on all actions $q_{\init}$ is a successor of $q_2$,
and $q_{\init}$ is a successor of itself, it follows that there is a
recurrent $q_{\init}$-state in $\M(\alpha)$, and that all periodic classes
of recurrent states in $\M(\alpha)$ contain a $q_{\init}$-state. Hence, in
each stationary distribution there is a $q_{\init}$-state with a positive
probability, and therefore the probability mass in $q_{\init}$ is bounded
away from zero. It follows that the probability mass in $q_2$ is
bounded away from $1$ thus $\alpha$ is not almost-sure eventually
synchronizing in $q_2$, a contradiction.  
\qed
\end{proof}\medskip

The membership problem for almost-sure eventually synchronizing can be reduced
to other winning modes since an almost-sure eventually synchronizing strategy 
is either sure eventually synchronizing or almost-sure weakly synchronizing. 
Nevertheless we give a direct proof that the problem is decidable in PSPACE, using
a characterization that will be useful later for almost-sure weakly synchronizing. 

It turns out that in general, 
almost-sure eventually synchronizing strategies can be constructed from a family
of limit-sure eventually synchronizing strategies if we can also ensure that the
probability mass remains in the winning region (as in the MDP
in~\figurename~\ref{fig:inf-mem}).
We present a characterization of the winning region for almost-sure
winning based on an extension of the limit-sure eventually
synchronizing objective \emph{with exact support}.  This objective
requires to ensure probability arbitrarily close to $1$ in the target
set $T$, and moreover that after the same number of
steps the support of the probability distribution is contained in the
given set $U$.  Formally, given an MDP $\M$, let
$\winlim{event}(\fsum_T,U)$ for $T \subseteq U$ be the set of all
initial distributions such that for all $\epsilon>0$ there exists a
strategy~$\alpha$ and $n \in \nat$ such that $\M^{\alpha}_n(T) \geq
1-\epsilon$ and $\M^{\alpha}_n(U)=1$.
We say that $\alpha$ is limit-sure eventually synchronizing in $T$ with support in $U$
(consider the example at the beginning of Section~\ref{sec:almost-eventually}
with $T = \{q_2\}$ and $U= \{q_{\init},q_2\}$).

We will present an algorithmic solution to limit-sure eventually
synchronizing objectives with exact support in Section~\ref{sec:limit-sure}.
Our characterization of the winning region for almost-sure winning is
as follows.

\begin{lemma}\label{lem: almost-limit-reduce-limit-event}
Let $\M$ be an MDP and $T$ be a target set. For all states $q_{\init}$, 
we have $q_{\init}\in \winas{event}(\fsum_T)$
if and only if there exists a set~$U$ of states such that:
  \begin{itemize}
    \item $q_{\init} \in \winsure{event}(\fsum_U)$, and \smallskip
    \item $d_U \in \winlim{event}(\fsum_{T},U)$ where $d_U$ is the uniform distribution over~$U$.
  \end{itemize}
\end{lemma}

\begin{proof}
First, if $q_{\init} \in \winas{event}(\fsum_T)$,
then there is a strategy $\alpha$ such that $\sup_{n \in \nat} \M^{\alpha}_n(T)=1$.
Then either $\M^{\alpha}_n(T) = 1$ for some $n \geq 0$,
or $\limsup_{n\to \infty} \M^{\alpha}_n(T)=1$. 
If $\M^{\alpha}_n(T) = 1$, then $q_{\init}$ is sure winning 
for eventually synchronizing in $T$, thus $q_{\init} \in \winsure{event}(\fsum_T)$
and we can take $U = T$. 
Otherwise,
for all $i>0$ there exists $n_i \in \nat$ such that 
$\M^{\alpha}_{n_i}(T) \geq 1-2^{-i}$, and moreover $n_{i+1} > n_i$ for all $i>0$.
Let $s_i = \Supp(\M^{\alpha}_{n_i})$ be the support of~$\M^{\alpha}_{n_i}$.
Since the state space is finite, there is a set~$U$ that 
occurs infinitely often in the sequence~$s_0 s_1 \dots$,
thus for all $k>0$ there exists $m_k \in \nat$ such that 
$\M^{\alpha}_{m_k}(T) \geq 1-2^{-k}$ and 
$\M^{\alpha}_{m_k}(U) = 1$.
It follows that $\alpha$ is sure eventually synchronizing in~$U$ 
from $q_{\init}$, hence $q_{\init} \in \winsure{event}(\fsum_U)$.
Moreover, $\M$ with initial distribution $d_1 = \M^{\alpha}_{m_1}$
is limit-sure eventually synchronizing in $T$ with exact support in $U$.
Since $\Supp(d_1) = U = \Supp(d_U)$, it follows 
by Corollary~\ref{col:uniform-dist-limit} that 
$d_U \in \winlim{event}(\fsum_{T},U)$.

To establish the converse, note that since $d_U \in
\winlim{event}(\fsum_{T},U)$, it follows from
Corollary~\ref{col:uniform-dist-limit} that from all initial
distributions with support in $U$, for all $\epsilon > 0$ there exists
a strategy $\alpha_{\epsilon}$ and a position $n_{\epsilon}$ such that
$\M^{\alpha_{\epsilon}}_{n_{\epsilon}}(T) \geq 1-\epsilon$ and
$\M^{\alpha_{\epsilon}}_{n_{\epsilon}}(U) = 1$.  We construct an
almost-sure limit eventually synchronizing strategy $\alpha$ as
follows.  Since $q_{\init} \in \winsure{event}(\fsum_U)$, play according to
a sure eventually synchronizing strategy from $q_{\init}$ until all the
probability mass is in $U$.  Then for $i=1,2, \dots$ and $\epsilon_i =
2^{-i}$, repeat the following procedure: given the current probability distribution, 
select the corresponding strategy $\alpha_{\epsilon_i}$
and play according to $\alpha_{\epsilon_i}$ for $n_{\epsilon_i}$
steps, ensuring probability mass at least $1-2^{-i}$ in $T$, and since
after that the support of the probability mass is again in $U$,
play according to $\alpha_{\epsilon_{i+1}}$ for $n_{\epsilon_{i+1}}$
steps, etc.  This strategy $\alpha$ ensures that $\sup_{n\in\nat}
\M^{\alpha}_n(T)=1$ from $q_{\init}$, hence $q_{\init}\in \winas{event}(\fsum_T)$.
Note that $\alpha$ is a pure strategy.
\qed
\end{proof}\medskip


As we show in Section~\ref{sec:limit-sure} that the membership problem
for limit-sure eventually synchronizing with exact support can be solved
in PSPACE, it follows from the characterization in Lemma~\ref{lem: almost-limit-reduce-limit-event} 
that the membership problem for almost-sure eventually synchronizing
is in PSPACE, using the following (N)PSPACE algorithm: 
guess the set~$U$, and check that $q_{\init} \in \winsure{event}(\fsum_U)$, and that 
$d_U \in \winlim{event}(\fsum_{T},U)$
where $d_U$ is the uniform distribution over~$U$ (this can be done
in PSPACE by Theorem~\ref{theo:sure-eventually-pspace-c}
and Theorem~\ref{theo:limit-sure-eventually}).
We present a matching lower bound.
 

\begin{lemma}\label{lem:as-event-pspace-hard}
The membership problem for $\winas{event}(\fsum_T)$ is PSPACE-hard
even if $T$ is a singleton.
\end{lemma}

\begin{proof}
We show the result by a reduction from the membership problem for sure eventually
synchronizing, which is PSPACE-complete by Theorem~\ref{theo:sure-eventually-pspace-c}.
Given an MDP $\M=\tuple{Q, \Act,\delta}$, an initial state $q_{\init} \in Q$,
and a state $\q \in Q$, we construct an MDP $\N=\tuple{Q',\Act',\delta'}$ with $Q \subseteq Q'$ 
and a state $\p \in Q'$ such that $q_{\init} \in \winsure{event}(\q)$ in $\M$
if and only if $q_{\init} \in \winas{event}(\p)$ in $\N$.
The MDP $\N$ is a copy
of $\M$ with two new states $\p$ and $\sink$ reachable only by a new action $\sharp$
(see \figurename~\ref{fig:almost-ss-reduction}).
Formally, $Q' = Q \cup \{\p, \sink\}$ and $\Act'= \Act \cup \{\sharp\}$,
and the transition function $\delta'$ is defined as follows, for all $q \in Q$
and $a \in \Act$:
$\delta'(q,a) = \delta(q,a)$, and 
$\delta'(q,\sharp)(\sink) = 1$ if $q \neq \q$, and $\delta'(\q,\sharp)(\p) = 1$;
finally, for all $a \in \Act'$, let $\delta'(\p,a)(\sink) = \delta'(\sink,a)(\sink) = 1$.

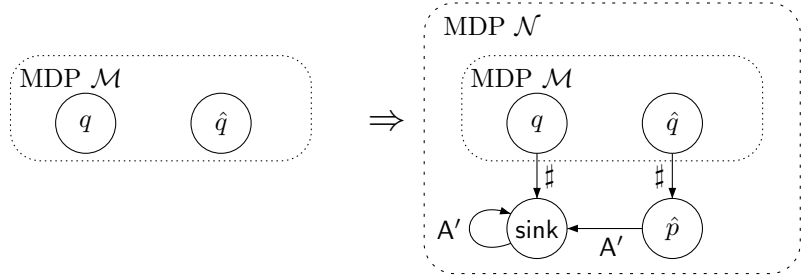
\begin{figure}[t]
\begin{center}
    \begin{picture}(105,36)(0,0)

\node[Nmarks=n, Nw=40, Nh=14, dash={0.2 0.5}0](m1)(20,22){}
\node[Nframe=n](label)(8,26){MDP $\M$}
\node[Nmarks=n](n1)(28,20){$\q$}
\node[Nmarks=n](n2)(10,20){$q$}
\node[Nframe=n](arrow)(50,20){{\Large $\Rightarrow$}}

\node[Nmarks=n, Nw=40, Nh=14, dash={0.2 0.5}0](nm1)(80,22){}
\node[Nmarks=n, Nw=50, Nh=36, dash={0.4 1}0](m2)(80,18){}
\node[Nframe=n](label)(64,33){MDP $\N$}

\node[Nframe=n](label)(68,26){MDP $\M$}
\node[Nmarks=n](n1)(88,20){$\q$}
\node[Nmarks=n](n2)(70,20){$q$}

\node[Nmarks=n](end)(70,6){$\sink$}
\node[Nmarks=n](qq)(88,6){$\p$} 

\drawloop[ELside=l,loopCW=y, loopangle=180, loopdiam=5](end){$\Act'$}

\drawedge[ELpos=50, ELside=l, curvedepth=0](n2,end){$\sharp$}
\drawedge[ELpos=50, ELside=r, curvedepth=0](n1,qq){$\sharp$}
\drawedge[ELpos=45](qq,end){$\Act'$}

\end{picture}
\end{center}
 \caption{Sketch of the reduction to show PSPACE-hardness of 
 the membership problem for almost-sure eventually synchronizing.}\label{fig:almost-ss-reduction}
\end{figure}

The goal is that $\N$ simulates $\M$ until the action $\sharp$ is played in $\q$
to move the probability mass from $\q$ to $\p$,
ensuring that if $\M$ is sure-winning for eventually synchronizing in $\q$, then
$\N$ is also sure-winning (and thus almost-sure winning) for eventually synchronizing in $\p$. Moreover,
the only way to be almost-sure eventually synchronizing in $\p$ is to have
probability~$1$ in $\p$ at some point, because the state $\p$ is transient
under all strategies, thus the probability mass cannot accumulate and tend to $1$
in $\p$ in the long run. Therefore (from all initial states
$q_{\init}$) $\M$ is sure-winning   for eventually synchronizing in $\q$ if and only if $\N$ is almost-sure winning
for eventually synchronizing in $\p$. It follows from this reduction that
the membership problem for almost-sure eventually synchronizing objective
is PSPACE-hard.
\qed
\end{proof}\medskip

\noindent The results of this section are summarized as follows.
 
\begin{theorem}\label{theo:almost-sure-eventually}
For almost-sure eventually synchronizing  in MDPs:

\begin{enumerate}
\item (Complexity). The membership problem is PSPACE-complete.

\item (Memory). Infinite memory is necessary in general for both pure 
and randomized strategies, and pure strategies are sufficient.
\end{enumerate}

\end{theorem}


\subsection{Limit-sure eventually synchronizing}  \label{sec:limit-sure}

In this section, we present the algorithmic solution for 
limit-sure eventually synchronizing with exact support,
which requires to get probability arbitrarily close to~$1$ in 
a target set~$T$ while all the probability mass is contained in a given set $U$.
Note that the limit-sure eventually synchronizing objective is
a special case where the support is the state space of the MDP.
Consider the MDP in \figurename~\ref{fig:almost-limit-eventually-differ}
which is limit-sure eventually synchronizing in $\{q_2\}$, 
as shown in Lemma~\ref{lem:dif-in-def}. For $i=0,1,\dots$, the sequence $\Pre^i(T)$ 
of predecessors of $T = \{q_2\}$ is ultimately periodic: 
$\Pre^0(T) = \{q_2\}$, and $\Pre^i(T) = \{q_1\}$ for all $i \geq 1$. 
Given $\epsilon > 0$, a strategy to get probability $1-\epsilon$ in $q_2$
first accumulates probability mass in the \emph{periodic} subsequence of predecessors (here $\{ q_1 \}$),
and when the probability mass is greater than $1-\epsilon$ in $q_1$, the
strategy injects the probability mass in $q_2$ (through the aperiodic prefix of
the sequence of predecessors). This is the typical shape of a limit-sure eventually 
synchronizing strategy. Note that in this scenario, the MDP is also limit-sure eventually synchronizing 
in every set $\Pre^i(T)$ of the sequence of predecessors.
A special case is when it is possible to get probability~$1$
in the sequence of predecessors after finitely many steps. 
In this case, the probability mass injected in $T$ is $1$ and the MDP is even sure-winning. 
The algorithm for deciding limit-sure eventually synchronizing relies
on the above characterization, generalized in Lemma~\ref{lem:lse-pre} 
to limit-sure eventually synchronizing with exact support, saying that 
limit-sure eventually synchronizing in~$T$ with support in~$U$ 
is equivalent to either 
sure eventually synchronizing in~$T$ (and therefore also in~$U$), or 
limit-sure eventually synchronizing in $\Pre^k(T)$ with support in $\Pre^k(U)$ (for arbitrary~$k$).
The intuition of the proof is that if an MDP is limit-sure eventually synchronizing 
in~$T$ with support in~$U$, then either a bounded number of steps is sufficient
to get probability $1-\epsilon$ in $T$ (and then we argue that the MDP is sure
eventually synchronizing), or unbounded number of steps is required,
which means that $k$ steps before getting probability $1-\epsilon$ in $T$, 
the probability mass in $\Pre^k(T)$ must also be close to $1$ (and arbitrarily 
close to $1$ as $\epsilon$ tends to $0$).

\begin{samepage}
\begin{lemma} \label{lem:lse-pre}
For all $T \subseteq U$ and all $k\geq 0$, we have 
$$\winlim{event}(\fsum_T,U) = \winsure{event}(\fsum_T) \cup \winlim{event}(\fsum_R,Z)$$
where $R = \Pre^{k}(T)$ and $Z = \Pre^{k}(U)$.  
\end{lemma}
\end{samepage}

\begin{proof}
We establish the equality in the lemma by showing inclusions in the two directions. 
First we show that 
$$\winsure{event}(\fsum_T) \cup \winlim{event}(\fsum_R,Z) \subseteq \winlim{event}(\fsum_T,U).$$
Since $T\subseteq U$, it follows from the definitions that 
$\winsure{event}(\fsum_T) \subseteq \winlim{event}(\fsum_T,U)$;
to show that $\winlim{event}(\fsum_R, Z) \subseteq \winlim{event}(\fsum_T, U)$
in an MDP $\M$,
let $\epsilon > 0 $ and
consider an initial distribution $d_0$ and a strategy $\alpha$ such that for 
some $i \geq 0$ we have $\M^{\alpha}_{i}(R)\geq 1-\epsilon$ and $\M^{\alpha}_i(Z) = 1$.
We construct a strategy~$\beta$ that plays like~$\alpha$ for the first~$i$
steps, and then since $R = \Pre^{k}(T)$ and $Z = \Pre^{k}(U)$ plays
from states in~$R$ according to a sure eventually synchronizing strategy with target~$T$,
and from states in $Z \setminus R$ according to a sure eventually synchronizing strategy with target~$U$
(such strategies exist by Lemma~\ref{lem:sure-ss-pre} since $R = \Pre^{k}(T)$).
The strategy $\beta$ ensures from $d_0$ that $\M^{\beta}_{i+k}(T) \geq 1-\epsilon$ 
and $\M^{\beta}_{i+k}(U) = 1$, showing that $\M$ is limit-sure eventually synchronizing 
in~$T$ with support in~$U$. 

Second we show the converse inclusion, namely that 
$$\winlim{event}(\fsum_T,U) \subseteq \winsure{event}(\fsum_T) \cup \winlim{event}(\fsum_R,Z).$$
Consider an initial distribution $d_0 \in \winlim{event}(\fsum_T,U)$ in the MDP $\M$
and for $\epsilon_i = \frac{1}{i}$ ($i \in \nat$) let $\alpha_i$ be a strategy 
and $n_i \in \nat$ such that $\M^{\alpha_i}_{n_i}(T)\geq 1-\epsilon_i$
 and $\M^{\alpha_i}_{n_i}(U) = 1$.
We consider two cases. 

\begin{itemize}
\item[$(a)$] If the set $\{n_i \mid i \geq 0\}$ is bounded,
then there exists a number $n$ that occurs \emph{infinitely often} in the sequence
$(n_i)_{i \in \nat}$. It follows that for all $i \geq 0$, there exists
a strategy $\beta_i$ such that $\M^{\beta_i}_{n}(T) \geq 1-\epsilon_i$ and $\M^{\beta_i}_{n}(U) = 1$.
Since $n$ is fixed, we can assume w.l.o.g. that the strategies $\beta_i$
are pure, and since there is a finite number of pure strategies
over paths of length at most $n$, it follows that there is a strategy $\beta$
that occurs infinitely often among the strategies $\beta_i$ and
such that for all $\epsilon > 0$ we have 
$\M^{\beta}_{n}(T) \geq 1-\epsilon$, hence $\M^{\beta}_{n}(T) = 1$,
showing that $\M$ is sure winning for eventually synchronizing in $T$, 
that is $d_0 \in \winsure{event}(\fsum_T)$.
\item[$(b)$] otherwise, the set $\{n_i \mid i \geq 0\}$ is unbounded
and we can assume w.l.o.g. that $n_i \geq k$ for all $i \geq 0$.
We claim that the family of strategies $\alpha_i$ ensures 
limit-sure eventually synchronizing in $R = \Pre^{k}(T)$ with 
support in $Z = \Pre^{k}(U)$. Essentially this is because if the probability
in $T$ is close to $1$ after $n_i$ steps, then $k$ steps before 
the probability in $\Pre^{k}(T)$ must be close to $1$ as well.
Formally, we show that $\alpha_i$ is such that 
$\M^{\alpha_i}_{n_i - k}(R)\geq 1- \frac{\epsilon_i}{\eta^k}$ and $\M^{\alpha_i}_{n_i - k}(Z) = 1$
where $\eta$ is the smallest positive probability in the transitions of $\M$.
Towards contradiction, assume that $\M^{\alpha_i}_{n_i - k}(R)< 1- \frac{\epsilon_i}{\eta^k}$.
Then $\M^{\alpha_i}_{n_i - k}(Q \setminus R) > \frac{\epsilon_i}{\eta^k}$ and from every
state $q \in Q \setminus R$, no matter which sequence of actions is played
by $\alpha_i$ for the next $k$ steps, there is a path from $q$ to a state
outside $T$ (by Lemma~\ref{lem:sure-ss-pre} since $R = \Pre^{k}(T)$), thus with probability at least $\eta^k$. 
Hence, the probability in $Q \setminus T$ after $n_i$ steps is greater than 
$\frac{\epsilon_i}{\eta^k} \cdot \eta^k$, and therefore 
$\M^{\alpha_i}_{n_i}(T)< 1-\epsilon_i$, in contradiction with the definition of $\alpha_i$.
This shows that $\M^{\alpha_i}_{n_i - k}(R)\geq 1- \frac{\epsilon_i}{\eta^k}$,
and an argument analogous to the proof of Lemma~\ref{lem:sure-ss-pre}
shows that $\M^{\alpha_i}_{n_i - k}(Z) = 1$.
It follows that $d_0 \in \winlim{event}(\fsum_R,Z)$ and the proof is complete.~\qed 
\end{itemize}
\end{proof}\medskip

Thanks to Lemma~\ref{lem:lse-pre}, since sure-winning is already
solved in Section~\ref{sec:sure-event-sync}, it suffices to solve the
limit-sure eventually synchronizing problem for target $R = \Pre^k(T)$
and support $Z = \Pre^k(U)$ with arbitrary $k$, instead of $T$ and
$U$. We choose $k$ such that both $\Pre^k(T)$ and $\Pre^k(U)$ lie in
the periodic part of the sequence of pairs of predecessors
$(\Pre^i(T),\Pre^i(U))$.  

Note that $\Pre^i(T) \subseteq \Pre^i(U) \subseteq Q$ for all $i \geq 0$,
and there are at most $3^{\abs{Q}}$ different pairs $(A,B)$ with $A \subseteq B \subseteq Q$
(each state $q \in Q$ belongs either to $A$, or to $B \setminus A$,
or to $Q \setminus B$). Hence, we can assume that $k \leq 3^{\abs{Q}}$.

For such value of $k$ the limit-sure problem is conceptually simpler:
once some probability is injected in $R = \Pre^k(T)$, it can loop
through the sequence of predecessors and visit $R$ infinitely often
(every $r$ steps, where $r \leq 3^{\abs{Q}}$ is the period of the
sequence of pairs of predecessors).
It follows that if a strategy ensures with probability $1$
that the set $R$ can be reached by finite paths whose lengths are 
congruent modulo $r$, then the whole probability mass can 
indeed synchronously accumulate in $R$ in the limit.

Therefore, limit-sure eventually synchronizing in $R$ 
reduces to standard limit-sure reachability (in the state-based semantics) with target set $R$ 
and the additional requirement that the numbers of steps at which the target set is reached 
be congruent modulo $r$.
In the case of limit-sure eventually synchronizing with support in~$Z$, we also need to 
ensure that no mass of probability leaves the sequence $\Pre^i(Z)$. 
In a state $q \in \Pre^i(Z)$, we say that an action $a \in \Act$ is \emph{$Z$-safe}
at position $i$ if\footnote{Since $\Pre^r(Z) = Z$ and $\Pre^r(R) = R$, 
we assume a modular arithmetic for exponents of $\Pre$, that is $\Pre^x(\cdot)$ 
is defined as $\Pre^{x \!\!\mod r}(\cdot)$. For example $\Pre^{-1}(Z)$ is $\Pre^{r-1}(Z)$.}
$\post(q,a)\subseteq \Pre^{i-1}(Z)$. In states 
$q \not\in \Pre^i(Z)$ there is no $Z$-safe action at position~$i$.

To encode the above requirements, we construct an MDP $\M_Z \times[r]$ 
that allows only $Z$-safe actions to be played (and then mimics the original
MDP), and tracks the position (modulo $r$) in the sequence of predecessors,
thus simply decrementing the position on each transition since all successors
of a state $q \in \Pre^i(Z)$ on a safe action are in $\Pre^{i-1}(Z)$.

Formally, if $\M = \tuple{Q, \Act, \delta}$ then 
$\M_Z \times[r] = \tuple{Q', \Act, \delta'}$ where 
\begin{itemize}
\item $Q' = Q \times \{r-1, \dots, 1, 0\} \cup \{\sink\}$; 
a state $\tuple{q,i}$ consisting of a state $q$ of $\M$ and 
a \emph{position}~$i$ in the predecessor sequence corresponds
to the promise that $q \in \Pre^i(Z)$;

\item $\delta'$ is defined as follows for all $\tuple{q,i} \in Q'$ and $a \in \Act$ 
      (assuming an arithmetic modulo $r$ on positions):
      if $a$ is a $Z$-safe action in $q$ at position $i$, 
      then $$\delta'(\tuple{q,i},a)(\tuple{q',i-1}) = \delta(q,a)(q'),$$
      otherwise $\delta'(\tuple{q,i},a)(\sink) = 1$ 
      (and $\sink$ is absorbing).
\end{itemize}

Note that the size of the MDP $\M_Z \times[r]$ is exponential in the size of $\M$
(since $r$ is at most $3^{\abs{Q}}$).

\begin{lemma}\label{lem:lssr-assr}
Let $\M$ be an MDP and $R \subseteq Z$ be two sets of states such that   
$\Pre^{r}(R)=R$ and $\Pre^{r}(Z)=Z$ where~$r>0$. 
Then a state $q_{\init}$ is limit-sure eventually synchronizing in $R$ with support in $Z$ 
($q_{\init} \in \winlim{event}(\fsum_R, Z)$) if and only if  
there exists $0 \leq t < r$ such that $\tuple{q_{\init},t}$ 
is limit-sure winning for the reachability objective 
$\Diamond (R \times \{0\})$ in the MDP $\M_Z \times [r]$.
\end{lemma}

\begin{proof}
For the first direction of the lemma, assume
that $q_{\init}$ is limit-sure eventually synchronizing in $R$ with support in
$Z$, and for $\epsilon >0$ let $\beta$ be a strategy such that
$\M^{\beta}_{k}(Z)=1$ and $\M^{\beta}_k(R) \geq 1-\epsilon$ for some
number $k$ of steps.  Let $0 \leq t \leq r$ such that $t=k \mod r$.
Let $R_0 = R \times \{0\}$.
We show that from initial state $(q_{\init}, t)$ the strategy $\alpha$ in $\M_Z \times [r]$ that
mimics (copies) the strategy $\beta$ is limit-sure winning for the
reachability objective $\Diamond R_0$: it follows from
Lemma~\ref{lem:sure-ss-pre} that $\alpha$ plays only $Z$-safe actions,
and since $Pr^{\alpha}(\Diamond R_0)\geq Pr^{\alpha}(\Diamond^k R_0) =
\M^{\beta}_k(R)\geq 1-\epsilon$, the result follows. 

For the converse direction, assuming
that there exists $0 \leq t < r$ such that $\tuple{q_{\init},t}$ is
limit-sure winning for the reachability objective $\Diamond R_0$ in
$\M_Z\times[r]$, we show that $q_{\init}$ is limit-sure synchronizing in
target set $R$ with exact support in $Z$.  
Since the winning region of limit-sure and almost-sure reachability coincide 
for MDPs~\cite{AHK07}, there exists a (pure) strategy $\alpha$ in 
$\M_Z \times [r]$ with initial state $\tuple{q_{\init},t}$ such that
$\Pr^{\alpha}(\Diamond R_0) = 1$.

Given $\epsilon>0$, we construct from $\alpha$ a pure strategy $\beta$ in $\M$ 
that is $(1-\epsilon)$-synchronizing in $R$ with support in~$Z$. 
Given a finite path $\rho = q_0 a_0 q_1 a_1 \dots q_n$ in $\M$ (with $q_0 = q_{\init}$), there 
is a corresponding path $\rho' = \tuple{q_0,k_0} a_0 \tuple{q_1,k_1} a_1 \dots \tuple{q_n, k_n}$ in 
$\M_Z \times [r]$ where $k_0 = t$ and $k_{i+1} = k_i - 1$ for all $i \geq 0$.
Since the sequence $k_0, k_1, \dots$ is uniquely determined from $\rho$, 
there is a clear bijection between the paths in $\M$ starting in $q_{\init}$
and the paths in $\M_Z \times [r]$ starting in $\tuple{q_{\init},t}$.
In the sequel, we freely omit to apply and mention this bijection.
Define the strategy $\beta$ as follows:
if $q_n \in \Pre^{k_n}(R)$, then there exists an action $a$ such that 
$\post(q_n,a) \subseteq \Pre^{k_n-1}(R)$ and we define $\beta(\rho) = a$, 
otherwise let $\beta(\rho) = \alpha(\rho')$.
Thus $\beta$ mimics $\alpha$ (thus playing only $Z$-safe actions) unless a 
state $q$ is reached at step $n$ such that $q \in \Pre^{t-n}(R)$, 
and then $\beta$ switches to always playing actions
that are $R$-safe (and thus also $Z$-safe since $R \subseteq Z$).
We now prove that $\beta$ is limit-sure eventually synchronizing in
target set $R$ with support in $Z$. First since $\beta$ plays only
$Z$-safe actions, it follows for all $k$ such that $t-k = 0$ (modulo $r$), all
states reached from $q_{\init}$ with positive probability after $k$ steps are in $Z$. 
Hence, $\M^{\beta}_{k}(Z)=1$ for all such $k$. 
Second, we show that given $\epsilon>0$ there exists $k$ such that
$t-k = 0$ and $\M^{\beta}_{k}(R) \geq 1-\epsilon$, thus also
$\M^{\beta}_{k}(Z)=1$ and $\beta$ is limit-sure eventually
synchronizing in target set $R$ with support in~$Z$.  To show this,
recall that $\Pr^{\alpha}(\Diamond R_0) = 1$, and therefore
$\Pr^{\alpha}(\Diamond^{\leq k} R_0) \geq 1-\epsilon$ for all
sufficiently large $k$. Without loss of generality, consider such a
$k$ satisfying $t-k = 0$ (modulo $r$). For $i = 1, \dots, r-1$, let
$R_i = \Pre^i(R) \times \{i\}$.  Then trivially
$\Pr^{\alpha}(\Diamond^{\leq k} \bigcup_{i=0}^{r} R_i) \geq
1-\epsilon$ and since $\beta$ agrees with $\alpha$ on all finite paths
that do not (yet) visit $\bigcup_{i=0}^{r} R_i$, given a path $\rho$
that visits $\bigcup_{i=0}^{r} R_i$ (for the first time), only
$R$-safe actions will be played by $\beta$ and thus all continuations
of $\rho$ in the outcome of $\beta$ will visit $R$ after $k$ steps (in
total).  It follows that $\Pr^{\beta}(\Diamond^{k} R_0) \geq
1-\epsilon$, that is $\M^{\beta}_k(R) \geq 1-\epsilon$. Note that we
used the same pure strategy $\beta$ for all $\epsilon >0$ and thus $\beta$
is also almost-sure eventually synchronizing in $R$.
\qed
\end{proof}\medskip

From the proof of Lemma~\ref{lem:lssr-assr} (last sentence), it follows that if the MDP $\M$
is limit-sure eventually synchronizing in $R$ with support in $Z$,
then $\M$ is also almost-sure eventually synchronizing in $R$. 
Since almost-sure synchronization implies limit-sure synchronization by
definition, the two notions coincide in this case.

\begin{corollary}\label{cor:as-same-as-ls-for-R}
Given $R \subseteq Z$ two sets of states in an MDP such that   
$\Pre^{r}(R)=R$ and $\Pre^{r}(Z)=Z$ where~$r>0$, we have
$\winlim{event}(\fsum_R, Z) = \winas{event}(\fsum_R)$.
\end{corollary}


Since deciding limit-sure reachability is PTIME-complete, it follows from
Lemma~\ref{lem:lssr-assr} that limit-sure eventually synchronizing (with exact support)
can be decided in EXPTIME. 

We show in Lemma~\ref{lem:pspace-mi} that the problem can be solved in PSPACE
by exploiting the special structure of the exponential MDP used in Lemma~\ref{lem:lssr-assr}.
We conclude this section by Lemma~\ref{lem:limit-event-pspace-hard} showing that limit-sure eventually synchronizing 
with exact support is PSPACE-complete (even in the special case where the 
support is the whole state space).

\begin{lemma}\label{lem:pspace-mi}
The membership problem for limit-sure eventually synchronizing with exact 
support is in PSPACE.
\end{lemma}

\begin{proof}
We present a (nondeterministic) PSPACE algorithm to decide, 
given an MDP $\M=\tuple{Q,\Act,\delta}$, a state $q_{\init}$, and two sets $T \subseteq U$, 
whether $q_{\init}$ is limit-sure eventually synchronizing in~$T$ with support in~$U$. 

First, the algorithm computes numbers $k\geq 0$ and $r>0$ such that
for $R = \Pre^k(T)$ and $Z = \Pre^k(U)$ we have $\Pre^{r}(R)=R$ and
$\Pre^{r}(Z)=Z$. As discussed before, this can be done by guessing
$k,r \leq 3^{\abs{Q}}$.  By Lemma~\ref{lem:lse-pre}, we have
$$\winlim{event}(\fsum_T,U) = \winlim{event}(\fsum_R,Z) \cup \winsure{event}(\fsum_T),$$
and since sure eventually synchronizing in $T$ can be decided in PSPACE (by
Theorem~\ref{theo:sure-eventually-pspace-c}), it suffices to decide
limit-sure eventually synchronizing in $R$ with support in $Z$ in
PSPACE.  According to Lemma~\ref{lem:lssr-assr}, it is therefore
sufficient to show that deciding limit-sure winning for the (standard)
reachability objective $\Diamond (R \times \{0\})$ in the MDP $\M_Z
\times [r]$ can be done in polynomial space.  As we cannot afford to
construct the exponential-size MDP $\M_Z \times [r]$, the algorithm
relies on the following characterization of the limit-sure winning set
for reachability objectives in MDPs. It is known that the winning
region for limit-sure and almost-sure reachability
coincide~\cite{AHK07}, and pure memoryless strategies are sufficient.
Therefore, we can see that the almost-sure winning set $W$ for the
reachability objective $\Diamond (R \times \{0\})$ satisfies the
following property: 
there exists a
memoryless strategy $\alpha: W \to \Act$ such that:
\begin{itemize}
\item[$(1)$] $W$ is closed, that is $\post(q,\alpha(q)) \subseteq W$ for all $q \in W$, and 
\item[$(2)$] in the graph of the Markov chain $M(\alpha)$, for every
state $q \in W$, there is a path (of length at most $\abs{W}$) from
$q$ to some state in $R \times \{0\}$.
\end{itemize}
This property ensures that from every state in $W$, the target set $R
\times \{0\}$ is reached within $\abs{W}$ steps with positive (and
bounded) probability, and since $W$ is closed it ensures that $R
\times \{0\}$ is reached with probability~$1$ in the long run.  Thus
any set $W$ satisfying the above property is almost-sure winning.

Our algorithm will guess and explore on the fly a set $W$ to ensure
that it satisfies this property, and contains the state
$\tuple{q_{\init},t}$ for some $t < r$.  As we cannot afford to explicitly
guess $W$ (remember that $W$ could be of exponential size), we
decompose $W$ into \emph{slices} $W_0, W_1,\dots$ such that $W_i
\subseteq Q$ and $W_i \times \{-i \mod r\} = W \cap (Q \times \{-i
\mod r\})$.  We start by guessing $W_0$, and we use the property 
that in $\M_Z \times [r]$, from a state
$(q,j)$ under all $Z$-safe actions, all successors are of the form
$(\cdot, j-1)$.  It follows that the successors of the states in~$W_i
\times \{-i\}$ should lie in the slice $W_{i+1} \times \{-i-1\}$, and
we can guess on the fly the next slice $W_{i+1} \subseteq Q$ by
guessing for each state $q$ in a slice $W_i$ an action $a_q$ such that
$\bigcup_{q \in W_i} \post(q,a_q) \subseteq W_{i+1}$.  Moreover, we
need to check the existence of a path from every state in $W$ to $R
\times \{0\}$. As $W$ is closed, it is
sufficient to check that there is a path from every state in $W_0
\times \{0\}$ to $R \times \{0\}$.  To do this we guess along with the
slices $W_0, W_1, \dots$ a sequence of sets $P_0, P_1, \dots$ where
$P_i \subseteq W_i$ contains the states of slice $W_i$ that belong to
the guessed paths. Formally, $P_0 = W_0$, and for all $i\geq 0$, the
set $P_{i+1}$ is such that $\post(q,a_q) \cap P_{i+1} \neq \emptyset$
for all $q \in P'_i$ (where $P'_i = P_i \setminus R$ if $i$ is a
multiple of $r$, and $P'_i = P_i$ otherwise), that is $P_{i+1}$
contains a successor of every state in $P_i$ that is not already in
the target $R$ (at position $0$ modulo $r$).

We need polynomial space to store the first slice $W_0$, 
the current slice $W_i$ and the set $P_i$, and the value of $i$ (in binary).
As $\M_Z \times [r]$ has $\abs{Q}\cdot r$ states, the algorithm runs 
for $\abs{Q}\cdot r$ iterations and then checks that:
\begin{itemize}
\item[$(1)$] $W_{\abs{Q}\cdot r} \subseteq W_0$ to ensure that $W = \bigcup_{i \leq \abs{Q}\cdot r} W_i\times\{i \mod r\}$ is closed, 
\item[$(2)$] $P_{\abs{Q}\cdot r} = \emptyset$ showing that from every state in $W_0 \times \{0\}$ 
there is a path to $R \times \{0\}$ (and thus also from all states in $W$), and 
\item[$(3)$] the state $q_{\init}$ occurs in some slice $W_i$. 
\end{itemize}
The correctness of the algorithm follows from the characterization of the almost-sure 
winning set for reachability in MDPs: if some state $\tuple{q_{\init},t}$ is limit-sure winning, then 
the algorithm accepts by guessing (slice by slice) the almost-sure winning set $W$ 
and the paths from $W_0 \times \{0\}$ to $R \times \{0\}$ (at position $0$ modulo $r$), 
and otherwise any set (and paths) correctly 
guessed by the algorithm would not contain $q_{\init}$ in any slice.
\qed
\end{proof}\medskip

It follows from the proof of Lemma~\ref{lem:lssr-assr} that
all winning modes for eventually synchronizing
are independent of the numerical value of the positive transition probabilities.

\begin{corollary}\label{col:uniform-dist-limit}
Let $\mu\in \{sure, almost, limit\}$ and $T\subseteq U$ be two sets. 
For two  distributions $d,d' $ with $\Supp(d)=\Supp(d')$, we have 
$d \in \win{event}{\mu}(\fsum_T,U)$~if and only if~$d' \in \win{event}{\mu}(\fsum_T,U)$.
\end{corollary}

\begin{remark}\label{rmk:limit-sure-event-fixed-dist}
Corollary~\ref{col:uniform-dist-limit} ensures that knowing the support of the 
initial distribution is sufficient to establish that it is eventually 
synchronizing. However, this corollary should be used carefully in the case
of limit-sure eventually synchronizing: given a support \mbox{$S \subseteq Q$},
if for all $\epsilon > 0$ there exists a distribution $d_\epsilon$ with support $S$
that is eventually $(1-\epsilon)$-synchronizing, this does not imply that
the distributions with support $S$ are limit-sure eventually synchronizing.

For example, consider an MDP with set of states $Q= \{q_1,q_2\}$,
self-loops on both $q_1$ and $q_2$, and target set $T=\{q_1\}$ (with function $\fsum_T$).
For $\epsilon > 0$, the initial distribution $d$ defined by $d(q_1) = 1-\epsilon$ and $d(q_2) = \epsilon$
has support $S = Q$ and ensures probability $1-\epsilon$ in $T$. 
Thus for all $\epsilon > 0$, we have an initial distribution that
satisfies the requirement, but the uniform distribution over $S$ is obviously not 
limit-sure eventually synchronizing in $T$.
\end{remark}

To establish the PSPACE-hardness for limit-sure eventually synchronizing in MDPs, 
we use a reduction from the universal finiteness problem for 1L-AFAs.

\begin{figure}[t]
\begin{center}
\begin{picture}(105,55)(0,0)

\node[Nmarks=n, Nw=40, Nh=22, dash={0.2 0.5}0](m1)(20,12){}
\node[Nframe=n](label)(8,5){MDP $\M$}
\drawpolygon[dash={0.8 0.5}0](38,22)(23,22)(23,2)(38,2)
\node[Nframe=n](label)(30,5){$T\subseteq Q$}
\node[Nmarks=r](n1)(30,12){$q_2$}

\node[Nframe=n](label)(19,12){$\dots$}
\node[Nmarks=n](n2)(10,12){$q_1$}
\node[Nframe=n](arrow)(45,12){{\Large $\Rightarrow$}}

\node[Nmarks=n, Nw=40, Nh=22, dash={0.2 0.5}0](nm1)(80,12){}
\node[Nmarks=n, Nw=50, Nh=48, dash={0.4 1}0](m2)(80,24){}
\node[Nframe=n](label)(65,52){MDP $\N$}

\node[Nframe=n](label)(68,5){MDP $\M$}
\drawpolygon[dash={0.8 0.5}0](98,22)(83,22)(83,2)(98,2)
\node[Nframe=n](label)(90,5){$T\subseteq Q$}
\node[Nmarks=r](nn1)(90,12){$q_2$}
\node[Nframe=n](label)(79,12){$\dots$}
\node[Nmarks=n](nn2)(70,12){$q_1$}

\node[Nmarks=i](qq)(80,35){$q_{\init}$} 

\drawloop[ELside=l,loopCW=y, loopangle=90, loopdiam=5](qq){$\Act$}

\drawedge[ELpos=50, ELside=l](qq,nn1){$\sharp$}
\node[Nframe=n](label)(80.25,28){$\cdots$}
\drawedge[ELpos=50, ELside=r](qq,nn2){$\sharp$}

\drawedge[ELpos=50, ELside=r, curvedepth=-8](nn1,qq){$\sharp$}

\drawedge[ELpos=50, ELside=l, curvedepth=+8](nn2,qq){$\sharp$}

\end{picture}
\caption{Sketch of the reduction to show PSPACE-hardness of the membership problem
for limit-sure eventually and almost-sure weakly synchronizing.}\label{fig:lim-sure-reduction}
\end{center}
\end{figure}
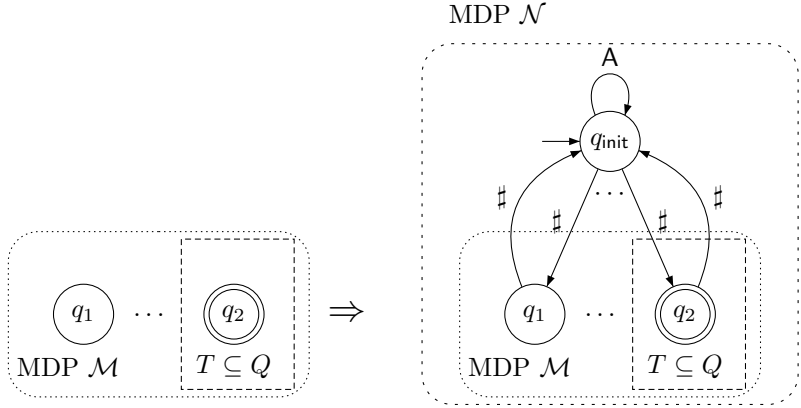

\begin{lemma}\label{lem:limit-event-pspace-hard}
The membership problem for $\winlim{event}(\fsum_T)$ is PSPACE-hard
even if $T$ is a singleton.
\end{lemma}

\begin{proof}
We show the result by a reduction from the universal finiteness problem for one-letter
alternating automata (1L-AFA), which is PSPACE-complete (by Lemma~\ref{lem:universal-finiteness-pspace-hard}).
It is easy to see that this problem remains PSPACE-complete even if the set $T$ of accepting
states of the 1L-AFA is a singleton, and 
given the tight relation between 1L-AFA and MDP (see Section~\ref{sec:1L-AFA}),
it follows from the definition of the universal finiteness problem that 
deciding, in an MDP $\M$, whether the sequence $\Pre^n_{\M}(T) \neq \emptyset$ 
for all $n \geq 0$ is PSPACE-complete.

The reduction is as follows (see also \figurename~\ref{fig:lim-sure-reduction}). 
Given an MDP $\M=\tuple{Q, \Act,\delta}$ and a singleton $T \subseteq Q$,
we construct an MDP $\N =\tuple{Q', \Act',\delta'}$ with state space $Q' = Q \uplus \{q_{\init}\}$ such that 
$\Pre^n_{\M}(T) \neq \emptyset$ for all $n \geq 0$ if and only if 
$q_{\init}$ is limit-sure eventually synchronizing in $T$.
The MDP $\N$ is essentially a copy of $\M$ with alphabet $\Act \uplus \{\sharp\}$ and 
the transition function on action $\sharp$ is the uniform 
distribution on $Q$ from $q_{\init}$, and the Dirac distribution on 
$q_{\init}$ from the other states $q \in Q$. There are self-loops on $q_{\init}$ 
for all other actions $a \in \Act$. Formally, the
transition function $\delta'$ is defined as follows, for all $q \in Q$:

\begin{itemize} 
\item $\delta'(q,a) = \delta(q,a)$ for all $a \in \Act$ (copy of $\M$),
      and $\delta'(q,\sharp)(q_{\init}) = 1$;
\item $\delta'(q_{\init},a)(q_{\init}) = 1$ for all $a \in \Act$,
      and $\delta'(q_{\init}, \sharp)(q) = \frac{1}{\abs{Q}}$.
\end{itemize}

We establish the correctness of the reduction as follows.
For the first direction, assume that 
$\Pre^n_{\M}(T) \neq \emptyset$ for all~$n \geq 0$. 
Then since $\N$ embeds a copy of $\M$ it follows that 
$\Pre^n_{\N}(T) \neq \emptyset$ for all~$n \geq 0$ and
there exist numbers $k_0,r \leq 2^{\abs{Q}}$ 
such that $\Pre^{k_0 + r}_{\N}(T) = \Pre^{k_0}_{\N}(T) \neq \emptyset$.
Using Lemma~\ref{lem:lse-pre} with $k = k_0$ and $R = \Pre^{k_0}_{\N}(T)$ 
(and $U = Z = Q'$ is the trivial support), it is sufficient to
prove that $q_{\init} \in \winlim{event}(R)$ to get $q_{\init} \in \winlim{event}(T)$
(in $\N$).
We show the stronger statement that $q_{\init}$ is actually almost-sure 
eventually synchronizing in $R$ with the pure strategy $\alpha$ defined as follows,
for all play prefixes $\rho$ (let $m = \abs{\rho} \!\!\mod r$):

\begin{itemize} 
\item if $\Last(\rho) = q_{\init}$, then $\alpha(\rho) = \sharp$;
\item if $\Last(\rho) = q \in Q$, then 
	\begin{itemize} 
	\item if $q \in \Pre^{r-m}_{\N}(R)$, then $\alpha(\rho)$ plays a $R$-safe action at position $r-m$;
	\item otherwise, $\alpha(\rho) = \sharp$.
	\end{itemize}
\end{itemize}

The strategy $\alpha$ ensures that the probability mass that is not (yet) in 
the sequence of predecessors $\Pre^n_{\N}(R)$ goes to $q_{\init}$, where by playing $\sharp$
at least a fraction $\frac{1}{\abs{Q}}$ of it would reach the sequence of predecessors
(at a synchronized position). It follows that after $2i$ steps, the probability
mass in $q_{\init}$ is $(1 - \frac{1}{\abs{Q}})^i$ and the probability mass
in the sequence of predecessors is $1 - (1 - \frac{1}{\abs{Q}})^i$. For $i \to \infty$,
the probability in the sequence of predecessors tends to $1$ and since 
$\Pre^n_{\N}(R) = R$ for all positions $n$ that are a multiple of $r$,
we get $\sup_{n}\M^{\alpha}_{n}(R) = 1$ and $q_{\init} \in \winas{event}(R)$.

For the converse direction, assume that $q_{\init} \in \winlim{event}(T)$
is limit-sure eventually synchronizing in $T$. By Lemma~\ref{lem:lse-pre},
either $(1)$ $q_{\init}$ is limit-sure eventually synchronizing in $\Pre^n_{\N}(T)$
for all $n \geq 0$, and then it follows that $\Pre^n_{\N}(T) \neq \emptyset$ 
for all $n \geq 0$, or $(2)$ $q_{\init}$ is sure eventually synchronizing in $T$,
and then since only the action $\sharp$ leaves the state $q_{\init}$ (and $\post(q_{\init},\sharp) = Q$), 
the characterization of sure eventually synchronizing in Lemma~\ref{lem:sure-ss-pre} shows that 
$Q \subseteq \Pre^k_{\N}(T)$ for some $k \geq 0$, and since 
$Q \subseteq \Pre_{\N}(Q)$ and $\Pre_{\N}(\cdot)$ is a monotone operator, 
it follows that $Q \subseteq \Pre^n_{\N}(T)$ for all $n \geq k$ and thus 
$\Pre^n_{\N}(T) \neq \emptyset$ for all $n \geq 0$.
We conclude the proof by noting that $\Pre^n_{\M}(T) = \Pre^n_{\N}(T) \cap Q$
and therefore $\Pre^n_{\M}(T) \neq \emptyset$ for all $n \geq 0$.
\qed
\end{proof}\medskip

The example in the proof of Lemma~\ref{lem:inf-mmeory-almost-event} can be used
to show that the memory needed by a family of strategies to win 
limit-sure eventually synchronizing objective (in target $T= \{q_2\}$) 
is unbounded.

Observe that  given~$\epsilon>0$, the required memory to accumulate~$1-\epsilon$ in~$T$ is finite, 
but the memory size increases and cannot be bounded as $\epsilon$ tends to $0$. 

The following theorem summarizes the results for limit-sure eventually synchronizing. 

\begin{theorem}\label{theo:limit-sure-eventually}
For limit-sure eventually synchronizing (with or without exact support) in MDPs:

\begin{enumerate}
\item (Complexity). The membership problem is PSPACE-complete.

\item (Memory). Unbounded memory is required for both pure 
and randomized strategies, and pure strategies are sufficient.
\end{enumerate}

\end{theorem}

\section{Weakly Synchronizing}

We establish the complexity and memory requirement for weakly synchronizing 
objectives. We show that the membership problem is PSPACE-complete for 
sure and almost-sure winning, that exponential memory is necessary and sufficient
for sure winning while infinite memory is necessary for almost-sure winning,
and we show that limit-sure and almost-sure winning coincide. 
By Lemma~\ref{lem:weakly-max-sum}, the complexity results established in this 
section for function $\fsum_T$ hold for function $\max_T$ as well.

The weakly synchronizing objective is reminiscent of a B\"uchi 
objective in the distribution-based semantics: it requires that 
in the sequence of distributions of an MDP $\M$ under strategy $\alpha$
we have $\limsup_{n \to \infty} \M^{\alpha}_n(T) = 1$ (and that $\M^{\alpha}_n(T) = 1$
for infinitely many $n$ in the case of sure winning).

The sure winning mode can be solved by a technique similar to the search for a
lasso in B\"uchi automata~\cite{Vardi07} (Section~\ref{sec:sure-weakly-sync}). We show that the almost-sure winning 
mode can be solved by a reduction analogous to the case of eventually
synchronizing (Section~\ref{sec:almost-weakly}). For the limit-sure winning mode, we show that it coincides
with the almost-sure winning mode. The proof of this result is technical
and requires a careful characterization of the limit-sure winning mode. 
We present examples to provide intuitive illustration of the proof (Section~\ref{sec:limit-sure-weakly}).

\subsection{Sure weakly synchronizing} \label{sec:sure-weakly-sync}
The PSPACE upper bound of the membership problem for sure weakly synchronizing is
obtained by the following characterization.

\begin{lemma}\label{lem: suWS}
Let $\M$ be an MDP and $T$ be a target set. For all states $q_{\init}$, 
we have \vspace{-2pt}
$$
\begin{array}{c}
q_{\init} \in \winsure{weakly}(\fsum_T) \text{ if and only if there exists a set $S \subseteq T$ such that} \\[2pt]
\text{$q_{\init} \in \Pre^{m}(S)$ for some $m \geq 0$ and $S \subseteq \Pre^{n}(S)$ for some $n \geq 1$.}
\end{array}
$$
\end{lemma}

\begin{proof}
First, if $q_{\init} \in \winsure{weakly}(\fsum_T)$, then let $\alpha$ be
a sure winning weakly synchronizing strategy.
Then there are infinitely many positions $n$ such that
$\M^{\alpha}_n(T)=1$, and since the state space is finite, 
there is a set $S$ of states 
such that for infinitely many positions~$n$ we have 
$\Supp(\M^{\alpha}_n) = S$ and $\M^{\alpha}_n(T)=1$, and thus $S \subseteq T$.
By the characterization of sure eventually synchronizing in Lemma~\ref{lem:sure-ss-pre}, it follows that 
$q_{\init} \in \Pre^{m}(S)$ for some $m \geq 0$, and 
by considering two positions $n_1 < n_2$ 
where $\Supp(\M^{\alpha}_{n_1}) = \Supp(\M^{\alpha}_{n_2}) = S$,
it follows that $S \subseteq \Pre^{n}(S)$ for $n  = n_2 - n_1 \geq 1$.

The reverse direction is straightforward by considering
a strategy $\alpha$ that ensures $\M^{\alpha}_m(S)=1$
for some $m \geq 0$, and then ensures that the probability mass 
from all states in $S$ remains in $S$ after every multiple of 
$n$ steps where $n>0$ is such that $S \subseteq \Pre^{n}(S)$,
showing that $\alpha$ is a sure winning weakly synchronizing strategy
in $S$ (and thus in $T$) from $q_{\init}$, thus $q_{\init} \in \winsure{weakly}(\fsum_T)$.
Note that $\alpha$ is a pure strategy.
\qed
\end{proof}\medskip

The PSPACE upper bound follows from the characterization in 
Lemma~\ref{lem: suWS}. A (N)PSPACE algorithm is to guess the set $S \subseteq T$,
and the numbers $m,n$ (with $m,n \leq 2^{\abs{Q}}$ since the sequence $\Pre^{n}(S)$ of 
predecessors is ultimately periodic), and check that $q_{\init} \in \Pre^{m}(S)$ and
$S \subseteq \Pre^{n}(S)$.
We present a matching PSPACE lower bound in the following lemma.

\begin{lemma}\label{lem: sure-weakly-psapce-hradness}
The membership problem for $\winsure{weakly}(\fsum_T)$ is PSPACE-hard
even if $T$ is a singleton.
\end{lemma}

\begin{proof}
We show the result by a reduction from the membership problem for
$\winsure{event}(\fsum_T)$ with a singleton~$T$, which is 
PSPACE-complete (Theorem~\ref{theo:sure-eventually-pspace-c}).
From an MDP~$\M=\tuple{Q,\Act,\delta}$ with initial state~$q_{\init}$
and target state~$\q$, 
we construct another MDP~$\N=\tuple{Q',\Act',\delta'}$ and a target state~$\p$ 
such that $q_{\init}\in \winsure{event}(\q)$ 
in~$\M$ if and only if $q_{\init}\in \winsure{weakly}(\p)$ in~$\N$.

The MDP~$\N$ is a copy of $\M$ with two new states~$\p$ and $\sink$ 
that are reachable only by a new action~$\sharp$ (see \figurename~\ref{fig:sure-ws-reduction}). 
Formally, $Q' = Q \cup \{\p, \sink\}$ and $\Act' = \Act \cup \{\sharp\}$. 
The transition function $\delta'$ is defined as follows:
$\delta'(q,a) = \delta(q,a)$ for all states $q \in Q$ and $a \in \Act$,
$\delta(q,\sharp)(\sink) = 1$ for all $q \in Q' \setminus \{\q\}$ and $\delta(\q,\sharp)(\p) = 1$.
The state $\sink$ is absorbing and from state~$\p$ all other transitions lead to 
the initial state, i.e. $\delta(\sink,a)(\sink)=1$ and 
$\delta(\p,a)(q_{\init})=1$ for all $a\in \Act$.

\begin{exclude}
\begin{figure}[t]
\begin{center}
    \begin{picture}(114,40)(0,2)

\node[Nmarks=n, Nw=40, Nh=22, dash={0.2 0.5}0](m1)(20,24){}
\node[Nframe=n](label)(8,31){MDP $\M$}

\node[Nmarks=i](n1)(10,22){$q_{\init}$}
\node[Nmarks=r](n1)(34,22){$\q$}
\node[Nmarks=n](n2)(22,22){$q$}
\node[Nframe=n](arrow)(50,24){{\Large $\Rightarrow$}}

\node[Nmarks=n, Nw=40, Nh=22, dash={0.2 0.5}0](nm1)(80,24){}
\node[Nmarks=n, Nw=57, Nh=40, dash={0.4 1}0](m2)(86,20){}
\node[Nframe=n](label)(67,37){MDP $\N$}

\node[Nframe=n](label)(68,31){MDP $\M$}
\node[Nmarks=i](n0)(70,22){$q_{\init}$}
\node[Nmarks=r](n1)(94,22){$\q$}
\node[Nmarks=n](n2)(82,22){$q$}

\node[Nmarks=n](end)(77,5){$\sink$}
\node[Nmarks=n](qq)(108,22){$\p$} 

\drawloop[ELside=l,loopCW=y, loopangle=180, loopdiam=5](end){$\Act'$}

\drawedge[ELpos=40, ELside=l, curvedepth=0](n0,end){$\sharp$}
\drawedge[ELpos=40, ELside=l, curvedepth=0](n2,end){$\sharp$}
\drawedge[ELpos=50, ELside=r, curvedepth=0](n1,qq){$\sharp$}
\drawedge[ELpos=50, ELside=r, curvedepth=-8](qq,n0){$\Act$}
\drawedge[ELpos=50, ELside=l, curvedepth=7](qq,end){$\sharp$}

\end{picture}
\end{center}
 \caption{The reduction sketch to show PSPACE-hardness of the membership problem 
for sure weakly synchronizing in MDPs.}
\label{fig:sure-ws-reduction}
\end{figure}
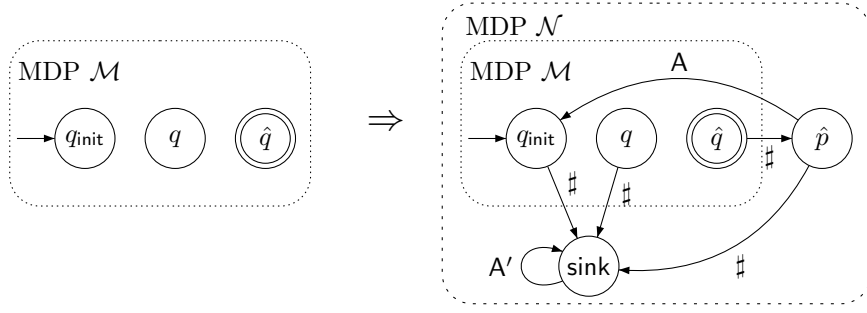
\end{exclude}

We establish the correctness of the reduction as follows.
First, if $q_{\init} \in \winsure{event}(\q)$ in~$\M$, then let $\alpha$
be a sure winning strategy in $\M$ for eventually synchronizing in $\{\q\}$.
A sure winning strategy in $\N$ for weakly synchronizing in $\{\p\}$ is 
to play according to $\alpha$ until the whole probability mass is in $\q$,
then play~$\sharp$ followed by some $a \in \A$ to visit $\p$ and get back to the initial state $q_{\init}$,
and then repeat the same strategy from $q_{\init}$. Hence, $q_{\init} \in \winsure{weakly}(\p)$
in $\N$.

Second, if $q_{\init} \in \winsure{weakly}(\p)$ in~$\N$, then consider 
a strategy $\alpha$ such that $\N^{\alpha}_n(\p) = 1$ for some $n \geq 0$.
By construction of $\N$, it follows that $\N^{\alpha}_{n-1}(\q)=1$,
that is all path-outcomes of $\alpha$ of length $n-1$ reach $\q$, and
$\alpha$ plays $\sharp$ in the next step. 
If $\alpha$ never plays $\sharp$ before position $n-1$, then $\alpha$
is a valid strategy in $\M$ up to step $n-1$ and it shows that 
$q_{\init} \in \winsure{event}(\q)$ is sure winning in $\M$ for eventually synchronizing in $\{\q\}$.
Otherwise let $m$ be the largest number such that 
there is a finite path-outcome $\rho$ of $\alpha$ of length $m < n-1$ with 
$\sharp \in \Supp(\alpha(\rho))$. %
Thus between position $m$ and $n-1$,
the strategy $\alpha$ does not play $\sharp$. 
Note that the action $\sharp$ can be played by $\alpha$ only in the state $\q$, 
and thus $\Last(\rho) = \q$. Hence two steps later, in the path-outcome $\rho'$
of length $m+2$ that extends $\rho$, we have $\Last(\rho') = q_{\init}$. %
Since the action $\sharp$ is not played by $\alpha$ until position $n-1$,
after position $m+2$ in $\rho'$ the strategy $\alpha$ corresponds to 
a valid strategy from $\Last(\rho')$ in $\M$ that brings all the probability mass of 
$\Last(\rho') = q_{\init}$ to $\q$, witnessing that $q_{\init} \in \winsure{event}(\q)$.
\qed
\end{proof}\medskip

The proof of Lemma~\ref{lem: suWS} suggests an exponential-memory pure strategy
for sure weakly synchronizing that in $q \in \Pre^{n}(S)$ plays
an action $a$ such that $\post(q,a)\subseteq \Pre^{n-1}(S)$, which  
can be realized with exponential memory since $n \leq 2^{\abs{Q}}$.
It can be shown that exponential memory is necessary in general,
using an argument similar to the proof of exponential memory 
lower bound for sure eventually synchronizing (Theorem~\ref{theo:sure-eventually-pspace-c}), 
and by modifying the MDPs $\M_n$ (illustrated in \figurename~\ref{fig:exp-mem}) as follows: 
let the transitions from state $q_T$ go to $q_{\init}$ (instead of the absorbing state $q_{\bot}$).

\begin{theorem} \label{theo:weakly-sure}
For sure weakly synchronizing in MDPs:
\begin{enumerate}
\item  (Complexity). The membership problem is PSPACE-complete.
\item  (Memory). Exponential memory is necessary and sufficient for both pure and
randomized strategies, and pure strategies are sufficient.
\end{enumerate}
\end{theorem}

\subsection{Almost-sure weakly synchronizing}\label{sec:almost-weakly}

We present a characterization of almost-sure weakly synchronizing that 
gives a PSPACE upper bound for the membership problem. Our characterization,
similar to Lemma~\ref{lem: almost-limit-reduce-limit-event} for almost-sure eventually synchronizing,
uses the limit-sure eventually synchronizing objectives with exact support
introduced in Section~\ref{sec:almost-eventually}. 
We show that 
an MDP is almost-sure weakly synchronizing in target $T$ if (and only if),
for some set $U$, there is a sure eventually synchronizing strategy in target $U$,
and from the probability distributions with support $U$ there is a limit-sure
winning strategy for eventually synchronizing in $\Pre(T)$ with support in $\Pre(U)$.
This ensures that from the initial state we can have the whole probability mass in~$U$, and from~$U$ 
have probability $1-\epsilon$ in $\Pre(T)$ (and in $T$ in the next step), 
while the whole probability mass is back in $\Pre(U)$ (and in $U$ in the next step), 
allowing to repeat the strategy for $\epsilon \to 0$, thus ensuring infinitely
often probability at least $1-\epsilon$ in $T$ (for all $\epsilon > 0$).

\begin{lemma}\label{lem: almost-weak-reduce-limit-event}
Let $\M$ be an MDP and $T$ be a target set. For all states $q_{\init}$, 
we have $q_{\init}\in \winas{weakly}(\fsum_T)$
if and only if there exists a set~$U$ of states such that
\begin{itemize}
\item $q_{\init} \in \winsure{event}(\fsum_U)$, and \smallskip
\item $d_U \in \winlim{event}(\fsum_{\Pre(T)},\Pre(U))$
where $d_U$ is the uniform distribution over~$U$.
\end{itemize}
\end{lemma}

\begin{proof}
First, if $q_{\init} \in \winas{weakly}(\fsum_T)$, then 
there exists a strategy $\alpha$ such that for all $i \geq 0$ 
there exists $n_i \in \nat$ such that $\M^{\alpha}_{n_i}(T) \geq 1-2^{-i}$,
and moreover $n_{i+1} > n_i$ for all $i \geq 0$.
Let $s_i = \Supp(\M^{\alpha}_{n_i})$ be the support of~$\M^{\alpha}_{n_i}$.
Since the state space is finite, there is a set~$U$ that 
occurs infinitely often in the sequence~$s_0 s_1 \dots$,
thus for all $k>0$ there exists $m_k \in \nat$ such that 
$\M^{\alpha}_{m_k}(T) \geq 1-2^{-k}$ and 
$\M^{\alpha}_{m_k}(U) = 1$.
It follows that $\alpha$ is sure eventually synchronizing in $U$ 
from $q_{\init}$, i.e. $q_{\init} \in \winsure{event}(\fsum_U)$.
Moreover, we can assume that $m_{k+1} > m_k$ for all $k>0$
and thus $\M$ is also limit-sure eventually synchronizing in $\Pre(T)$
with exact support in $\Pre(U)$ from the initial distribution\footnote{%
Note that the initial distribution $d_1 = \M^{\alpha}_{m_1}$ can be fixed 
before the other quantifications in the statement that we want to prove, 
namely: $\exists d_1 \in \dist(U) \cdot \forall \epsilon > 0 \cdot \exists \alpha \cdot
\exists m_k: \M^{\alpha}_{m_k}(T) \geq 1-2^{-k}$ where we compute $\M^{\alpha}$ with 
initial distribution $d_1$. This is because we fixed the strategy $\alpha$
in the first step of the proof, and this is why we need that $q_{\init}$ is
almost-sure weakly synchronizing. 
Otherwise, if $q_{\init}$ is only limit-sure 
weakly synchronizing, we would get a possibly different initial distribution $d_1$
for each $\epsilon > 0$ (induced by a possibly different strategy $\alpha$
for each $\epsilon$) which would be problematic (see the example in Remark~\ref{rmk:limit-sure-event-fixed-dist}, p.\pageref{rmk:limit-sure-event-fixed-dist}). 
%
} $d_1 = \M^{\alpha}_{m_1}$.
Since $\Supp(d_1) = U = \Supp(d_U)$ and since 
only the support of the initial probability distributions is relevant for the limit-sure
eventually synchronizing objective (Corollary~\ref{col:uniform-dist-limit}),
it follows that $d_U \in \winlim{event}(\fsum_{\Pre(T)},\Pre(U))$.

To establish the converse, note that 
since $d_U \in \winlim{event}(\fsum_{\Pre(T)},\Pre(U))$, it follows
from Corollary~\ref{col:uniform-dist-limit} that
from all initial
distributions with support in~$U$, for all $\epsilon > 0$ there exists
a strategy $\alpha_{\epsilon}$ and a position $n_{\epsilon}$ such that
$\M^{\alpha_{\epsilon}}_{n_{\epsilon}}(T) \geq 1-\epsilon$ and
$\M^{\alpha_{\epsilon}}_{n_{\epsilon}}(U) = 1$.  We construct an
almost-sure weakly synchronizing strategy $\alpha$ as
follows:
\begin{itemize}
\item Since $q_{\init} \in \winsure{event}(\fsum_U)$, play according to
a sure eventually synchronizing strategy from $q_{\init}$ until all the
probability mass is in~$U$.  
\item Then for $i=1,2, \dots$ and $\epsilon_i = 2^{-i}$, repeat the following procedure: 
 \begin{itemize}
 \item given the current probability distribution, play according to $\alpha_{\epsilon_i}$ for $n_{\epsilon_i}$
 steps (ensuring probability mass at least $1-2^{-i}$ in $\Pre(T)$ and support of the probability mass in $\Pre(U)$);
 \item then from states in $\Pre(T)$, play an action to ensure reaching $T$ in the next step, and from states
       in $\Pre(U)$ ensure reaching $U$. 
 \item continue playing according to $\alpha_{\epsilon_{i+1}}$ for $n_{\epsilon_{i+1}}$ steps, etc.  
 \end{itemize}
\end{itemize}
 Since $n_{\epsilon_i} + 1 > 0$ for all $i \geq 0$, this strategy ensures that $\limsup_{n \to \infty}
\M^{\alpha}_n(T) = 1$ from $q_{\init}$, hence $q_{\init} \in \winas{weak}(\fsum_T)$. 
Note that this strategy is pure.
\qed
\end{proof}\medskip

Since the membership problems for sure eventually synchronizing
and for limit-sure eventually synchronizing with exact support are PSPACE-complete 
(Theorem~\ref{theo:sure-eventually-pspace-c} and Theorem~\ref{theo:limit-sure-eventually}), 
the membership problem for almost-sure weakly synchronizing
is in PSPACE by  
guessing the set $U$, and checking that $q_{\init} \in \winsure{event}(\fsum_U)$, and that 
$d_U \in \winlim{event}(\fsum_{\Pre(T)},\Pre(U))$.
We establish a matching PSPACE lower bound.

\begin{lemma}\label{lem:almost-weakly-pspace-hard}
The membership problem for $\winas{weakly}(\fsum_T)$ is PSPACE-hard
even if $T$ is a singleton.
\end{lemma}

\begin{proof}
We use the same reduction and construction as in the PSPACE-hardness proof of 
Lemma~\ref{lem:limit-event-pspace-hard} where from an MDP~$\M$ and a singleton~$T$, 
we constructed~$N$ and $q_{\init}$. Referring to that construction, we 
show that $\Pre^n_{\M}(T) \neq \emptyset$ for all~$n \geq 0$ if and only if 
$q_{\init} \in \winas{weakly}(T)$. 

First, if $\Pre^n_{\M}(T) \neq \emptyset$ for all~$n \geq 0$, then by 
Lemma~\ref{lem: almost-weak-reduce-limit-event} we need
to show that $(i)$~$q_{\init} \in \winsure{event}(\fsum_Q)$, and 
$(ii)$ $d_Q \in \winlim{event}(\fsum_{\Pre(T)},\Pre(Q))$
where $d_Q$ is the uniform distribution over~$Q$.
To show $(i)$, we can play $\sharp$ from $q_{\init}$ to get
the probability mass synchronized in $Q$. To show $(ii)$, since playing
$\sharp$ from $d_Q$ ensures to reach $q_{\init}$, it suffices to prove 
that $q_{\init} \in \winlim{event}(\fsum_{T},Q)$, which is done in the
proof of Lemma~\ref{lem:limit-event-pspace-hard}.

For the converse direction, if $q_{\init}$ is almost-sure weakly synchronizing in~$T$,
then $q_{\init}$ is also limit-sure eventually synchronizing in~$T$,
and we can directly use that argument in the proof of Lemma~\ref{lem:limit-event-pspace-hard}
to show that $\Pre^n_{\M}(T) \neq \emptyset$ for all~$n \geq 0$.

It follows from this reduction that the membership problem for almost-sure weakly 
synchronization is PSPACE-hard.
\qed
\end{proof}\medskip

It is easy to show that winning strategies require infinite memory for almost-sure 
weakly synchronizing in the same example that we used in 
the proof of Lemma~\ref{lem:inf-mmeory-almost-event} to show that 
infinite memory may be necessary for almost-sure eventually synchronizing (\figurename~\ref{fig:inf-mem}),


\begin{theorem}\label{theo:weakly-almost}
For almost-sure weakly synchronizing in MDPs:

\begin{enumerate}
\item (Complexity). The membership problem is PSPACE-complete.

\item (Memory). Infinite memory is necessary in general for both pure and 
randomized strategies, and pure strategies are sufficient.
\end{enumerate}
\end{theorem}

\subsection{Limit-sure weakly synchronizing}\label{sec:limit-sure-weakly}

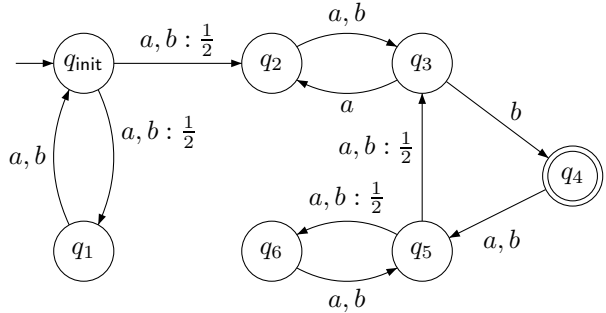
\begin{figure}[t]

\begin{center}
    \begin{picture}(79,40)

\node[Nmarks=i](n0)(10,32){$q_{\init}$}
\node[Nmarks=n](n1)(10,7){$q_1$}
\node[Nmarks=n](n2)(35,32){$q_2$}
\node[Nmarks=n](n3)(55,32){$q_3$}
\node[Nmarks=r](n4)(75,17){$q_4$}
\node[Nmarks=n](n5)(55,7){$q_5$}
\node[Nmarks=n](n6)(35,7){$q_6$}

\drawedge[ELpos=40, ELside=l, curvedepth=4](n0,n1){$a,b:\frac{1}{2}$}
\drawedge[ELpos=50, ELside=l](n0,n2){$a,b:\frac{1}{2}$}
\drawedge[ELpos=50, ELside=l, curvedepth=4](n1,n0){$a,b$}

\drawedge[ELpos=50, ELside=l, curvedepth=4](n2,n3){$a,b$}
\drawedge[ELpos=50, ELside=l, curvedepth=4](n3,n2){$a$}

\drawedge[ELpos=55](n3,n4){$b$}
\drawedge[ELpos=55](n4,n5){$a,b$}
\drawedge[ELpos=55](n5,n3){$a,b:\frac{1}{2}$}

\drawedge[ELpos=50, ELside=r, curvedepth=-4](n5,n6){$a,b:\frac{1}{2}$}
\drawedge[ELpos=50, ELside=r, curvedepth=-4](n6,n5){$a,b$}


\end{picture}
\end{center}
 \caption{An example to show $q_{\init}\in \winlim{weakly}(q_4)$ implies $q_{\init}\in \winas{weakly}(q_4)$.\label{fig:weak-limit}}

\end{figure}

We show that the winning regions for almost-sure
and limit-sure weakly synchronizing coincide. The result
is not intuitively obvious (recall that it does not
hold for eventually synchronizing, by Lemma~\ref{lem:dif-in-def}$(ii)$) 
and requires a careful analysis of the structure of limit-sure winning strategies
to show that they always imply the existence of an almost-sure 
winning strategy. The construction of an almost-sure
winning strategy from a family of limit-sure winning strategies
is illustrated in the following example.

Consider the MDP $\M_{\ref{fig:weak-limit}}$ in~\figurename~\ref{fig:weak-limit}
with initial state $q_{\init}$ and target set $T = \{q_4\}$.
Note that there is a relevant strategic choice only in $q_3$,
where we can either loop through $q_2$, or go to the target $q_4$. 
First we argue that $\M_{\ref{fig:weak-limit}}$ is limit-sure weakly synchronizing,
then we explain why limit-sure weakly synchronizing implies that we can construct  
an almost-sure weakly synchronizing strategy in this example, using the same
line of arguments as in our proof of the general result (that limit-sure winning 
implies almost-sure winning) presented further as Lemma~\ref{lem:weak-without-vanishing}, 
Lemma~\ref{lem:weak-event-uniform}, and Theorem~\ref{theo:weakly-ls-is-as}.
 
\subsubsection{The MDP $\M_{\ref{fig:weak-limit}}$ is limit-sure weakly synchronizing}\label{sec:M10-limit-sure-weakly}
To show that $\M_{\ref{fig:weak-limit}}$ (\figurename~\ref{fig:weak-limit}) is limit-sure weakly synchronizing,
we rely on the following claims:
\begin{itemize}
\item $q_{\init}$ is limit-sure eventually synchronizing with target $T = \{q_4\}$; 

\item $q_4$, which can be viewed as a uniform distribution over $T$, 
is also limit-sure eventually synchronizing with target $T = \{q_4\}$ (even after at least one step). 
\end{itemize}

The above claims hold since for arbitrarily small $\epsilon > 0$, from both 
$q_{\init}$ and $q_4$, we can inject probability mass $1-\epsilon$ 
in $q_3$ (by playing $a$ long enough in $q_3$),
and then switching to playing $b$ in $q_3$ gets probability $1 - \epsilon$ 
in $T$. 

Now, these two claims are sufficient to show that $q_{\init}$
is limit-sure weakly synchronizing in $T = \{q_4\}$,
and to define a family $\alpha_{\epsilon}$ of limit-sure winning strategies as follows: 
given $\epsilon > 0$, let $\alpha_{\epsilon}$ play from $q_{\init}$ a strategy 
to ensure probability at least $p_1 = 1-\frac{\epsilon}{2}$ in $q_4$ (in finitely many steps),
and then play according to a strategy that ensures from $q_4$ 
probability $p_2  = p_1 - \frac{\epsilon}{4}$ in $q_4$ (in finitely many, and at least one step), 
and repeat this process using strategies that ensure, if the probability mass in $q_4$ 
is at least $p_i$, that (in at least one step) the probability in $q_4$ is at least
$p_{i+1}  = p_i - \frac{\epsilon}{2^{i+1}}$.
It follows that $p_i = 1 - \frac{\epsilon}{2} - \frac{\epsilon}{4} - \dots - \frac{\epsilon}{2^i} > 1 - \epsilon$ for all $i \geq 1$, and thus $\limsup_{i \to \infty} p_i \geq 1 - \epsilon$, thus 
$\alpha_{\epsilon}$ is weakly $(1-\epsilon)$-synchronizing. 
Therefore $q_{\init}$ is limit-sure weakly synchronizing for target~$\{q_4\}$.

\paragraph{Illustration of Lemma~\ref{lem:weak-event-uniform}} 
We show in Lemma~\ref{lem:weak-event-uniform} that in general the above two claims
hold in a limit-sure weakly synchronizing MDP (and it is easy to generalize the argument
we used for $\M_{\ref{fig:weak-limit}}$ to show that the
converse implication of Lemma~\ref{lem:weak-event-uniform} holds as well, although
we do not need to prove this for our purpose). Hence, Lemma~\ref{lem:weak-event-uniform} shows 
that limit-sure weakly synchronizing strategies can always be decomposed as a repetition
of eventually $(1-\epsilon)$-synchronizing strategies, played for finitely many steps
(and with decreasing $\epsilon$).

\subsubsection{The MDP $\M_{\ref{fig:weak-limit}}$ is almost-sure weakly 
synchronizing}\label{sec:M10-almost-sure-weakly}

The following claims are central to show that $\M_{\ref{fig:weak-limit}}$ (\figurename~\ref{fig:weak-limit})
is almost-sure weakly synchronizing (note the slight difference with the
claims in Section~\ref{sec:M10-limit-sure-weakly}):
\begin{itemize}
\item $q_{\init}$ is limit-sure eventually synchronizing with target $\{q_3\}$; 

\item $q_4$, which can be viewed as a uniform distribution over $T$, 
is also limit-sure eventually synchronizing with target $\{q_3\}$. 
\end{itemize}

The above claims hold by the exact same argument as in Section~\ref{sec:M10-limit-sure-weakly} 
(and follow directly from the fact that $\M_{\ref{fig:weak-limit}}$ is limit-sure weakly synchronizing).
Intuitively, an almost-sure weakly synchronizing strategy in $\M_{\ref{fig:weak-limit}}$
repeats the following phases (informally):
\begin{enumerate}
\item accumulate probability mass (arbitrarily close to $1$, say $1-\epsilon_0$) in $q_3$; \label{item:as-w-one}

\item transfer the probability mass from $q_3$ to $q_4$; \label{item:as-w-two}

\item given the current distribution, decrease $\epsilon_0$ by half and repeat from~(\ref{item:as-w-one}.). \label{item:as-w-three}
 
\end{enumerate}
 
Such a strategy would ensure, for all $\epsilon > 0$, probability mass at least 
$1-\epsilon$ in $q_4$ infinitely often, and thus it is almost-sure weakly synchronizing.
To show that such a strategy exists and is well defined, we need to show that at every iteration
from the distribution at the beginning of step~(\ref{item:as-w-one}.)~we can indeed
accumulate probability mass in $q_3$. This is true in the first iteration, as we start
from $q_{\init}$. After one iteration, the distribution has support $S = \{q_1,q_2,q_4\} = \{q_1,q_2\} \cup \{q_4\}$
where the distributions over $\{q_1,q_2\}$ are limit-sure eventually synchronizing to $\{q_3\}$ (by analogous argument as the first 
claim above), and $q_4$ is also limit-sure eventually synchronizing to $\{q_3\}$ (by the second claim above).
In the next iterations, from the distribution at step~\ref{item:as-w-one} 
the situation is similar (as in fact all states are limit-sure weakly synchronizing with target $\{q_3\}$).

\paragraph{Illustration of Theorem~\ref{theo:weakly-ls-is-as} (claim 1 of the proof)}
The state $q_3$ plays a crucial role here because $\{q_3\} = \Pre(T)$
and $\{q_3\} = \Pre^2(\{q_3\})$, thus $R = \{q_3\}$ occurs infinitely
often in the sequence $\Pre^i(T)$ (for $i \geq 0$), which is ultimately 
periodic with period $r=2$. It follows from the general result established in Claim~1 
of the proof of Theorem~\ref{theo:weakly-ls-is-as} that limit-sure weakly synchronizing
with target $T = \{q_4\}$ implies limit-sure eventually (and even almost-sure weakly) 
synchronizing with target $\{q_3\}$.

Intuitively, from the fact that the distributions over both $\{q_1,q_2\}$ and $\{q_4\}$
are limit-sure eventually synchronizing to $\{q_3\}$, it may not be obvious
that the distributions over $\{q_1,q_2,q_4\}$ are limit-sure eventually synchronizing to $\{q_3\}$.
For instance in the example of $\M_{\ref{fig:weak-limit}}$ (\figurename~\ref{fig:weak-limit}), 
(almost all) the probability mass in $T = \{q_4\}$ can move
to $q_3$ in an even number of steps, while from $\{q_1,q_2\}$ an odd number
of steps is required, resulting in a \emph{shift} of the probability mass. 

\paragraph{Illustration of Theorem~\ref{theo:weakly-ls-is-as} (claim 2 of the proof)}
Although, the simplest strategy accumulates probability mass in $q_3$ after 
\emph{even} number of steps from $\{q_4\}$, by repeating the same strategy two times from $q_4$ 
(injecting large probability mass in $q_3$, moving to $q_4$, and injecting
in $q_3$ again), we can accumulate probability mass in $q_3$ after 
\emph{odd} number of steps from $\{q_4\}$, thus in synchronization with the
probability mass accumulated in $q_3$ from $\{q_1,q_2\}$. However, by doing that,
we also hit several other states and the remaining (small) probability mass
is distributed over support $\{q_1,q_2,q_3,q_4,q_5,q_6\}$ when the next iteration
starts. By a similar argument, we can again construct a strategy to implement the 
phases described above, and this can be done for all iterations and for $\epsilon \to 0$.
Indeed, the result of Claim~2 in the proof of Theorem~\ref{theo:weakly-ls-is-as}
shows that by repeating strategies with shifting, we can eventually synchronize
all the shifts.

\paragraph{Vanishing states and Lemma~\ref{lem:weak-without-vanishing}}
In the example of $\M_{\ref{fig:weak-limit}}$ (\figurename~\ref{fig:weak-limit}),
the target $T$ is a singleton, which makes the result easier to prove than
for an arbitrary set $T$. In particular, the second claim at the beginning
of this section (that $q_4$ is limit-sure eventually synchronizing with 
target $\{q_3\}$) follows from the fact that $q_4$ is limit-sure weakly synchronizing
to itself: it is easy to argue that if the probability is infinitely often
arbitrarily close to $1$ in $T = \{q_4\}$, then (starting with probability~$1$) 
from $q_4$ there must be a way to inject (almost all) the probability mass back to $q_4$ (via $q_3$).
However, if $T$ is not a singleton, the same argument is more difficult because
when the probability mass is $1-\epsilon$ in $T$, it may still be that some state $q$ in $T$ holds 
only a tiny (less than $\epsilon$) probability mass, which makes it more difficult
to argue that we must be able to inject (almost all) the probability mass from $q$ back to $T$
(because if the tiny probability in $q$ could not be injected at all in $T$,
there would be no contradiction to the fact that probability $1-\epsilon$ is in $T$
infinitely often).

Therefore, given an arbitrary target set $T$, we need
to get rid of the states in $T$ that do not contribute a significant (i.e., bounded away from $0$)
probability mass in the limit, that we call the vanishing states.
We show that the vanishing states can be removed from $T$ without changing the winning 
region for limit-sure winning. When the target set has no vanishing state, we can 
construct an almost-sure winning strategy similarly to the case of a singleton target set.
\smallskip

Given an MDP $\M$ with initial state $q_{\init} \in \winlim{weakly}(\fsum_T)$ that is 
limit-sure winning for the weakly synchronizing objective in target set $T$,
let~$(\alpha_i)_{i\in\nat}$ be a family of limit-sure winning strategies
such that $\limsup_{n \to \infty}\M^{\alpha_i}_n(T)\geq 1-\epsilon_i$
where $\lim_{i \to \infty} \epsilon_i = 0$. Hence, by definition of $\limsup$,
for all $i \geq 0$ there exists a strictly increasing sequence
$k_{i,0} < k_{i,1} < \cdots$ of positions such that $\M^{\alpha_i}_{k_{i,j}}(T) \geq 1-2 \epsilon_i$ 
for all $j \geq 0$.
A state $q \in T$ is \emph{vanishing} if $\liminf_{i\to \infty} \liminf_{j\to \infty} \M^{\alpha_i}_{k_{i,j}}(q)=0$
for some family of limit-sure weakly synchronizing strategies~$(\alpha_i)_{i\in\nat}$.
Intuitively, the contribution of a vanishing state $q$ to the probability in $T$
tends to $0$ and therefore $\M$ is also limit-sure winning for the weakly 
synchronizing objective in target set $T \setminus \{q\}$.

\begin{figure}[t]
\begin{center}
    \begin{picture}( 40,35)

\node[Nmarks=i,iangle=180](n0)(10,25){$q_{\init}$}
\node[Nmarks=n](n1)(35,25){$q_1$}
\node[Nmarks=n](n2)(10,5){$q_2$}
\node[Nmarks=n](n3)(35,5){$q_3$}

\drawedge[ELdist=.5](n0,n1){$a: \frac{1}{2}$}
\drawloop[ELside=l,loopCW=y, loopangle=90, loopdiam=4](n0){$a,b:\frac{1}{2}$}
\drawedge[ELdist=.5](n0,n2){$b: \frac{1}{2}$}

\drawloop[ELside=l,loopCW=y, loopangle=90, loopdiam=4](n1){$a$}
\drawedge[ELdist=.5,ELside=r](n1,n3){$b$}

\drawedge[ELpos=50, ELdist=.5, ELside=l, curvedepth=10](n2,n0){$a,b$}
\drawedge[ELpos=50, ELdist=.5, ELside=r, curvedepth=-10](n3,n1){$a,b$}
\end{picture}
\end{center}
 \caption{The state~$q_2$ is  vanishing  for target set~$T=\{q_2,q_3\}$ and 
strategies~$(\alpha)_{i\in\nat}$ where~$\alpha_i$ repeats playing~$i$ times~$a$, then playing~$b$ forever.
\label{fig:vanishing-state}}
\end{figure}
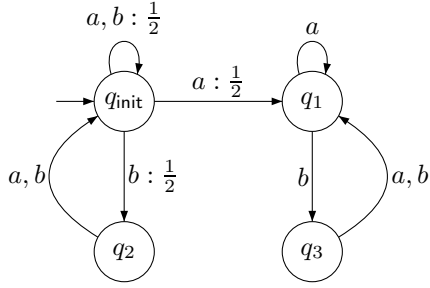

\paragraph{Example}
Consider the MDP in~\figurename~\ref{fig:vanishing-state} where 
all transitions are deterministic except from the initial state~$q_{\init}$.
The state~$q_{\init}$ has two successors on all actions:
$$\delta(q_{\init},a)(q_{\init})=\delta(q_{\init},a)(q_1)=\frac{1}{2} 
\quad\text{ and }\quad \delta(q_{\init},b)(q_{\init})=\delta(q_{\init},b)(q_2)=\frac{1}{2}.$$
Let~$T=\{q_2,q_3\}$ be the target set and for all $i\in \nat$, let 
$\alpha_i$ be the strategy 
that repeats forever the following template in every state:
playing~$i$ times~$a$ and then playing~$b$. 
The family of strategies~$(\alpha_i)_{i\in\nat}$ is a witness to show that
$q_{\init} \in \winlim{weakly}(\fsum_T)$ where
 the state~$q_2$ is a vanishing state. 
The contribution of~$q_2$ in accumulating the probability
mass in $\{q_2,q_3\}$ tends to~$0$ when~$i\to\infty$.
As a result, $q_{\init} \in \winlim{weakly}(q_3)$ too.  

\subsubsection{Proof that limit-sure weakly and almost-sure weakly coincide}

We present the formal proof of the main result (Theorem~\ref{theo:weakly-ls-is-as})
along with the intermediate lemmas that we illustrated in Section~\ref{sec:M10-limit-sure-weakly}
and Section~\ref{sec:M10-almost-sure-weakly}.

\begin{lemma}\label{lem:weak-without-vanishing}
If an MDP $\M$ is limit-sure weakly synchronizing in target set~$T$, 
then there exists a set $T' \subseteq T$ such that 
$\M$ is limit-sure weakly synchronizing in~$T'$
without vanishing states.
\end{lemma}

\begin{proof}
If there is no vanishing state for~$(\alpha_i)_{i\in\nat}$, then take $T' = T$
and the proof is complete. 
Otherwise, let $(\alpha_i)_{i\in\nat}$ be a 
family of limit-sure winning strategies such that $\limsup_{n \to \infty}\M^{\alpha_i}_n(T)\geq 1-\epsilon_i$
where $\lim_{i \to \infty} \epsilon_i = 0$ and let $q$ be a vanishing state 
for $(\alpha_i)_{i\in\nat}$.
We show that $(\alpha_i)_{i\in\nat}$ is limit-sure weakly synchronizing in~$T \setminus \{q\}$.
For every $i \geq 0$ let $k_{i,0} < k_{i,1} < \cdots$ be a strictly increasing sequence 
such that $(a)$ $\M^{\alpha_i}_{k_{i,j}}(T) \geq 1-2\epsilon_i$ for all $i,j \geq 0$, and
$(b)$ $\liminf_{i\to \infty} \liminf_{j\to \infty} \M^{\alpha_i}_{k_{i,j}}(q)=0$.

It follows from $(b)$ that for all $\epsilon > 0$ and all~$x > 0$ there exists~$i>x$ such that
for all~$y > 0$ there exists~$j>y$ such that~$\M^{\alpha_i}_{k_{i,j}}(q)<\epsilon$,
and thus 
$$\M^{\alpha_i}_{k_{i,j}}(T \setminus \{q\}) \geq 1-2\epsilon_i - \epsilon$$
by $(a)$. Since this holds for infinitely many $i$'s, we can choose $i$ such that
$\epsilon_i < \epsilon$ and we have 
$$\limsup_{j \to \infty} \M^{\alpha_i}_{k_{i,j}}(T \setminus \{q\}) \geq 1- 3 \epsilon$$
and thus 
$$\limsup_{n \to \infty} \M^{\alpha_i}_{n}(T \setminus \{q\}) \geq 1- 3 \epsilon$$
since the sequence $(k_{i,j})_{j\in\nat}$ is strictly increasing.
This shows that $(\alpha_i)_{i\in\nat}$ is limit-sure weakly synchronizing in~$T \setminus \{q\}$.

By repeating this argument as long as there is a vanishing state (thus at most $\abs{T}-1$ times), 
we can construct the desired set $T' \subseteq T$ without vanishing state.
\qed 
\end{proof}\medskip

For a limit-sure weakly synchronizing MDP in target set $T$ (without vanishing 
states), we show that from a probability distribution with support $T$,
a probability mass arbitrarily close to $1$ can be injected synchronously
back in $T$ (in at least one step), that is $d_T \in \winlim{event}(\fsum_{\Pre(T)})$.
The same holds from the initial state $q_{\init}$ of the MDP. This property
is the key to construct an almost-sure weakly synchronizing strategy.

\begin{lemma}\label{lem:weak-event-uniform}
If an MDP $\M$ with initial state $q_{\init}$ is limit-sure weakly synchronizing in a target set~$T$ without 
vanishing states, then we have $q_{\init} \in \winlim{event}(\fsum_{\Pre(T)})$ and
\mbox{$d_T \in \winlim{event}(\fsum_{\Pre(T)})$}
where $d_T$ is the uniform distribution over~$T$.
\end{lemma}

\begin{proof}
Since $q_{\init} \in \winlim{weakly}(\fsum_T)$ and 
$\winlim{weakly}(\fsum_T) \subseteq \winlim{event}(\fsum_T)$,
we have $q_{\init} \in \winlim{event}(\fsum_T)$ and thus it suffices
to prove that $d_T \in \winlim{event}(\fsum_{\Pre(T)})$. This is because
then from $q_{\init}$, probability arbitrarily close to $1$ can be injected
in $\Pre(T)$ through a distribution with support in $T$ (since 
by Corollary~\ref{col:uniform-dist-limit} only the support of the 
initial probability distribution is important for limit-sure 
eventually synchronizing).

Let~$(\alpha_i)_{i\in\nat}$ be a family of limit-sure winning strategies
such that $$\limsup_{n \to \infty}\M^{\alpha_i}_n(T)\geq 1-\epsilon_i
\text{ where } \lim_{i \to \infty} \epsilon_i = 0,$$ and such that
there is no vanishing state. 
For every $i \geq 0$ let $k_{i,0} < k_{i,1} < \cdots$ be a strictly increasing sequence 
such that $\M^{\alpha_i}_{k_{i,j}}(T) \geq 1 - 2\epsilon_i$ for all $i,j \geq 0$, and
let 
$$B = \min_{q \in T}\, \liminf_{i\to \infty} \, \liminf_{j\to \infty} \M^{\alpha_i}_{k_{i,j}}(q)
    = \liminf_{i\to \infty} \, \liminf_{j\to \infty}\, \min_{q \in T} \M^{\alpha_i}_{k_{i,j}}(q).$$
Note that $B > 0$ since there is no vanishing state.
It follows that there  exists~$x > 0$ such that for all~$i > x$ 
there exists~$y_i > 0$ such that for all~$j > y_i$ and all $q \in T$
we have~$\M^{\alpha_i}_{k_{i,j}}(q) \geq \frac{B}{2}$.

Given~$\nu > 0$, let $i>x$ such that $\epsilon_i <  \frac{\nu B}{4}$,
and for $j > y_i$, consider the positions $n_1 = k_{i,j}$ and $n_2 = k_{i,j+1}$.
We have $n_1 < n_2$ and $\M^{\alpha_i}_{n_1}(T) \geq 1 - 2\epsilon_i$
and $\M^{\alpha_i}_{n_2}(T) \geq 1 - 2\epsilon_i$, and 
$\M^{\alpha_i}_{n_1}(q) \geq \frac{B}{2}$ for all $q \in T$. 
Consider the strategy $\beta$ that plays 
like $\alpha_i$ plays from position $n_1$ and thus transforms the
distribution $\M^{\alpha_i}_{n_1}$ into $\M^{\alpha_i}_{n_2}$. 
For all states $q \in T$, from the Dirac distribution on $q$ under strategy $\beta$, the probability 
to reach $Q \setminus T$ in $n_2 - n_1$ steps is thus at most 
$\frac{\M^{\alpha_i}_{n_2}(Q \setminus T)}{\M^{\alpha_i}_{n_1}(q)} \leq \frac{2\epsilon_i}{B/2} < \nu$.

Therefore, from an arbitrary probability distribution with support $T$
we have $\M^{\beta}_{n_2-n_1}(T) > 1 - \nu$, showing that $d_T$ is
limit-sure eventually synchronizing in $T$ and thus in $\Pre(T)$ 
since $n_2 - n_1 > 0$ (it is easy to show that if the mass of probability in
$T$ is at least $1-\nu$, then the mass of probability in $\Pre(T)$ one step
before is at least $1-\frac{\nu}{\eta}$ where $\eta$ is the smallest positive
probability in $\M$).
\qed
\end{proof}\medskip

To show that limit-sure and almost-sure winning coincide for weakly synchronizing
objectives, from a family of limit-sure winning strategies we construct 
an almost-sure winning strategy that uses the eventually synchronizing 
strategies of Lemma~\ref{lem:weak-event-uniform}. The construction consists
in using successively strategies that ensure probability mass $1-\epsilon_i$
in the target $T$, for a decreasing sequence $\epsilon_i \to 0$. Such strategies
exist by Lemma~\ref{lem:weak-event-uniform}, both from the initial state and
from the set $T$. However, the mass of probability that can be guaranteed to
be synchronized in~$T$ by the successive strategies is always smaller than $1$, 
and therefore we need to argue that the remaining mass of probability 
(of total size $\epsilon_i$) scattered in the state space can also get synchronized in $T$, 
despite the variable shifts with the main mass of probability.

Two main
key arguments are needed to establish the correctness of the construction:
$(1)$ eventually synchronizing implies that a finite number of steps is sufficient
to obtain a probability mass of $1-\epsilon_i$ in $T$, and thus the construction
of the strategy is well defined, and $(2)$ by the finiteness of the period $r$
(such that $R = \Pre^r(R)$ where $R = \Pre^k(T)$ for some $k$) from every state,
we can accumulate shifts such that their sum is $0 \mod r$, 
and thus the probability mass from every state contributes
(synchronously) to the probability accumulated in the target.

\begin{theorem}\label{theo:weakly-ls-is-as}
$\winlim{weakly}(\fsum_T) = \winas{weakly}(\fsum_T)$ for all MDPs and target sets~$T$.
\end{theorem}

\begin{proof}
Since $\winas{weakly}(\fsum_T) \subseteq \winlim{weakly}(\fsum_T)$ holds
by the definition, it is sufficient to prove that $\winlim{weakly}(\fsum_T) \subseteq \winas{weakly}(\fsum_T)$
and by Lemma~\ref{lem:weak-without-vanishing} it is sufficient to prove
that if $q_{\init} \in \winlim{weakly}(\fsum_T)$ is limit-sure weakly synchronizing 
in $T$ without vanishing state, then $q_{\init}$ is almost-sure weakly synchronizing 
in $T$. If $T$ has vanishing states, then consider $T' \subseteq T$ as in 
Lemma~\ref{lem:weak-without-vanishing} and it will follows that 
$q_{\init}$ is almost-sure weakly synchronizing in $T'$ and thus also in $T$.
We proceed with the proof that $q_{\init} \in \winlim{weakly}(\fsum_T)$ implies 
$q_{\init} \in \winas{weakly}(\fsum_T)$.

For $i=1,2,\dots$ consider the sequence of predecessors~$\Pre^{i}(T)$,
which is ultimately periodic: let $1 \leq k, r \leq 2^{\abs{Q}}$ such
that $\Pre^{k}(T) = \Pre^{k+r}(T)$, and let $R = \Pre^{k}(T)$.
Thus $R = \Pre^{k+r}(T) = \Pre^{r}(R)$.

\paragraph{Claim 1} {\it We have $q_{\init} \in \winas{event}(\fsum_R)$ and 
$d_T \in \winas{event}(\fsum_R)$.}

\paragraph{Proof of Claim 1}
By Lemma~\ref{lem:weak-event-uniform}, since there is no vanishing state in $T$
we have $q_{\init} \in \winlim{event}(\fsum_{\Pre(T)})$
and $d_T \in \winlim{event}(\fsum_{\Pre(T)})$. 
The characterization of the winning region for limit-sure eventually synchronizing 
given by Lemma~\ref{lem:lse-pre}, and the fact that almost-sure and limit-sure
coincide for eventually synchronizing in the set $R$ (Corollary~\ref{cor:as-same-as-ls-for-R}) 
give the following:\smallskip

\begin{tabular}{ll}
either $(1)$ $q_{\init} \in \winsure{event}(\fsum_{\Pre(T)})$ or & $\!\!\!\!(2)$ $q_{\init} \in \winas{event}(\fsum_R)$, and  \\
either $(a)$ $d_T \in \winsure{event}(\fsum_{\Pre(T)})$ or & $\!\!\!\!(b)$ $d_T \in \winas{event}(\fsum_R)$.
\end{tabular} \medskip

We show that $(a)$ implies $(b)$, hence that $(b)$ holds. Then we show that $(1)$ implies $(2)$,
hence that $(2)$ holds. We conclude that both $(2)$ and $(b)$, which establishes Claim~1.

To show that $(a)$ implies $(b)$: by the characterization of sure eventually synchronizing (Lemma~\ref{lem:sure-ss-pre}),
if $(a)$ holds, then $T \subseteq \Pre^{i}(T)$ for some $i \geq 1$, and thus  
$T \subseteq \Pre^{n \cdot i}(T)$ for all $n \geq 0$ by monotonicity of $\Pre^{i}(\cdot)$.
This entails for $n \cdot i \geq k$ that $T \subseteq \Pre^m(R)$ where $m = (n \cdot i - k) \mod r$
and thus $d_T$ is sure (and almost-sure) winning for the eventually synchronizing 
objective in target~$R$ (by Lemma~\ref{lem:sure-ss-pre}), hence $(b)$ holds.

To show that $(1)$ implies $(2)$: if $(1)$ holds, then we can play a sure-winning
strategy from $q_{\init}$ to ensure in finitely many steps probability~$1$ in $\Pre(T)$ 
and in the next step probability~$1$ in $T$, and by $(b)$ play an almost-sure 
winning strategy for eventually synchronizing in $R$. Hence, $q_{\init} \in \winas{event}(\fsum_R)$, i.e. $(2)$ holds.
The proof of Claim~1 is done.
\smallskip

\begin{figure}[!tbp]
  \centering
  \hrule
  \subfigure[From state $q$ with shift $h$]{
    \label{fig:shift-strategy}
    \begin{picture}(60,25)

\gasset{Nframe=n, Nh=2, Nw=2, Nmr=0}

\node[Nmarks=n, Nw=3](n0)(2,10){$q$}

\node[Nmarks=n, ExtNL=y, NLangle=90](n1)(17,15){$R$}
\node[Nmarks=n, ExtNL=y, NLangle=270](n2)(17,5){$Q \!\setminus\! R$}
\node[Nmarks=n, ExtNL=y, NLangle=90](n3)(32,15){$R$}
\node[Nmarks=n, ExtNL=y, NLangle=270](n4)(32,5){$Q \!\setminus\! R$}
\node[Nmarks=n, ExtNL=y, NLangle=90](n5)(47,15){$R$}
\put(47,14){\makebox(0,0)[l]{\scalebox{.75}{$\geq 1\!-\epsilon\! \to 1$}}}
\node[Nmarks=n, ExtNL=y, NLangle=270](n6)(47,5){$Q \!\setminus\! R$}
\put(47,6){\makebox(0,0)[l]{\scalebox{.75}{$\leq \epsilon \to 0$}}}

\put(8,10){\makebox(0,0)[l]{$\alpha_h$}}

\drawline[AHnb=0](2,22)(47,22)
\drawline[AHnb=0, dash={.5 .8}0](47,22)(53,22)
\drawline[AHnb=0](2,23)(2,21)
\drawline[AHnb=0](17,23)(17,21)
\drawline[AHnb=0](32,23)(32,21)
\drawline[AHnb=0](47,23)(47,21)
\node[Nmarks=n](label)(9.5,24){$h$}
\node[Nmarks=n](label)(24.5,24){$r$}
\node[Nmarks=n](label)(39.5,24){$r$}

\drawedge[ELpos=40, ELside=l, curvedepth=0](n0,n1){}  
\drawedge[ELpos=40, ELside=r, curvedepth=0](n0,n2){}  
\drawedge[ELpos=40, ELside=l, curvedepth=0](n1,n3){}
\drawedge[ELpos=40, ELside=l, curvedepth=0](n2,n3){}
\drawedge[ELpos=40, ELside=l, curvedepth=0](n2,n4){}
\drawedge[ELpos=40, ELside=l, curvedepth=0](n3,n5){}
\drawedge[ELpos=40, ELside=l, curvedepth=0](n4,n5){}
\drawedge[ELpos=40, ELside=l, curvedepth=0](n4,n6){}


\end{picture}}
  \hfill
  \subfigure[From uniform distribution $d_T$ with shift $t$]{
    \label{fig:shift-strategy2}
    \begin{picture}(62,25)(-2,0)

\gasset{Nframe=n, Nh=2, Nw=2, Nmr=0}

\node[Nmarks=n, Nw=4](n0)(2,10){$d_T$}

\node[Nmarks=n, ExtNL=y, NLangle=90](n1)(17,15){$R$}
\node[Nmarks=n, ExtNL=y, NLangle=270](n2)(17,5){$Q \!\setminus\! R$}
\node[Nmarks=n, ExtNL=y, NLangle=90](n3)(32,15){$R$}
\node[Nmarks=n, ExtNL=y, NLangle=270](n4)(32,5){$Q \!\setminus\! R$}
\node[Nmarks=n, ExtNL=y, NLangle=90](n5)(47,15){$R$}
\put(47,14){\makebox(0,0)[l]{\scalebox{.75}{$\geq 1\!-\eta\! \to 1$}}}
\node[Nmarks=n, ExtNL=y, NLangle=270](n6)(47,5){$Q \!\setminus\! R$}
\put(47,6){\makebox(0,0)[l]{\scalebox{.75}{$\leq \eta \to 0$}}}

\put(8,10){\makebox(0,0)[l]{$\alpha_t$}}

\drawline[AHnb=0](2,22)(47,22)
\drawline[AHnb=0, dash={.5 .8}0](47,22)(53,22)
\drawline[AHnb=0](2,23)(2,21)
\drawline[AHnb=0](17,23)(17,21)
\drawline[AHnb=0](32,23)(32,21)
\drawline[AHnb=0](47,23)(47,21)
\node[Nmarks=n](label)(9.5,24){$t$}
\node[Nmarks=n](label)(24.5,24){$r$}
\node[Nmarks=n](label)(39.5,24){$r$}

\drawedge[ELpos=40, ELside=l, curvedepth=0](n0,n1){}  
\drawedge[ELpos=40, ELside=r, curvedepth=0](n0,n2){}  
\drawedge[ELpos=40, ELside=l, curvedepth=0](n1,n3){}
\drawedge[ELpos=40, ELside=l, curvedepth=0](n2,n3){}
\drawedge[ELpos=40, ELside=l, curvedepth=0](n2,n4){}
\drawedge[ELpos=40, ELside=l, curvedepth=0](n3,n5){}
\drawedge[ELpos=40, ELside=l, curvedepth=0](n4,n5){}
\drawedge[ELpos=40, ELside=l, curvedepth=0](n4,n6){}


\end{picture}}
  \hrule
  \caption{Sketch of the outcome of almost-sure eventually synchronizing strategies (with shifts). \label{fig:shift-strat}}
\end{figure}
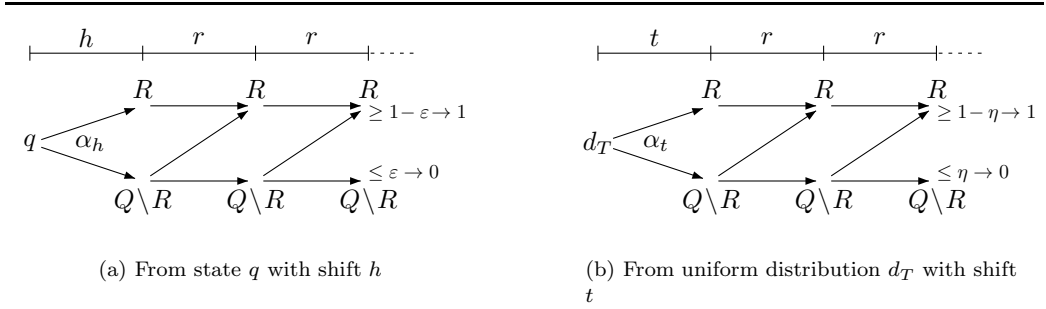

We now show that 
there exists an almost-sure winning strategy for the weakly synchronizing 
objective in target~$T$. 
Recall that $\Pre^r(R) = R$ and thus once some probability mass $p$ is in $R$,
it is possible to ensure that the probability mass in $R$ after $r$ steps
is at least $p$, and thus that (with period $r$) the probability in $R$ does not
decrease. 
By the result of Lemma~\ref{lem:lssr-assr}, almost-sure winning for eventually 
synchronizing in $R$ implies that there exists a strategy $\alpha$ such that the 
probability in $R$ tends to $1$ at periodic positions: 
for some $0 \leq h < r$ the strategy $\alpha$ is \emph{almost-sure 
eventually synchronizing in $R$ with shift $h$}, that is 
$\forall \epsilon > 0 \cdot \exists N \cdot \forall n \geq N: n \equiv h \mod r \implies 
\M^{\alpha}_{n}(R) \geq 1 - \epsilon$. We also say that the initial distribution
$d_0 = \M^{\alpha}_{0}$ is almost-sure eventually synchronizing in $R$ with shift $h$.
Almost-sure eventually synchronizing strategies with shift are illustrated in 
\figurename~\ref{fig:shift-strat}.

\paragraph{Claim 2} 
{\it
\begin{itemize}
\item[($\star$)] If $\M^{\alpha}_{0}$ is almost-sure eventually synchronizing 
in $R$ with some shift $h$, then $\M^{\alpha}_{i}$ 
is almost-sure eventually synchronizing in $R$ with shift $h-i \mod r$.
\item[($\star\star$)] Let $t$ such that $d_T$ is almost-sure eventually synchronizing in $R$ with shift $t$.
If a distribution is almost-sure eventually synchronizing in $R$ with some shift $h$, then 
it is also almost-sure eventually synchronizing in $R$ with shift $h + k + t \mod r$
(where we chose $k$ such that $R = \Pre^{k}(T)$).
\end{itemize}
}

\begin{figure}[t]
\begin{center}
    \begin{picture}(120,80)(0,2)

\gasset{Nframe=n, Nh=2, Nw=2, Nmr=0}

\node[Nmarks=n, Nw=3, Nh=3](n0)(2,5){$q$}

\node[Nmarks=n, ExtNL=y, NLangle=90](na1)(17,60){$R$}
\node[Nmarks=n, ExtNL=y, NLangle=270](nd1)(17,5){$Q \!\setminus\! R$}
\node[Nmarks=n, ExtNL=y, NLangle=90](nb1)(37,40){$R$}
\node[Nmarks=n, ExtNL=y, NLangle=270](nd2)(37,5){$Q \!\setminus\! R$}
\node[Nmarks=n, ExtNL=y, NLangle=90](nc1)(57,20){$R$}
\node[Nmarks=n, ExtNL=y, NLangle=270](nd3)(57,5){$Q \!\setminus\! R$}

\drawpolygon[Nframe=n, AHnb=0, arcradius=0, fillcolor=gray](62,5)(92,1)(92,9)

\node[Nmarks=n, ExtNL=y, NLangle=90](na2)(32,60){$T$}
\node[Nmarks=n, ExtNL=y, NLangle=90](nb2)(52,40){$T$}
\node[Nmarks=n, ExtNL=y, NLangle=90](nc2)(72,20){$T$}

\drawline[Nframe=y, AHnb=0, arcradius=3, linecolor=red](0,0)(0,10)(16,68)(64,20)(64,0)(0,0)
\drawpolygon[Nframe=y, AHnb=0, arcradius=3, linecolor=red](23,61)(47,70)(92,70)(92,50)(44,50)
\drawpolygon[Nframe=y, AHnb=0, arcradius=3, linecolor=red](43,41)(67,49.3)(92,49.3)(92,30)(64,30)
\drawpolygon[Nframe=y, AHnb=0, arcradius=3, linecolor=red](63,21)(87,29.5)(92,29.5)(92,10)(84,10)

\drawline[Nframe=y, AHnb=0, arcradius=3, linecolor=blue, dash={1.1 1.1}0](67,75)(47,69)
\drawline[Nframe=y, AHnb=0, arcradius=3, linecolor=blue, dash={1.1 1.28}0](67,75)(67,69)
\drawline[Nframe=y, AHnb=0, arcradius=3, linecolor=blue, dash={1.1 1.1}0](67,75)(87,69)
\node[Nmarks=n](shift)(67,77){shift $h+k+t$}

\node[Nmarks=n, ExtNL=y, NLangle=90](na2r)(47,65){\blue{$R$}}
\node[Nmarks=n, ExtNL=y, NLangle=270](na2q)(47,55){$Q \!\setminus\! R$}
\node[Nmarks=n, ExtNL=y, NLangle=90](na3r)(67,65){\blue{$R$}}
\node[Nmarks=n, ExtNL=y, NLangle=270](na3q)(67,55){$Q \!\setminus\! R$}
\node[Nmarks=n, ExtNL=y, NLangle=90](na4r)(87,65){\blue{$R$}}
\node[Nmarks=n, ExtNL=y, NLangle=270](na4q)(87,55){$Q \!\setminus\! R$}

\node[Nmarks=n, ExtNL=y, NLangle=90](nb2r)(67,45){\blue{$R$}}
\node[Nmarks=n, ExtNL=y, NLangle=270](nb2q)(67,35){$Q \!\setminus\! R$}
\node[Nmarks=n, ExtNL=y, NLangle=90](nb3r)(87,45){\blue{$R$}}
\node[Nmarks=n, ExtNL=y, NLangle=270](nb3q)(87,35){$Q \!\setminus\! R$}

\node[Nmarks=n, ExtNL=y, NLangle=90](nc2r)(87,25){\blue{$R$}}
\node[Nmarks=n, ExtNL=y, NLangle=270](nc2q)(87,15){$Q \!\setminus\! R$}

\node[Nmarks=n, ExtNL=n](dummy)(10,56){\rotatebox{75}{\red{Fig.~\ref{fig:shift-strategy}}}}
\node[Nmarks=n, ExtNL=n](dummy)(33,66.5){\rotatebox{21}{{\scriptsize \red{Fig.~\ref{fig:shift-strategy2}}}}}
\node[Nmarks=n, ExtNL=n](dummy)(53,46.5){\rotatebox{21}{{\scriptsize \red{Fig.~\ref{fig:shift-strategy2}}}}}
\node[Nmarks=n, ExtNL=n](dummy)(73,26.5){\rotatebox{21}{{\scriptsize \red{Fig.~\ref{fig:shift-strategy2}}}}}

\node[Nmarks=n, ExtNL=n](dummy)(20,55){\red{$p_1$}}
\node[Nmarks=n, ExtNL=n](dummy)(40,35){\red{$p_2$}}
\node[Nmarks=n, ExtNL=n](dummy)(60,15){\red{$p_3$}}

\put(94,65){\makebox(0,0)[l]{$\geq (1-\eta)\cdot p_1$}}
\put(94,55){\makebox(0,0)[l]{$\leq \eta\cdot p_1$}}
\put(94,45){\makebox(0,0)[l]{$\geq (1-\eta)\cdot p_2$}}
\put(94,35){\makebox(0,0)[l]{$\leq \eta\cdot p_2$}}
\put(94,25){\makebox(0,0)[l]{$\geq (1-\eta)\cdot p_3$}}
\put(94,15){\makebox(0,0)[l]{$\leq \eta\cdot p_3$}}

\put(94,5){\makebox(0,0)[l]{\red{$\leq \epsilon$}}}

\drawedge[ELpos=50, ELside=l, curvedepth=0](n0,na1){$h$}
\drawedge[ELpos=50, ELside=l, curvedepth=0](n0,nd1){$h$}
\drawedge[ELpos=50, ELside=l, curvedepth=0](nd1,nb1){$r$}
\drawedge[ELpos=50, ELside=l, curvedepth=0](nd1,nd2){$r$}
\drawedge[ELpos=50, ELside=l, curvedepth=0](nd2,nc1){$r$}
\drawedge[ELpos=50, ELside=l, curvedepth=0](nd2,nd3){$r$}

\drawedge[ELpos=50, ELside=l, curvedepth=0](na1,na2){$k$}
\drawedge[ELpos=50, ELside=r, curvedepth=0](na1,na2){sure}
\drawedge[ELpos=50, ELside=l, curvedepth=0](nb1,nb2){$k$}
\drawedge[ELpos=50, ELside=r, curvedepth=0](nb1,nb2){sure}
\drawedge[ELpos=50, ELside=l, curvedepth=0](nc1,nc2){$k$}
\drawedge[ELpos=50, ELside=r, curvedepth=0](nc1,nc2){sure}

\drawedge[ELpos=50, ELside=l, curvedepth=0](na2,na2r){$t$}
\drawedge[ELpos=50, ELside=r, curvedepth=0](na2,na2q){$t$}
\drawedge[ELpos=50, ELside=l, curvedepth=0](na2r,na3r){$r$}
\drawedge[ELpos=50, ELside=l, curvedepth=0](na2q,na3r){$r$}
\drawedge[ELpos=50, ELside=r, curvedepth=0](na2q,na3q){$r$}
\drawedge[ELpos=50, ELside=l, curvedepth=0](na3r,na4r){$r$}
\drawedge[ELpos=50, ELside=l, curvedepth=0](na3q,na4r){$r$}
\drawedge[ELpos=50, ELside=r, curvedepth=0](na3q,na4q){$r$}

\drawedge[ELpos=50, ELside=l, curvedepth=0](nb2,nb2r){$t$}
\drawedge[ELpos=50, ELside=r, curvedepth=0](nb2,nb2q){$t$}
\drawedge[ELpos=50, ELside=l, curvedepth=0](nb2r,nb3r){$r$}
\drawedge[ELpos=50, ELside=l, curvedepth=0](nb2q,nb3r){$r$}
\drawedge[ELpos=50, ELside=r, curvedepth=0](nb2q,nb3q){$r$}

\drawedge[ELpos=50, ELside=l, curvedepth=0](nc2,nc2r){$t$}
\drawedge[ELpos=50, ELside=r, curvedepth=0](nc2,nc2q){$t$}


\end{picture}
\end{center}
    \caption{Proof of Claim~2($\star\star$) for Theorem~\ref{theo:weakly-ls-is-as}.\label{fig:shift-claim}}
\end{figure}
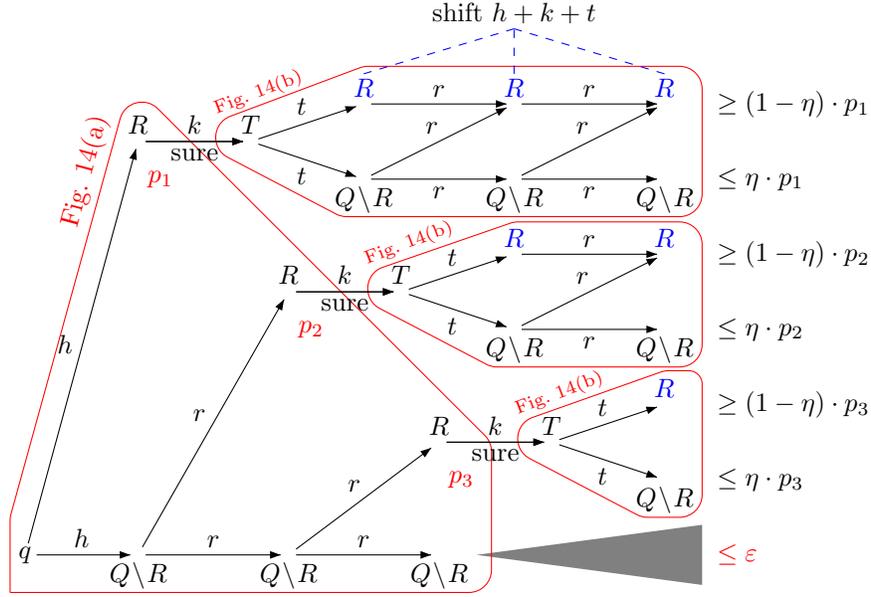

\paragraph{Proof of Claim~2}
The result ($\star$) immediately follows from the definition of shift, 
and we prove ($\star\star$) as follows.
We show that almost-sure eventually synchronizing in $R$ with shift $h$
implies almost-sure eventually synchronizing in $R$ with shift $h + k + t \mod r$.
The argument is illustrated in \figurename~\ref{fig:shift-claim}. 
Intuitively, the probability mass
that is in $R$ with shift $h$ can be injected in $T$ in $k$ steps, and then from $T$
we can play an almost-sure eventually synchronizing strategy in target $R$ with shift $t$,
thus a total shift of $h + k + t \mod r$.
Precisely, an almost-sure winning strategy $\alpha$ is constructed as follows 
(\figurename~\ref{fig:shift-claim}): 
\begin{itemize}
\item given a finite prefix of play~$\rho$, if there is no state $q \in R$ that 
occurs in $\rho$ at a position $n \equiv h \mod r$, then play in $\rho$ according to 
the almost-sure winning strategy $\alpha_h$ for eventually synchronizing in $R$ with shift $h$;
\item otherwise, 
  \begin{itemize}
  \item if there is no $q \in T$ that occurs in $\rho$ at a position $n \equiv h+k \mod r$, then
we play according to a sure winning strategy $\alpha_{sure}$ for eventually synchronizing in $T$,
  \item and otherwise we play according to an almost-sure winning strategy $\alpha_t$ from $T$ for 
eventually synchronizing in $R$ with shift $t$.
  \end{itemize}
\end{itemize}
To show that $\alpha$ is almost-sure eventually synchronizing in $R$ with shift $h + k + t$,
note that $\alpha_h$ ensures with probability~$1$ that $R$ is reached at positions
$n \equiv h \mod r$ (see \figurename~\ref{fig:shift-claim}). Consider positions
$h,h+r,h+2r,\dots$ and the probability mass $p_i$ in $R$ at position $h+ir$.
Then for all $\epsilon >0$, by considering sufficiently long sequence of positions, 
we have $\sum_i p_i \geq 1-\epsilon$ (\figurename~\ref{fig:shift-strategy}). Since $\alpha_{sure}$ is sure eventually synchronizing in $T$,
we also have probability mass at least $p_i$ in $T$ at position $h+k+ir$.
From the states in $T$ the strategy $\alpha_t$ ensures with probability~$1$ 
that $R$ is reached at positions $h+k+t \mod r$, thus for all $\eta >0$,
by considering sufficiently long sequence of positions (and \figurename~\ref{fig:shift-strategy2}), we have probability
mass at least $\sum_i p_i \cdot (1-\eta) \geq (1-\epsilon)\cdot(1-\eta)$ at
in $R$ at some position $h+k+t+ir$, thus with shift $h+k+t$ (see also \figurename~\ref{fig:shift-claim}). This concludes the proof of Claim~2.
\smallskip

\begin{figure}[t]
\begin{center}
    \begin{picture}(135,32)(5,0)

\gasset{Nframe=n, Nh=4, Nw=8, Nmr=0}

\node[Nmarks=n, Nw=6](n0)(7,20){$q_{\init}$}

\node[Nmarks=n, ExtNL=n, NLangle=90](n1)(27,28){$R$}
\nodelabel[ExtNL=y, NLangle=270, NLdist=1](n1){{\small $1\!-\!\epsilon$}}
\node[Nmarks=n, ExtNL=n, NLangle=270](n2)(27,12){$Q \!\setminus\!\! R$}
\nodelabel[ExtNL=y, NLangle=270, NLdist=1](n2){{\small $\epsilon$}}
\node[Nmarks=n, ExtNL=n, NLangle=90](n3)(52,28){$T$}
\nodelabel[ExtNL=y, NLangle=270, NLdist=1](n3){{\quad \small $\geq 1\!-\!\epsilon$}}
\node[Nmarks=n, ExtNL=n, NLangle=270](n4)(52,12){$Q \!\setminus\!\! T$}

\put(13,20){\makebox(0,0)[l]{$\alpha_{\epsilon}$}}

\node[Nmarks=n, ExtNL=n, NLangle=270](label)(27,5){$\underbrace{\phantom{Q \!\setminus\!\! R}}_{d}$}
\node[Nmarks=n, ExtNL=n, NLangle=270](label)(52,5){$\underbrace{\phantom{Q \!\setminus\!\! T}}_{d'}$}

\put(57,26){\makebox(0,0)[l]{\begin{tabular}{ll}$\cdots$ & \!\!\!\!\!almost-sure eventually synchronizing in $R$ \\ 
& \!\!\!\!\!with shift $t$ \end{tabular}}}
\put(57,10){\makebox(0,0)[l]{\begin{tabular}{llr}$\cdots$ & \multicolumn{2}{l}{\!\!\!\!\!almost-sure eventually synchronizing in $R$} \\ 
& \!\!\!\!\!with shift $h- (h+k)$ & by ($\star$)\end{tabular}}}
\put(57,0){\makebox(0,0)[l]{\begin{tabular}{llr}\phantom{$\cdots$} & \multicolumn{2}{l}{\!\!\!\!\!almost-sure eventually synchronizing in $R$} \\ 
& \!\!\!\!\!with shift $h- (h+k) + k + t = t$ & by ($\star\star$)\end{tabular}}}

\drawedge[ELpos=50, ELside=l, curvedepth=0](n0,n1){{\small $h$}}
\drawedge[ELpos=50, ELside=r, curvedepth=0](n0,n2){{\small $h$}}
\drawedge[ELpos=50, ELside=l, curvedepth=0](n1,n3){{\small $k$}}
\drawedge[ELpos=50, ELside=r, curvedepth=0](n1,n3){{\small sure}}
\drawedge[ELpos=40, ELside=l, curvedepth=0](n2,n3){{\small $k$}}
\drawedge[ELpos=38, ELside=r, curvedepth=0](n2,n4){{\small $k$}}


\end{picture}
\end{center}
    \caption{Construction of an almost-sure weakly synchronizing strategy.\label{fig:as-weakly-strategy}}
\end{figure}
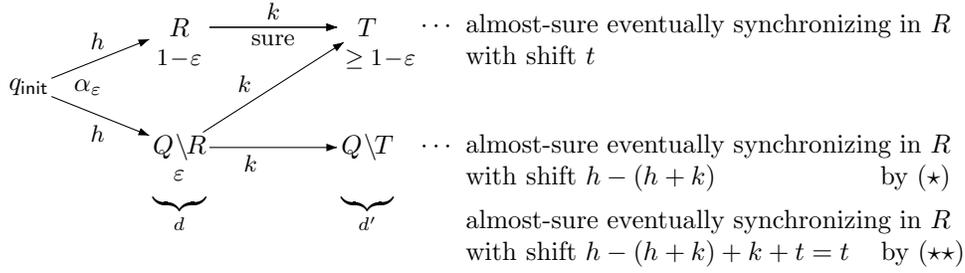

\paragraph{Construction of an almost-sure winning strategy}
We construct strategies $\alpha_\epsilon$ for $\epsilon > 0$ that ensure,
from a distribution that is almost-sure eventually 
synchronizing in $R$ (with some shift $h$), that after finitely many steps,
a distribution $d'$ is reached such that $d'(T) \geq 1-\epsilon$ and 
$d'$ is almost-sure eventually synchronizing in $R$ (with some shift $h'$).
Since $q_{\init}$ is almost-sure eventually synchronizing in $R$ (with some shift $h$),
it follows that the strategy $\alpha_{as}$ that plays successively 
the strategies (each for finitely many steps) $\alpha_{\frac{1}{2}}$, 
$\alpha_{\frac{1}{4}}$, $\alpha_{\frac{1}{8}}, \dots$ is 
almost-sure winning from $q_{\init}$ for the weakly synchronizing objective in target~$T$. 

We define the strategies $\alpha_\epsilon$ as follows (the construction is 
illustrated in \figurename~\ref{fig:as-weakly-strategy}).
Given an initial
distribution that is almost-sure eventually synchronizing in $R$ with a shift 
$h$ and given $\epsilon > 0$, let $\alpha_\epsilon$ be the strategy that plays
according to the almost-sure winning strategy $\alpha_h$ for eventually 
synchronizing in $R$ with shift $h$ for a number of steps $n \equiv h \mod r$ until a distribution 
$d$ is reached such that $d(R) \geq 1-\epsilon$, and then from $d$ it plays
according to a sure winning strategy $\alpha_{sure}$ for eventually synchronizing in $T$
from the states in $R$ (for $k$ steps), and keeps playing according to $\alpha_h$ from 
the states in $Q \setminus R$ (for $k$ steps). The distribution
$d'$ reached from $d$ after $k$ steps is such that $d'(T) \geq 1-\epsilon$ and we claim
that it is almost-sure eventually synchronizing in $R$ with shift $t$. 
This holds by definition of $\alpha_t$ from the states in $\Supp(d') \cap T$,
and by ($\star$) the states in $\Supp(d') \setminus T$ are almost-sure eventually synchronizing 
in $R$ with shift $h-(h+k) \mod r$, and by ($\star\star$) with shift $h-(h+k)+k+t = t$.

It follows that the strategy $\alpha_{as}$ is well-defined and ensures,
for all $\epsilon >0$, that the probability mass in $T$ is infinitely
often at least $1-\epsilon$, thus is almost-sure weakly synchronizing in $T$. 
This concludes the proof of Theorem~\ref{theo:weakly-ls-is-as}.
\qed
\end{proof}

\section{Strongly Synchronizing}\label{sec:strongly-synch}

The strongly synchronizing objective is reminiscent of a coB\"uchi 
objective in the distribution-based semantics: with function $\fsum_T$
it requires that in the sequence of distributions of an MDP $\M$ under strategy $\alpha$
we have $\liminf_{n \to \infty} \M^{\alpha}_n(T) = 1$ (and that $\M^{\alpha}_n(T) = 1$
from some point on in the case of sure winning).

We show that the membership problem
for strongly synchronizing objectives can be solved in polynomial time,
for all winning modes, and both with function $\fmax_T$ (Section~\ref{sec:strongly-max}) and
function $\fsum_T$ (Section~\ref{sec:strongly-sum}). 
We show that linear-size memory is necessary in general
for $\fmax_T$, and memoryless strategies are sufficient for $\fsum_T$.
It follows from our results that the limit-sure and almost-sure winning
modes coincide for strongly synchronizing.

\subsection{Strongly synchronizing with function $\fmax$}\label{sec:strongly-max}

First, note
that for strongly synchronizing the membership problem with function $\fmax_T$ 
reduces to the membership problem with function $\fmax_Q$ where $Q$
is the entire state space, by a construction similar to the proof of 
Lemma~\ref{lem:weakly-max-sum}: states in $Q \setminus T$ are duplicated,
ensuring that only states in $T$ are used to accumulate probability.

The strongly synchronizing objective with function $\fmax_Q$ requires that from 
some point on, almost all the probability mass is at every step in a single state.
Intuitively, the sequence of states that contain almost all the probability 
corresponds to a sequence of deterministic transitions in the MDP, and thus 
eventually to a cycle of deterministic transitions.

\begin{figure}[t]

\begin{center}
    \begin{picture}(60,50)

\node[Nmarks=i,iangle=90](n0)(25,20){$q_{\init}$}
\node(n1)(55,20){$q_1$}
\node(n2)(55,35){$q_2$}
\node(n4)(63,5){$q_4$}
\node(n3)(47,5){$q_3$}

\node(n5)(20,5){$q_5$}
\node(n6)(1,5){$q_6$}
\node(n7)(-4,20){$q_7$}
\node(n8)(11,32){$q_8$}

\drawedge[ELdist=.5](n0,n1){$a: \frac{1}{2}$}
\drawedge[ELdist=.5,ELside=l, curvedepth=3](n0,n5){$a: \frac{1}{2}$}

\drawedge[ELdist=.5,ELside=l, curvedepth=5](n1,n2){$a$}
\drawedge[ELdist=.5,ELside=l, curvedepth=5](n2,n1){$a,b$}
\drawedge[ELdist=.5,ELside=r, curvedepth=-5](n1,n3){$b$}
\drawedge[ELdist=.5,ELside=r, curvedepth=-5](n3,n4){$a,b$}
\drawedge[ELdist=.5,ELside=r, curvedepth=-5](n4,n1){$a,b$}

\drawedge[ELdist=.5,ELside=l, curvedepth=4](n5,n6){$a,b$}
\drawedge[ELdist=.5,ELside=l, curvedepth=4](n6,n7){$a,b$}
\drawedge[ELdist=.5,ELside=l, curvedepth=4](n7,n8){$a,b$}
\drawedge[ELdist=.5,ELside=l, curvedepth=4](n8,n0){$a,b$}
\end{picture}
\end{center}
 \caption{An example to show $q_{\init} \in \winas{strong}(max_Q)$ 
reduces to synchronized reachability of a state in a simple deterministic cycle.
\label{fig:strong-almost-max}}

\end{figure}
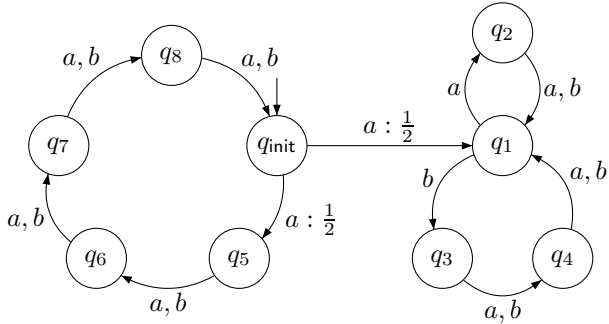

Consider the MDP in~\figurename~\ref{fig:strong-almost-max} with initial state $q_{\init}$:
all transitions are deterministic except from~$q_{\init}$ 
where on both actions $a$ and $b$,
the successors are~$q_1$ and~$q_5$ with probability~$\frac{1}{2}$.
The strategic choice is only relevant in~$q_1$ where
$\delta(q_1,a)(q_2)=1$
and 
$\delta(q_1,b)(q_3)=1$. 
We present a strategy such that
the sequence of states that contain almost all the probability is 
the cycle~$q_1 q_2 q_1 q_3 q_4 q_1$ of deterministic transitions.

The state $q_{\init}$ is almost-sure strongly synchronizing (according to function $\fmax$)
with the strategy~$\alpha$ defined as follows, 
for all paths $\rho$ such that $\Last(\rho) = q_1$:
\begin{itemize}
\item if the number of occurrences of $q_1$ in $\rho$ is odd (i.e., the length
of $\rho$ is $1$ modulo $5$), then play action $a$;

\item if the number of occurrences of $q_1$ in $\rho$ is even (i.e., the length
of $\rho$ is $3$ modulo $5$), then play action $b$.
\end{itemize}
The strategy~$\alpha$ ensures the probability mass injected from $q_{\init}$ in 
$q_1$ after every other~$5$ steps loops in the cycle~$q_1 q_2 q_1 q_3 q_4  q_1$ (with length~$5$).  
Hence, the probability mass from $q_{\init}$ is always injected
in $q_1$ synchronously (i.e., when the probability mass in the cycle 
is also in $q_1$).

It follows that after $5i$ steps, the probability mass in $q_{\init}$ 
is $\frac{1}{2^i}$ and the probability mass in $q_1$ is $1 - \frac{1}{2^i}$.
Considering $i \to \infty$, we then get $\liminf_{n\to \infty}\norm{\M^{\alpha}_{n}} = 1$ 
and $q_{\init} \in \winas{strong}(\fmax)$.
Note  that only the states in the cycle~$q_1 q_2 q_1 q_3 q_4  q_1$ (of deterministic transitions)
are used to accumulate the probability mass tending to~$1$.

Cycles consisting of deterministic transitions are keys to decide strongly synchronizing.
A \emph{deterministic cycle} of length $\ell \geq 1$ in an MDP $\M$ is a finite sequence 
$\q_0 \q_1 \dots \q_{\ell}$ of states such that $\q_0 = \q_{\ell}$ and 
for all $0 \leq i < \ell$, there exists an action $a_i$ such that $\delta(q_{i},a_i)(q_{i+1}) = 1$.
The cycle is \emph{simple} if $\q_i \neq \q_j$ for all $1 \leq i < j \leq \ell$. 

We show that sure (resp., almost-sure and limit-sure) strongly synchronizing 
is equivalent to sure (resp., almost-sure and limit-sure) reachability to a 
state in a \emph{simple} deterministic cycle, with the requirement that the 
state can be reached in a synchronized way 
(i.e., by finite paths whose lengths are congruent modulo the length $\ell$ of the cycle). 

In the MDP of \figurename~\ref{fig:strong-almost-max}, we can construct
an almost-sure strongly synchronizing strategy $\beta$ that accumulates 
the probability mass only in the simple cycle~$q_1 q_3 q_4 q_1$. 
The strategy~$\beta$ is defined as follows, 
for all paths $\rho$ such that $\Last(\rho) = q_1$:
\begin{itemize}
\item if the length of $\rho$ is $0$ modulo $3$, then play action $b$;

\item if the length of $\rho$ is $1$ or $2$ modulo $3$, then play action $a$.
\end{itemize}

Note that if the length of $\rho$ is a multiple of $3$ and the action $b$
is played, then on the next visit to $q_1$ the length of the path is
also a multiple of $3$, and the action $b$ is played again. Hence, once a probability
mass follows the cycle~$q_1 q_3 q_4 q_1$, it will follow this cycle forever.
Whenever probability mass is injected in $q_1$ (from $q_{\init}$) on a path $\rho$ of length $1$ or $2$ modulo $3$,
the action $a$ is played to visit the other cycle~$q_1 q_2 q_1$ until getting
back to $q_1$ with a path whose length is a multiple of $3$. The probability mass is then
injected (synchronously) into the cycle~$q_1 q_3 q_4 q_1$ where eventually the probability
mass tends to $1$, thus the strategy $\beta$ is almost-sure strongly synchronizing
and it ensures with probability~$1$ that $q_1$ is reached with by paths whose 
length is a multiple of $3$.

We show in Lemma~\ref{lem: p-strong} that simple deterministic cycles are always 
sufficient for strongly synchronizing in MDPs, and that strongly synchronizing
reduces to a synchronized reachability problem of reaching a state $q_1$
of a simple deterministic cycle by paths of length that is a multiple of the length
$\ell$ of the cycle.
To check synchronized reachability, 

we keep track of a modulo-$\ell$ counter along the path.
Define the MDP $\M \times [\ell] = \tuple{Q',\Act,\delta'}$
where $Q' = Q \times \{0,1,\dots,\ell-1\}$ and 
$\delta'(\tuple{q,i},a)(\tuple{q',i-1}) = \delta(q,a)(q')$ 
(where $i-1$ is $\ell-1$ for $i=0$)
for all states $q,q'\in Q$, actions $a \in \Act$, and $0 \leq i \leq \ell-1$. 
Note that given a finite path $\rho = q_0 a_0 q_1 a_1\dots q_n$ in $\M$, 
there is a corresponding path $\rho' = \tuple{q_0,k_0} a_0 \tuple{q_1,k_1} a_1 \dots \tuple{q_n, k_n}$ 
in $\M \times [\ell]$ where $k_i = -i~mod~\ell$.
Since the sequence $k_0 k_1 \dots$ is uniquely defined, 
there is a clear bijection between the paths in $\M$ (starting from $q_0$)
and the paths in $\M \times [\ell]$ (starting from $\tuple{q_0,0}$)
that we often omit to apply and mention in the sequel.

\begin{lemma}\label{lem: p-strong}
Let $\eta$ be the smallest positive probability in the transitions of $\M$,
and let $\frac{1}{1+\eta}< p \leq 1$. 
There exists a strategy~$\alpha$ such that 
$\liminf_{n \to \infty} \norm{\M^{\alpha}_n}\geq p$
from an initial state~$q_{\init}$
if and only if 
there exist a simple deterministic 
cycle~$\q_0 \q_1 \dots \q_{\ell}$ in $\M$ 
and a strategy~$\beta$ in~$\M\times[\ell]$ 
such that $\Pr^{\beta}(\Diamond \{\tuple{\q_0,0}\}) \geq p$
from~$\tuple{q_{\init},0}$.
\end{lemma}

\begin{proof}
For the first direction of the lemma, assume that there exists a strategy~$\alpha$ 
such that $\liminf_{n \to \infty} \norm{\M^{\alpha}_n} \geq p$ from~$q_{\init}$.
Thus for all $\epsilon > 0$ (in particular, we consider $\epsilon < p - \frac{1}{1+\eta}$), 
there exists $k \in \nat$ such that for all $n \geq k$ we have $\norm{\M^{\alpha}_{n}} \geq p-\epsilon$,
and let $\p_n$ be a state such that $\M^{\alpha}_n(\p_n) \geq p-\epsilon$.
We claim that for all $n \geq k$, there exists an action $a \in \Act$ 
such that $\post(\p_{n},a) = \{\p_{n+1}\}$ i.e., there is a deterministic transition from 
$\p_{n}$ to $\p_{n+1}$. Assume towards contradiction that for some $n \geq k$,
for all~$a \in \Act$ there exists $q_a \neq \p_{n+1}$ such that $q_a \in \post(\p_{n},a)$.
Then no matter the actions played by $\alpha$ at step $n$, we have
$\M^{\alpha}_{n+1}(\{q_a \mid a \in \Act\}) \geq \M^{\alpha}_{n}(\p_n) \cdot \eta \geq (p-\epsilon) \cdot \eta$,
and since $\p_{n+1} \neq q_a$ for all $a \in \Act$, it follows that 
$$\M^{\alpha}_{n+1}(\p_{n+1}) \leq 1 - \M^{\alpha}_{n+1}(\{q_a \mid a \in \Act\}) 
\leq 1 - (p-\epsilon) \cdot \eta \leq 1 - \frac{\eta}{1+\eta} < p - \epsilon,$$
in contradiction with the fact that $\p_{n+1}$ is a state such that  
$\M^{\alpha}_{n+1}(\p_{n+1}) \geq p-\epsilon$. This concludes the argument
showing that for all $n\geq k$, there exists an action $a \in \Act$ 
such that $\post(\p_{n},a) = \{\p_{n+1}\}$. 

Now in the sequence $\p_k \p_{k+1} \dots$, we can extract a simple (and deterministic)
cycle $\C = \p_i \p_{i+1} \dots \p_{i + \ell}$ since the state space  
is finite. Let $\q_0 = \p_{i+j}$ where $j \leq \ell$ is such that 
$i+j~mod~\ell = 0$. Then $\q_0$ is on a simple deterministic cycle, and
is reachable after a multiple of $\ell$ steps with probability at least $p - \epsilon$
by a strategy $\beta$ in~$\M\times[\ell]$ that copies the strategy $\alpha$.
Hence, we have $\Pr^{\beta}(\Diamond\{\tuple{\q_0,0}\}) \geq p - \epsilon$ from~$\tuple{q_{\init},0}$.
Since for every $\epsilon > 0$, we can find such a cycle and state $\q_0$, 
and since the state space is finite (as well as the number of simple cycles),
it follows that there is a cycle $\C$ and state $\q_0$ in $\C$ such that 
for all $\epsilon > 0$ we have $\Pr^{\beta}(\Diamond\{\tuple{\q_0,0}\}) \geq p - \epsilon$,
and thus $\Pr^{\beta}(\Diamond\{\tuple{\q_0,0}\}) \geq p$.

For the second direction of the lemma, assume that there exist a simple deterministic 
cycle~$\q_0 \q_1 \dots \q_{\ell}$ and a strategy~$\beta$ in $\M \times [\ell]$ 
that ensures the target set $\{\tuple{\q_0,0}\}$ is reached with probability at least $p$  
from~$\tuple{q_{\init},0}$. Since randomization is not necessary for 
reachability objectives in MDPs, we can assume that $\beta$ is a pure strategy.
We show that there exists a strategy~$\alpha$ such that 
$\liminf_{n \to \infty} \norm{\M^{\alpha}_n}\geq p$ from~$q_{\init}$.
From $\beta$, we construct a pure strategy $\alpha$ in $\M$. 
Given $\rho = q_0 a_0 q_1 a_1 \dots q_n$, we define $\alpha(\rho)$ as follows: 
if $q_n = \q_{n~mod~\ell}$, then there exists an action $a$ such that 
$\post(q_n,a) = \{\q_{n+1~mod~\ell}\}$ and we define 
$\alpha(\rho) = a$, otherwise let $\alpha(\rho) = \beta(\rho)$.	
Thus $\alpha$ mimics $\beta$ until a state $\q_{k}$ of the cycle is reached at step $n$ such that 
$k = n~mod~\ell$, and then $\alpha$ switches to always playing actions
that keeps $\M$ in the simple deterministic cycle~$\q_0 \q_1 \dots \q_{\ell}$.
Note that $\alpha$ is a pure strategy.

We claim that given $\epsilon>0$ there exists $k$ such that 
for all $n\geq k$, we have $\norm{\M^{\alpha}_{n}} \geq p-\epsilon$,
which entails that $\liminf_{n \to \infty} \norm{\M^{\alpha}_n}\geq p$ from $q_{\init}$
and concludes the proof. To show the claim, since $\Pr^{\beta}(\Diamond \{\tuple{\q_0,0}\})\geq p$,
consider $k$ such that $\Pr^{\beta}(\Diamond^{\leq k} \{\tuple{\q_0,0}\}) \geq p-\epsilon$,
and for $i = 1, 2, \dots, \ell$, let $R_i = \{\tuple{\q_{i}, \ell - i}\}$. 
Note that $R_{\ell} = \{\tuple{\q_0,0}\}$.
Then trivially $\Pr^{\beta}(\Diamond^{\leq k} \bigcup_{i=1}^{\ell} R_i) \geq p-\epsilon$
and since $\alpha$ agrees with $\beta$ on all finite paths  
that do not (yet) visit $\bigcup_{i=1}^{\ell} R_i$, given a path $\rho$ that
visits $\bigcup_{i=1}^{\ell} R_i$ (for the first time), only  actions that keep $\M$
in the simple cycle $\q_0 \q_1 \dots \q_{\ell}$
are played by $\alpha$ and thus all continuations of $\rho$ in the outcome of $\alpha$ 
will visit $\q_0$ after a multiple of $\ell$ steps, say $j \cdot \ell$ steps (in total).  
Since next, $\alpha$ will always play actions that keeps $\M$ looping through the  cycle 
$\q_0 \q_1 \dots \q_{\ell}$, we have $\M^{\alpha}_{j \cdot \ell + i}(\q_{i}) \geq p-\epsilon$
for all $0 \leq i < \ell$, and thus $\norm{\M^{\alpha}_n} \geq p-\epsilon$ for all $n \geq j \cdot \ell$. 
\qed
\end{proof}\medskip

It follows directly from Lemma~\ref{lem: p-strong} with $p=1$ that 
almost-sure strongly synchronizing is equivalent to almost-sure
reachability to a deterministic cycle in~$\M \times [\ell]$. The same equivalence
holds for the sure and limit-sure winning modes.

\begin{lemma}\label{lem: strong}
A state $q_{\init}$ is sure (resp., almost-sure or limit-sure) winning
for the strongly synchronizing objective (according to $\fmax_Q$) in $\M$
if and only if there exists a simple deterministic cycle~$\q_0 \q_1 \dots \q_{\ell}$ 
such that $\tuple{q_{\init},0}$ is sure (resp., almost-sure or limit-sure) 
winning for the reachability objective~$\Diamond \{\tuple{\q_0,0}\}$
in $\M \times [\ell]$.
\end{lemma}

\begin{proof}
We consider the three winning modes:

{\bf (1) sure winning mode.}
The proof is similar to the proof of Lemma~\ref{lem: p-strong}.
For the first direction, given a strategy~$\alpha$ and $k$
such that for all $n \geq k$ we have $\norm{\M^{\alpha}_n}=1$
from the initial state~$q_{\init}$, we can construct a sequence
$\p_k \p_{k+1} \dots$ of states where there is deterministic transition
from $\p_n$ to $\p_{n+1}$ for all $n \geq k$ (let $\p_n$ be the state such 
that  $\M^{\alpha}_n(\p_n) = 1$). This sequence contains a 
simple deterministic cycle and a state $\q_0$ in this cycle occurs in 
the sequence at a position $\p_{j \cdot \ell}$ that is a multiple of 
the length $\ell$ of the cycle. Hence, the strategy $\alpha$ played in $\M\times[\ell]$
ensures to reach $\tuple{\q_0,0}$ surely from~$\tuple{q_{\init},0}$.

For the second direction, if a strategy $\beta$ ensures to
reach a state $\tuple{\q_0,0}$ in $\M\times[\ell]$ where $\q_0$ belongs
to a simple deterministic cycle of length $\ell$, then a strategy $\alpha$
that mimics $\beta$ until $\tuple{\q_0,0}$ is reached, and then switches
to playing actions to follow the simple cycle, ensures sure strongly 
synchronizing with function $\fmax_Q$ in $\M$. Note that $\alpha$ is a pure strategy.

{\bf (2) almost-sure winning mode.} 
This case follows from Lemma~\ref{lem: p-strong} with $p=1$.

{\bf (3) limit-sure winning mode:} 
For the first direction, if $q_{\init}$ is limit-sure
winning for the strongly synchronizing objective, then 
for all $\epsilon > 0$, there exists a strategy $\alpha$ such that 
$\liminf_{n \to \infty} \norm{\M^{\alpha_i}_n} \geq 1-\epsilon$.
By Lemma~\ref{lem: p-strong}, for a decreasing sequence $\epsilon_i \to 0$
such that $\epsilon_i < 1 - \frac{1}{1+\eta}$ there exist
a simple deterministic cycle~$\C_i$ of length $\ell_i$, 
a state $\q^i_0$ in $\C_i$, and a strategy~$\beta_i$ in~$\M \times[\ell_i]$ 
such that $\Pr^{\beta_i}(\Diamond \{\tuple{\q_0,0}\}) \geq 1-\epsilon_i$
from~$\tuple{q_{\init},0}$.
Since there is a finite number of simple deterministic cycles in $\M$,
some simple cycle $\C = \q_0 \q_1 \dots \q_{\ell}$ and state $\q_0$ 
occurs infinitely often in the sequence of $(\C_i, \q^i_0)$, and thus 
$\tuple{\q_{\init},0}$ is limit-sure winning for
the reachability objective~$\Diamond \{\tuple{\q_0,0}\})$ in $\M \times [\ell]$.

For the second direction, since limit-sure winning implies almost-sure winning
for reachability objectives in MDPs, it follows from case {\bf (2)} that 
$q_{\init}$ is almost-sure (and thus also limit-sure) winning
for the strongly synchronizing objective in $\M$.
\qed
\end{proof}\medskip

Since the winning regions of almost-sure and limit-sure winning
coincide for reachability objectives in MDPs~\cite{AHK07},
the next corollary follows from  Lemma~\ref{lem: strong}. 

\begin{corollary}\label{col: almost-limit-strong}
$\winlim{strongly}(\fmax_T)= \winas{strongly}(\fmax_T)$
for all target sets~$T$.
\end{corollary}

If there exists a cycle $\C$ satisfying the condition in Lemma~\ref{lem: strong},
then all cycles reachable from $\C$ in the graph $G$ of deterministic transitions
also satisfies the condition. Hence, it is sufficient to check the condition
for an arbitrary simple cycle in each strongly connected component (SCC)
of $G$. As shown in the next theorem, it follows that strongly synchronizing can be 
decided in polynomial time 
and the length of the cycle gives a linear bound on the memory needed to win.

\begin{theorem}\label{theo:strongly-max}
For the three winning modes of strongly synchronizing 
according to $\fmax_T$:

\begin{enumerate}
\item (Complexity). The membership problem is PTIME-complete.

\item (Memory). Linear memory is necessary and sufficient for both pure 
and randomized strategies, and pure strategies are sufficient.
\end{enumerate}
\end{theorem}

\begin{proof}
First, we prove the PTIME upper bound.
Given an MDP~$\M=\tuple{Q,\Act,\delta}$ and a state~$q_{\init}$, we say that 
a simple deterministic cycle~$\C = \q_0 \q_1 \dots \q_{\ell}$ is sure 
(resp., almost-sure, and limit-sure) winning from~$q_{\init}$ 
if $\tuple{q_{\init},0}$ is sure (resp., almost-sure, and limit-sure) winning 
for the reachability objective~$\Diamond \{\tuple{\q_0,0}\}$ in $\M \times [\ell]$.

We claim that if $\C$ is sure (resp., almost-sure, and limit-sure) winning 
from~$q_{\init}$, then so are all simple cycles $\C'$ 
reachable from $\C$ in the graph of deterministic transitions induced by $\M$.  
Given a strategy to reach a state $\q_0$ of $\C$ surely (resp., with probability $p$), 
we can use the path of deterministic transitions from $\C$ to $\C'$ to obtain
a strategy to reach a state $\q'_0$ of $\C'$ surely (resp., with probability $p$): since $\q_0$
is reached after a multiple of $\ell$ steps ($\{\tuple{\q_0,0}\}$ is reached in $\M \times [\ell]$),
we can let the probability mass loop through the cycle $\C$, and transfer it to $\C'$
after a number of steps that is also a multiple of $\ell'$, and then let it loop in $\C'$,
ensuring that $\tuple{\q'_0,0}$ is reached surely (resp., with probability $p$) in $\M \times [\ell']$.
This establishes the claim for the three winning modes. 

Using this claim and Lemma~\ref{lem: strong}, it suffices to decide sure 
(resp., almost-sure, and limit-sure) winning for one simple cycle in each bottom
SCC (reachable from~$q_{\init}$) of the graph of deterministic transitions. 
Since SCC decomposition for graphs, as well as sure, almost-sure, and limit-sure reachability for MDPs
can be computed in polynomial time, and the number of bottom SCCs is at most the size~$\abs{Q}$ 
of the graph, the PTIME upper bound for the membership problem follows.

For PTIME-hardness, the proof is by a reduction from 
the monotone Boolean circuit value problem, which is PTIME-complete~\cite{G77}.
This problem is to compute the output value of a given Boolean circuit consisting of
AND-gates, OR-gates, and fixed Boolean input values. 
From a circuit, we construct an MDP~$\M$ with actions $L$ and $R$, where the states 
correspond to the gates and input values of the circuit, and with three new absorbing states 
$q_1$, $q_2$, and $\sync$. The successors of an AND-gate $n_1 \land n_2$ 
are $n_1$ and $n_2$ with probability~$\frac{1}{2}$ on all actions, the successors of an OR-gate $n_1 \lor n_2$ 
are $n_1$ on action $L$, and $n_2$ on action $R$. On all actions, a node defining 
input value $1$ has unique successor $\sync$, and a node defining input value $0$ 
has successors $q_1$ and $q_2$ with probability~$\frac{1}{2}$.
Let $q_{\init}$ be the state corresponding to the output node.
Then $\M$ is sure (resp., almost-sure, limit-sure) strongly synchronizing 
(in $\sync$) from $q_{\init}$ if and only if the value of the circuit is $1$,
which establishes PTIME-hardness of strongly synchronizing in the three winning modes.

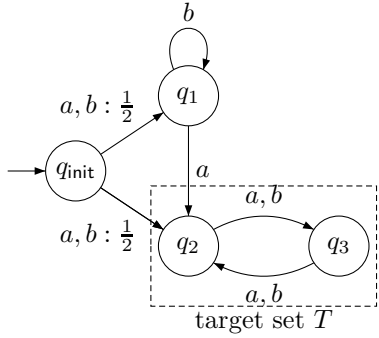
\begin{figure}[t]
\begin{center}
    \begin{picture}(50,42)

\node[Nmarks=i](n0)(10,21){$q_{\init}$}
\node[Nmarks=n](n1)(25,31){$q_1$}
\node[Nmarks=n](n2)(25,11){$q_2$}
\node[Nmarks=n](n3)(45,11){$q_3$}

\drawpolygon[dash={0.8 0.5}0](20,19)(50,19)(50,3)(20,3)
\node[Nframe=n](label)(35,1){target set $T$}

\drawedge[ELpos=40, ELside=l](n0,n1){$a,b: \frac{1}{2}$}
\drawedge[ELpos=40, ELside=r](n0,n2){$a,b: \frac{1}{2}$}
\drawedge(n0,n2){}
\drawloop[ELside=l,loopCW=y, loopangle=90, loopdiam=5](n1){$b$}
\drawedge(n1,n2){$a$}

\drawedge[ELpos=50, ELside=l, curvedepth=4](n2,n3){$a,b$}
\drawedge[ELpos=50, ELside=l, curvedepth=4](n3,n2){$a,b$}

\end{picture}
\end{center}
 \caption{An MDP where all strategies to win sure strongly synchronizing with 
function $\fmax_{\{q_2,q_3\}}$ require memory.\label{fig:strong-max-memory}}
\end{figure}

Finally, the result on memory requirement is established as follows.
Since memoryless strategies are sufficient for reachability objectives
in MDPs, it follows from the proof of Lemma~\ref{lem: p-strong} and 
Lemma~\ref{lem: strong} that the (memoryless) winning strategies in $\M \times[\ell]$ 
can be transferred to winning strategies with memory $\{0,1,\dots,\ell-1\}$ in $\M$.
Since $\ell \leq \abs{Q}$, linear-size memory is sufficient to win strongly 
synchronizing objectives. We present a family of MDPs $\M_n$ ($n \in \nat$)
that are sure winning for strongly synchronizing (according to $\fmax_Q$), 
and where the sure winning strategies require linear memory. 
The MDP $\M_2$ is shown in \figurename~\ref{fig:strong-max-memory},
and the MDP $\M_n$ is obtained from $\M_2$ by replacing the cycle $q_2 q_3$ of deterministic
transitions by a simple cycle of length $n$. Note that only in $q_1$ there is a relevant
strategic choice. Since both $q_1$ and $q_2$ contain probability mass after one step, we need to wait
in $q_1$ (by playing $b$) until the probability mass in $q_2$ comes back to $q_2$ 
through the cycle. It is easy to show that to ensure strongly synchronizing,
we need to play $n-1$ times $b$ in $q_1$ before playing $a$, and this requires linear memory.
\qed
\end{proof}

\subsection{Strongly synchronizing with function $\fsum$}\label{sec:strongly-sum}

The strongly synchronizing objective with function $\fsum_{T}$ requires
that eventually all the probability mass remains in $T$. We show that this
is equivalent to a traditional reachability objective with target defined
by the set $S$ of sure winning initial distributions for the safety objective $\Box T$.

It follows that almost-sure (and limit-sure) winning for strongly synchronizing
is equivalent to almost-sure (or equivalently limit-sure) winning for
the coB\"uchi objective $\Diamond \Box T = \{q_0 a_0 q_1 \dots \in \Paths(\M) 
\mid \exists j \cdot \forall i > j: q_i \in T\}$ in the state-based semantics. 
However, sure strongly synchronizing
is not equivalent to sure winning for the coB\"uchi objective, as shown by
the MDP in \figurename~\ref{fig:coBuchi} which is:
\begin{itemize}
\item sure winning for the coB\"uchi objective $\Diamond \Box \{q_{\init},q_2\}$
from $q_{\init}$ (because in \emph{all} possible infinite paths from $q_{\init}$, 
there is a point from which only states in $\{q_{\init},q_2\}$ are visited), but 
\item not sure winning for the reachability objective $\Diamond S$
where $S = \{q_2\}$ is the winning region for the safety objective $\Box \{q_{\init},q_2\}$, 
thus not sure strongly synchronizing (the probability mass assigned to $q_1$ is always positive
after the first step). 
\end{itemize}

Note that this MDP is almost-sure
strongly synchronizing in target $T = \{q_{\init},q_2\}$ from $q_{\init}$, and almost-sure
winning for the coB\"uchi objective $\Diamond \Box T$, as well as almost-sure
winning for the reachability objective $\Diamond S$.

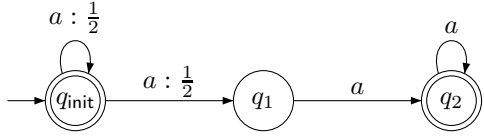
\begin{figure}[t]
\begin{center}
    \begin{picture}(62,17)(0,0)

\node[Nmarks=ir](q0)(8,4){$q_{\init}$}
\node[Nmarks=n](q1)(33,4){$q_1$}
\node[Nmarks=r](q2)(58,4){$q_2$}

\drawloop[ELside=l, loopangle=90, loopdiam=4](q0){$a:\frac{1}{2}$}
\drawedge[ELside=l, ELdist=.5](q0,q1){$a:\frac{1}{2}$}
\drawedge[ELside=l, ELdist=1](q1,q2){$a$}
\drawloop[ELside=l, loopangle=90, loopdiam=4](q2){$a$}

\end{picture}
\end{center}
 \caption{An MDP such that $q_{\init}$ is sure-winning for 
coB\"uchi objective in $T=\{q_{\init},q_2\}$ but not for 
strongly synchronizing according to $\fsum_T$. 
\label{fig:coBuchi}}
\end{figure}

\begin{lemma}\label{lem: strong-sum}
Given a target set~$T$, an MDP $\M$ is sure (resp., almost-sure or limit-sure) winning 
for the strongly synchronizing objective according to $\fsum_T$
if and only if $\M$ is sure (resp., almost-sure or limit-sure) winning 
for the reachability objective~$\Diamond S$  
where $S$ is the sure winning region for the safety objective~$\Box T$.
\end{lemma}

\begin{proof}
First, assume that a state $q_{\init}$ of $\M$ is sure (resp., almost-sure 
or limit-sure) winning for the strongly synchronizing objective according to $\fsum_T$,
and show that $q_{\init}$ is sure (resp., almost-sure or limit-sure) 
winning for the reachability objective~$\Diamond S$.

$(i)$ \emph{Limit-sure winning}. For all $\epsilon >0$, 
let $\epsilon' = \frac{\epsilon}{\abs{Q}} \cdot \eta^{\abs{Q}}$ where
$\eta$ is the smallest positive probability in the transitions of $\M$.
By the assumption, from $q_{\init}$ there exists a strategy~$\alpha$ and $N \in \nat$
such that for all $n \geq N$, we have $\M^{\alpha}_n(T) \geq 1-\epsilon'$.
We claim that at step $N$, all non-safe states have probability at most $\frac{\epsilon}{\abs{Q}}$,
that is $\M^{\alpha}_N(q) \leq \frac{\epsilon}{\abs{Q}}$ for all $q \in Q \setminus S$.
Towards contradiction, assume that $\M^{\alpha}_N(q) > \frac{\epsilon}{\abs{Q}}$
for some non-safe state $q \in Q \setminus S$. Since $q \not\in S$ is not safe, 
there is a path of length $\ell \leq \abs{Q}$ from $q$ to a state in $Q \setminus T$,
thus with probability at least $\eta^{\abs{Q}}$. It follows that after $N + \ell$ steps we have
$\M^{\alpha}_{N+\ell}(Q \setminus T) > \frac{\epsilon}{\abs{Q}} \cdot \eta^{\abs{Q}} = \epsilon'$,
in contradiction with the fact $\M^{\alpha}_n(T) \geq 1-\epsilon'$ for all $n \geq N$.
Now, since all non-safe states have probability at most $\frac{\epsilon}{\abs{Q}}$
at step $N$, it follows that $\M^{\alpha}_N(Q \setminus S) \leq \frac{\epsilon}{\abs{Q}}\cdot \abs{Q} = \epsilon$
and thus $\Pr^{\alpha}(\Diamond S) \geq 1-\epsilon$. Therefore, $\M$ is limit-sure
winning for the reachability objective $\Diamond S$ from $q_{\init}$.

$(ii)$ \emph{Almost-sure winning}. Since almost-sure strongly synchronizing
implies limit-sure strongly synchronizing, it follows from $(i)$ that
$\M$ is limit-sure (and thus also almost-sure) winning for the reachability objective 
$\Diamond S$, as limit-sure and almost-sure reachability coincide for MDPs~\cite{AHK07}.

$(iii)$ \emph{Sure winning}. From $q_{\init}$ there exists a strategy~$\alpha$ and $N \in \nat$ such 
that for all $n \geq N$, we have $\M^{\alpha}_n(T)=1$. Hence, $\alpha$ is sure
winning for the reachability objective $\Diamond \Supp(\M^{\alpha}_N)$,
and from all states in $\Supp(\M^{\alpha}_N)$ the strategy $\alpha$ ensures
that only states in $T$ are visited. It follows that $\Supp(\M^{\alpha}_N) \subseteq S$
is sure winning for the safety objective $\Box T$, and thus $\alpha$ is sure winning
for the reachability objective $\Diamond S$ from $q_{\init}$.

For the converse direction of the lemma, assume that a state $q_{\init}$ is sure (resp., almost-sure 
or limit-sure) winning for the reachability objective~$\Diamond S$.
We construct a winning strategy for strongly synchronizing in $T$ as follows: 
play according to a sure (resp., almost-sure or limit-sure) winning
strategy for the reachability objective~$\Diamond S$, and whenever a state
of $S$ is reached, then switch to a winning
strategy for the safety objective~$\Box T$. The constructed 
strategy is sure (resp., almost-sure or limit-sure) winning
for strongly synchronizing according to $\fsum_T$ because for
sure winning, after finitely many steps all paths from $q_{\init}$
end up in $S \subseteq T$ and stay in $S$ forever, and for almost-sure
(or equivalently limit-sure) winning, for all $\epsilon > 0$, after sufficiently 
many steps, the set $S$ is reached with probability at least $1 - \epsilon$,
showing that the outcome is strongly ($1 - \epsilon$)-synchronizing in $S \subseteq T$,
thus the strategy is almost-sure (and also limit-sure) strongly synchronizing.
\qed
\end{proof}

\begin{corollary}\label{col: almost-limit-strong-sum}
$\winlim{strongly}(\fsum_T) = \winas{strongly}(\fsum_T)$
for all target sets~$T$.
\end{corollary}

The following result follows from Lemma~\ref{lem: strong-sum} and the fact 
that the winning region for sure safety, sure reachability, and almost-sure 
reachability can be computed in polynomial time for MDPs~\cite{AHK07}. 
Moreover, memoryless strategies are sufficient for these objectives.

\begin{theorem}\label{theo:strongly-sum}
For the three winning modes of strongly synchronizing 
according to $\fsum_T$ in MDPs:

\begin{enumerate}
\item (Complexity). The membership problem is PTIME-complete.

\item (Memory). Pure memoryless strategies are sufficient.
\end{enumerate}
\end{theorem}

\section{Conclusion}\label{sec:conclusion}

We studied synchronizing properties for Markov decision processes and
presented comprehensive expressiveness and decidability results, identifying
the expressively equivalent winning modes (Lemma~\ref{lem:always}, Theorem~\ref{theo:weakly-ls-is-as}, Corollary~\ref{col: almost-limit-strong-sum}),
and showing, in all winning modes, PSPACE-completeness for eventually and weakly synchronizing, 
and PTIME-completeness for always and strongly synchronizing (Table~\ref{tab:complexity}). 
We showed that pure strategies are sufficient for all synchronizing objectives
and winning modes, and the memory requirements are given in Table~\ref{tab:memory}.

The $p$-synchronizing objectives we considered are qualitative in the sense that 
they are defined for $p = 1$ (sure-winning) or for $p \to 1$ (almost-sure and limit-sure winning).
A natural generalization is to consider the same objectives with $p < 1$.
However, the quantitative problem, which is to decide, 
given a rational number $p < 1$ whether an MDP is eventually $p$-synchronizing (in a given target state)
is at least as hard as the Skolem problem (which is to decide whether a linear recurrence 
sequence over the integers has a zero) whose decidability is a long-standing open question~\cite{OW14}.
The proof is by a reduction that can even be carried out for the special case of Markov chains~\cite[Theorem~3]{AAOW15}.
A variant of the problem where it is asked whether there exists $p' > p$ such that the given 
MDP is eventually $p'$-synchronizing is also Skolem-hard~\cite[Corollary~4]{AAOW15}.
An interesting direction for future research is to consider approximation problems
such as deciding, given $p$ and $\epsilon > 0$, whether an MDP is eventually $p'$-synchronizing 
for some $p' \in [p-\epsilon,p+\epsilon]$. 
In another direction, the qualitative problem can be generalized to multiple 
synchronizing objectives (e.g., conjunctions of objectives, in the flavor of
limit-sure winning with exact support), and to Boolean combinations
of synchronizing objectives, which is completely open. 

As we mention in the paragraph on related work (Section~\ref{sec:intro}),
synchronizing properties have been considered in several other models
of computation, such as weighted automata, register automata, timed systems, 
and partial-observation systems. An intriguing question is to consider
two-player stochastic games and to determine if (or which) synchronizing objectives are
decidable. In two-player stochastic games, some states are controlled by an adversary
and the synchronizing objectives need to be achieved no matter the choice
of the adversary at their state. The presence of an adversary makes 
the problem significantly more challenging, as it incurs an alternation of quantifiers
over the strategies. 

Finally, given the previous works, it is also interesting to extend the
results of this article to continuous-time Markov decision processes, 
pushdown Markov decision processes, and the 
subclass of one-counter Markov decision processes.

\paragraph{{\bf Acknowledgment}}
We are grateful to Winfried Just and German A. Enciso for helpful discussions
on Boolean networks and for the gadget in the proof of Lemma~\ref{lem:universal-finiteness-pspace-hard}.

\bibliographystyle{elsarticle-harv}  

\bibliography{biblio} 

\begin{thebibliography}{63}
\expandafter\ifx\csname natexlab\endcsname\relax\def\natexlab#1{#1}\fi
\expandafter\ifx\csname url\endcsname\relax
  \def\url#1{\texttt{#1}}\fi
\expandafter\ifx\csname urlprefix\endcsname\relax\def\urlprefix{URL }\fi

\bibitem[{Agrawal et~al.(2012)Agrawal, Akshay, Genest, and
  Thiagarajan}]{AAGT12}
Agrawal, M., Akshay, S., Genest, B., Thiagarajan, P.~S., 2012. Approximate
  verification of the symbolic dynamics of {M}arkov chains. In: Proc. of LICS:
  Logic in Computer Science. IEEE, pp. 55--64.

\bibitem[{Akshay et~al.(2015)Akshay, Antonopoulos, Ouaknine, and
  Worrell}]{AAOW15}
Akshay, S., Antonopoulos, T., Ouaknine, J., Worrell, J., 2015. Reachability
  problems for {M}arkov chains. Inf. Process. Lett. 115~(2), 155--158.

\bibitem[{Aspnes and Herlihy(1990)}]{AspnesH90}
Aspnes, J., Herlihy, M., 1990. Fast randomized consensus using shared memory.
  J.~Algorithm 11~(3), 441--461.

\bibitem[{Babari et~al.(2016)Babari, Quaas, and Shirmohammadi}]{BQS16}
Babari, P., Quaas, K., Shirmohammadi, M., 2016. Synchronizing data words for
  register automata. In: Proc. of {MFCS}: Mathematical Foundations of Computer
  Science. Vol.~58 of LIPIcs. Schloss Dagstuhl - Leibniz-Zentrum fuer
  Informatik, pp. 15:1--15:15.

\bibitem[{Bach and Shallit(1996)}]{BS96}
Bach, E., Shallit, J., 1996. Algorithmic Number Theory, Vol. 1: Efficient
  Algorithms. MIT Press.

\bibitem[{Baier et~al.(2006)Baier, Bertrand, and Schnoebelen}]{BBS06}
Baier, C., Bertrand, N., Schnoebelen, P., 2006. On computing fixpoints in
  well-structured regular model checking, with applications to lossy channel
  systems. In: Proc. of LPAR: Logic for Programming, Artificial Intelligence,
  and Reasoning. LNCS 4246. Springer, pp. 347--361.

\bibitem[{Baier et~al.(2012)Baier, Gr{\"{o}}{\ss}er, and Bertrand}]{BGB12}
Baier, C., Gr{\"{o}}{\ss}er, M., Bertrand, N., 2012. Probabilistic
  {\(\omega\)}-automata. J. {ACM} 59~(1), 1--52.

\bibitem[{Baldoni et~al.(2008)Baldoni, Bonnet, Milani, and Raynal}]{BBMR08}
Baldoni, R., Bonnet, F., Milani, A., Raynal, M., 2008. On the solvability of
  anonymous partial grids exploration by mobile robots. In: Proc. of OPODIS:
  Principles of Distributed Systems. LNCS 5401. Springer, pp. 428--445.

\bibitem[{Beauquier et~al.(2002)Beauquier, Rabinovich, and Slissenko}]{BRS02}
Beauquier, D., Rabinovich, A.~M., Slissenko, A., 2002. A logic of probability
  with decidable model-checking. In: Proc. of CSL: Computer Science Logic. LNCS
  2471. Springer, pp. 306--321.

\bibitem[{Berlinkov(2016)}]{Ber16}
Berlinkov, M.~V., 2016. On the probability of being synchronizable. In: Proc.
  of {CALDAM}: Algorithms and Discrete Applied Mathematics. LNCS 9602.
  Springer, pp. 73--84.

\bibitem[{Bertrand et~al.(2017)Bertrand, Dewaskar, Genest, and
  Gimbert}]{BDGG17}
Bertrand, N., Dewaskar, M., Genest, B., Gimbert, H., 2017. Controlling a
  population. In: Proc. of {CONCUR}: Concurrency Theory. Vol.~85 of LIPIcs.
  Schloss Dagstuhl - Leibniz-Zentrum fuer Informatik, pp. 12:1--12:16.

\bibitem[{Bianco and de~Alfaro(1995)}]{BA95}
Bianco, A., de~Alfaro, L., 1995. Model checking of probabalistic and
  nondeterministic systems. In: Proc. of FSTTCS: Foundations of Software
  Technology and Theoretical Computer Science. LNCS 1026. Springer, pp.
  499--513.

\bibitem[{Burkhard(1976)}]{Burkhard76a}
Burkhard, H.-D., 1976. Zum {L}{\"a}ngenproblem homogener {E}xperimente an
  determinierten und nicht-deterministischen {A}utomaten. Elektronische
  Informationsverarbeitung und Kybernetik 12~(6), 301--306.

\bibitem[{Cern\'{y}(1964)}]{Cer64}
Cern\'{y}, J., 1964. Pozn\'{a}mka k. homog\'{e}nnym experimentom s konecnymi
  automatmi. In: Matematicko-fyzik\'{a}lny \v{c}asopis. Vol. 14(3). pp.
  208--216.

\bibitem[{Chadha et~al.(2011)Chadha, Korthikanti, Viswanathan, Agha, and
  Kwon}]{CKVAK11}
Chadha, R., Korthikanti, V.~A., Viswanathan, M., Agha, G., Kwon, Y., 2011.
  Model checking {MDP}s with a unique compact invariant set of distributions.
  In: Proc. of QEST: Quantitative Evaluation of Systems. IEEE Computer Society,
  pp. 121--130.

\bibitem[{Chandra et~al.(1981)Chandra, Kozen, and Stockmeyer}]{CKS81}
Chandra, A.~K., Kozen, D., Stockmeyer, L.~J., 1981. Alternation. J. ACM 28~(1),
  114--133.

\bibitem[{Chatterjee and Doyen(2016)}]{CD16}
Chatterjee, K., Doyen, L., 2016. Computation tree logic for synchronization
  properties. In: Proc. of {ICALP}: Automata, Languages, and Programming.
  Vol.~55 of LIPIcs. Schloss Dagstuhl - Leibniz-Zentrum fuer Informatik, pp.
  98:1--98:14.

\bibitem[{Chatterjee et~al.(2011)Chatterjee, Henzinger, Joglekar, and
  Shah}]{CHJS11}
Chatterjee, K., Henzinger, M., Joglekar, M., Shah, N., 2011. Symbolic
  algorithms for qualitative analysis of {M}arkov decision processes with
  {B}{\"{u}}chi objectives. In: Proc. of CAV: Computer Aided Verification. LNCS
  6806. Springer, pp. 260--276.

\bibitem[{Chatterjee and Henzinger(2012)}]{CH12}
Chatterjee, K., Henzinger, T.~A., 2012. A survey of stochastic $\omega$-regular
  games. J. Comput. Syst. Sci. 78~(2), 394--413.

\bibitem[{Chistikov et~al.(2016)Chistikov, Martyugin, and
  Shirmohammadi}]{CMS16}
Chistikov, D., Martyugin, P., Shirmohammadi, M., 2016. Synchronizing automata
  over nested words. In: Proc. of {FoSSaCS}: Foundations of Software Science
  and Computation Structures. LNCS 9634. Springer, pp. 252--268.

\bibitem[{Courcoubetis and Yannakakis(1995)}]{CY95}
Courcoubetis, C., Yannakakis, M., 1995. The complexity of probabilistic
  verification. J. ACM 42~(4), 857--907.

\bibitem[{de~Alfaro(1997)}]{deAlfaro97}
de~Alfaro, L., 1997. Formal verification of probabilistic systems. Ph.D.
  thesis, Stanford University.

\bibitem[{de~Alfaro and Henzinger(2000)}]{ConcOmRegGames}
de~Alfaro, L., Henzinger, T.~A., 2000. Concurrent omega-regular games. In:
  Proc. of LICS: Logic in Computer Science. IEEE, pp. 141--154.

\bibitem[{de~Alfaro et~al.(2007)de~Alfaro, Henzinger, and Kupferman}]{AHK07}
de~Alfaro, L., Henzinger, T.~A., Kupferman, O., 2007. Concurrent reachability
  games. Theor. Comput. Sci. 386~(3), 188--217.

\bibitem[{Doyen et~al.(2014)Doyen, Juhl, Larsen, Markey, and
  Shirmohammadi}]{DJLMS14}
Doyen, L., Juhl, L., Larsen, K.~G., Markey, N., Shirmohammadi, M., 2014.
  Synchronizing words for weighted and timed automata. In: Proc. of {FSTTCS}:
  Foundation of Software Technology and Theoretical Computer Science. Vol.~29
  of LIPIcs. Schloss Dagstuhl - Leibniz-Zentrum fuer Informatik, pp. 121--132.

\bibitem[{Doyen et~al.(2011{\natexlab{a}})Doyen, Massart, and
  Shirmohammadi}]{DMS11b}
Doyen, L., Massart, T., Shirmohammadi, M., 2011{\natexlab{a}}. Infinite
  synchronizing words for probabilistic automata. In: Proc. of MFCS:
  Mathematical Foundations of Computer Science. LNCS 6907. Springer, pp.
  278--289.

\bibitem[{Doyen et~al.(2011{\natexlab{b}})Doyen, Massart, and
  Shirmohammadi}]{DMS11a}
Doyen, L., Massart, T., Shirmohammadi, M., 2011{\natexlab{b}}. Synchronizing
  objectives for {M}arkov decision processes. In: Proc. of iWIGP: Interactions,
  Games and Protocols. EPTCS 50. pp. 61--75.

\bibitem[{Doyen et~al.(2012)Doyen, Massart, and Shirmohammadi}]{DMS11Err}
Doyen, L., Massart, T., Shirmohammadi, M., 2012. Infinite synchronizing words
  for probabilistic automata ({E}rratum). CoRR abs/1206.0995.

\bibitem[{Etessami et~al.(2008)Etessami, Kwiatkowska, Vardi, and
  Yannakakis}]{EKVY08}
Etessami, K., Kwiatkowska, M.~Z., Vardi, M.~Y., Yannakakis, M., 2008.
  Multi-objective model checking of {M}arkov decision processes. Logical
  Methods in Computer Science 4~(4).

\bibitem[{Fijalkow et~al.(2016)Fijalkow, Kiefer, and Shirmohammadi}]{FKS16}
Fijalkow, N., Kiefer, S., Shirmohammadi, M., 2016. Trace refinement in labelled
  markov decision processes. In: Proc. of {FoSSaCS}: Foundations of Software
  Science and Computation Structures. LNCS 9634. Springer, pp. 303--318.

\bibitem[{Filar and Vrieze(1997)}]{FV97}
Filar, J., Vrieze, K., 1997. Competitive {Markov} Decision Processes. Springer.

\bibitem[{Fokkink and Pang(2006)}]{FokkinkP06}
Fokkink, W., Pang, J., 2006. Variations on {I}tai-{R}odeh leader election for
  anonymous rings and their analysis in {PRISM}. Journal of Universal Computer
  Science 12~(8), 981--1006.

\bibitem[{Forejt et~al.(2011)Forejt, Kwiatkowska, Norman, and Parker}]{FKNP11}
Forejt, V., Kwiatkowska, M., Norman, G., Parker, D., 2011. Automated
  verification techniques for probabilistic systems. In: Proc. of SFM: Formal
  Methods for Eternal Networked Software Systems. LNCS 6659. Springer, pp.
  53--113.

\bibitem[{Futcher(1999)}]{Futcher99}
Futcher, B., 1999. Cell cycle synchronization. Methods in Cell Science 21~(2),
  79--86.

\bibitem[{Gast et~al.(2012)Gast, Gaujal, and {Le Boudec}}]{GGB12}
Gast, N., Gaujal, B., {Le Boudec}, J.-Y., 2012. Mean field for {M}arkov
  decision processes: From discrete to continuous optimization. {IEEE} Trans.
  Automat. Contr. 57~(9), 2266--2280.

\bibitem[{Gimbert and Oualhadj(2010)}]{GO10}
Gimbert, H., Oualhadj, Y., 2010. Probabilistic automata on finite words:
  Decidable and undecidable problems. In: Proc. of ICALP (2): Automata,
  Languages and Programming. LNCS 6199. Springer, pp. 527--538.

\bibitem[{Goldschlager(1977)}]{G77}
Goldschlager, L.~M., 1977. The monotone and planar circuit value problems are
  log space complete for {P}. SIGACT News 9~(2), 25--29.

\bibitem[{Gr{\"a}del et~al.(2002)Gr{\"a}del, Thomas, and Wilke}]{automata}
Gr{\"a}del, E., Thomas, W., Wilke, T. (Eds.), 2002. Automata, Logics, and
  Infinite Games: A Guide to Current Research. LNCS 2500. Springer.

\bibitem[{Henzinger et~al.(2009)Henzinger, Mateescu, and Wolf}]{HMW09}
Henzinger, T.~A., Mateescu, M., Wolf, V., 2009. Sliding window abstraction for
  infinite {M}arkov chains. In: Proc. of CAV. LNCS 5643. Springer, pp.
  337--352.

\bibitem[{Hermanns et~al.(2014)Hermanns, Krc{\'{a}}l, and
  Kret{\'{\i}}nsk{\'{y}}}]{HKK14}
Hermanns, H., Krc{\'{a}}l, J., Kret{\'{\i}}nsk{\'{y}}, J., 2014. Probabilistic
  bisimulation: Naturally on distributions. In: Proc. of CONCUR: Concurrency
  Theory. LNCS 8704. Springer, pp. 249--265.

\bibitem[{Holzer(1995)}]{Holzer95}
Holzer, M., 1995. On emptiness and counting for alternating finite automata.
  In: Developments in Language Theory. pp. 88--97.

\bibitem[{Imreh and Steinby(1999)}]{IS99}
Imreh, B., Steinby, M., 1999. Directable nondeterministic automata. Acta
  Cybern. 14~(1), 105--115.

\bibitem[{Iv{\'{a}}n(2014)}]{Iva14}
Iv{\'{a}}n, S., 2014. Synchronizing weighted automata. In: Proc. of {AFL}:
  Automata and Formal Languages. Vol. 151 of {EPTCS}. pp. 301--313.

\bibitem[{Jancar and Sawa(2007)}]{AFA1}
Jancar, P., Sawa, Z., 2007. A note on emptiness for alternating finite automata
  with a one-letter alphabet. Inf. Process. Lett. 104~(5), 164--167.

\bibitem[{Kattenbelt et~al.(2010)Kattenbelt, Kwiatkowska, Norman, and
  Parker}]{KKNP10}
Kattenbelt, M., Kwiatkowska, M.~Z., Norman, G., Parker, D., 2010. A game-based
  abstraction-refinement framework for {M}arkov decision processes. Formal
  Methods in System Design 36~(3), 246--280.

\bibitem[{Kfoury(1970)}]{Kfo70}
Kfoury, D.~J., 1970. Synchronizing sequences for probabilistic automata.
  Studies in Applied Mathematics 29, 101--103.

\bibitem[{Korthikanti et~al.(2010)Korthikanti, Viswanathan, Agha, and
  Kwon}]{KVAK10}
Korthikanti, V.~A., Viswanathan, M., Agha, G., Kwon, Y., 2010. Reasoning about
  {MDP}s as transformers of probability distributions. In: Proc. of QEST:
  Quantitative Evaluation of Systems. IEEE Computer Society, pp. 199--208.

\bibitem[{Kret{\'{\i}}nsk{\'{y}} et~al.(2015)Kret{\'{\i}}nsk{\'{y}}, Larsen,
  Laursen, and Srba}]{KLLS15}
Kret{\'{\i}}nsk{\'{y}}, J., Larsen, K.~G., Laursen, S., Srba, J., 2015.
  Polynomial time decidability of weighted synchronization under partial
  observability. In: Proc. of {CONCUR}: Concurrency Theory. Vol.~42 of LIPIcs.
  Schloss Dagstuhl - Leibniz-Zentrum fuer Informatik, pp. 142--154.

\bibitem[{Larsen et~al.(2014)Larsen, Laursen, and Srba}]{LLS14}
Larsen, K.~G., Laursen, S., Srba, J., 2014. Synchronizing strategies under
  partial observability. In: Proc. of {CONCUR}: Concurrency Theory. LNCS 8704.
  Springer, pp. 188--202.

\bibitem[{Madani et~al.(2003)Madani, Hanks, and Condon}]{MHC03}
Madani, O., Hanks, S., Condon, A., 2003. On the undecidability of probabilistic
  planning and related stochastic optimization problems. Artif. Intell.
  147~(1-2), 5--34.

\bibitem[{Martyugin(2014)}]{Martyugin14}
Martyugin, P., 2014. Computational complexity of certain problems related to
  carefully synchronizing words for partial automata and directing words for
  nondeterministic automata. Theory Comput. Syst. 54~(2), 293--304.

\bibitem[{Mukovskiy et~al.(2010)Mukovskiy, Slotine, and Giese}]{MSG10}
Mukovskiy, A., Slotine, J.~J., Giese, M.~A., 2010. Design of the dynamic
  stability properties of the collective behavior of articulated bipeds. In:
  Proc. of Humanoids: Conference on Humanoid Robots. {IEEE}, pp. 66--73.

\bibitem[{Nicaud(2016)}]{Nic16}
Nicaud, C., 2016. Fast synchronization of random automata. In: Proc. of
  {RANDOM}: Workshop on Randomization and Computation. Vol.~60 of LIPIcs.
  Schloss Dagstuhl - Leibniz-Zentrum fuer Informatik, pp. 43:1--43:12.

\bibitem[{Ouaknine and Worrell(2014)}]{OW14}
Ouaknine, J., Worrell, J., 2014. Positivity problems for low-order linear
  recurrence sequences. In: Proc. of SODA: Symposium on Discrete Algorithms.
  {SIAM}, pp. 366--379.

\bibitem[{Paz(1971)}]{PAZBook}
Paz, A., 1971. Introduction to probabilistic automata. Academic Press, Inc.
  Orlando, FL, USA.

\bibitem[{Pogosyants et~al.(2000)Pogosyants, Segala, and Lynch}]{PSL00}
Pogosyants, A., Segala, R., Lynch, N.~A., 2000. Verification of the randomized
  consensus algorithm of {A}spnes and {H}erlihy: a case study. Distributed
  Computing 13~(3), 155--186.

\bibitem[{Puterman(1994)}]{Puterman}
Puterman, M.~L., 1994. {M}arkov Decision Processes. John Wiley and Sons.

\bibitem[{Rabin(1963)}]{Rabin63}
Rabin, M.~O., 1963. Probabilistic automata. Information and Control 6,
  230--245.

\bibitem[{Shirmohammadi(2014)}]{Shi14}
Shirmohammadi, M., 2014. Qualitative analysis of probabilistic synchronizing
  systems. Ph.D. thesis, Universit\'e Libre de Bruxelles.

\bibitem[{Szykula(2018)}]{Szy18}
Szykula, M., 2018. Improving the upper bound on the length of the shortest
  reset word. In: Proc. of {STACS}: Symposium on Theoretical Aspects of
  Computer Science. LIPIcs. Schloss Dagstuhl - Leibniz-Zentrum fuer Informatik,
  pp. 56:1--56:13.

\bibitem[{Vardi(1985)}]{Vardi-focs85}
Vardi, M.~Y., 1985. Automatic verification of probabilistic concurrent
  finite-state programs. In: Proc. of FOCS: Foundations of Computer Science.
  IEEE Computer Society, pp. 327--338.

\bibitem[{Vardi(2007)}]{Vardi07}
Vardi, M.~Y., 2007. Automata-theoretic model checking revisited. In: Proc. of
  VMCAI: Verification, Model Checking, and Abstract Interpretation. LNCS 4349.
  Springer, pp. 137--150.

\bibitem[{Volkov(2008)}]{Volkov08}
Volkov, M.~V., 2008. Synchronizing automata and the {C}erny conjecture. In:
  Proc. of LATA: Language and Automata Theory and Applications. LNCS 5196.
  Springer, pp. 11--27.

\end{thebibliography}
\end{document}